\pdfoutput=1
\documentclass[acmsmall]{acmart}
\usepackage[utf8]{inputenc}
\usepackage{amsfonts}
\usepackage{braket,microtype,amsmath,mathtools,graphicx,enumitem,booktabs,bm,xspace,float,mathdots,mleftright}
\usepackage{amsthm}
\hypersetup{hypertexnames=false}

\usepackage[T1]{fontenc}
\usepackage{textcomp}
\usepackage[english]{babel}
\usepackage[capitalize,nameinlink]{cleveref}
\AddToHook{cmd/appendix/before}{\crefalias{section}{appendix}}
\usepackage[breakable]{tcolorbox}
\usepackage{crossreftools}
\usepackage{thm-restate}

\AtEndPreamble{%
\theoremstyle{acmdefinition}
\newtheorem{remark}[theorem]{Remark}}
\allowdisplaybreaks[4]

\numberwithin{equation}{section}

\DeclareMathOperator{\tr}{tr}
\DeclarePairedDelimiter\abs{\lvert}{\rvert}

\DeclarePairedDelimiter\ceil{\lceil}{\rceil}
\DeclarePairedDelimiter\parens{\lparen}{\rparen}
\DeclarePairedDelimiter\braces{\lbrace}{\rbrace}
\DeclarePairedDelimiter\bracks{\lbrack}{\rbrack}

\DeclarePairedDelimiter\Bracks{\lBrack}{\rBrack}

\newcommand{\ketbra}[2]{\ket{#1} \!\! \bra{#2}}
\newcommand{\proj}[1]{\ket{#1} \!\! \bra{#1}}
\newcommand{\ot}{\otimes}

\newcommand{\eps}{\varepsilon}

\renewcommand{\vec}[1]{\boldsymbol{#1}}
\newcommand{\q}{{\vec{q}}}

\newcommand{\cH}{\mathcal H}
\newcommand{\cE}{\mathcal E}
\newcommand{\cP}{\mathcal P}
\newcommand{\cM}{\mathcal M}

\newcommand{\cI}{\mathcal I}

\newcommand{\cS}{\mathcal S}
\newcommand{\cO}{\mathcal O}
\newcommand{\cB}{\mathcal B}

\newcommand{\R}{\mathbb R}
\newcommand{\C}{\mathbb C}

\newcommand{\N}{\mathbb N}

\newcommand{\irange}[2]{\{#1, \ldots, #2\}}
\newcommand{\cPHle}[1][\cH]{\cP_{\leq 1}(#1)}
\newcommand{\dH}[1][\cH]{\mathcal{D}(#1)}
\newcommand{\dHle}[1][\cH]{\mathcal{D}_{\leq 1}(#1)}
\newcommand{\qVars}{\vec{\mathsf{qVars}}}
\newcommand{\hoare}[3]{\braces*{#1} #2 \braces*{#3}}

\newcommand{\parcorr}{\models_\text{par}}
\newcommand{\totcorr}{\models_\text{tot}}
\newcommand{\satisfies}[2]{\mathbb{E}_{#1}\parens*{#2}}
\newcommand{\imp}{\Rightarrow}
\newcommand{\QFT}{\mathsf{QFT}}
\newcommand{\vecr}{\vec{r}}
\newcommand{\tilx}{\tilde{x}}
\newcommand{\tily}{\tilde{y}}

\newcommand{\ctprog}{\ensuremath{A_\text{toss-until-zero}}}
\newcommand{\Spec}{S}
\newcommand{\SWAP}{\mathsf{SWAP}}
\newcommand{\CRz}{\mathsf{CRz}}
\newcommand{\Rz}{\mathsf{Rz}}

\colorlet{qpbgcol}{gray!15}
\newtcbox{\qpouterbox}{on line,
  breakable,
  colback=qpbgcol,
  colframe=qpbgcol,
  size=fbox,
  arc=2pt,
  boxsep=1pt,
  left=1pt,right=1pt,top=0pt,bottom=0pt,
}
\newtcbox{\qpinnerbox}{on line,
  enforce breakable,
  colback=white,
  colframe=white,
  size=fbox,
  arc=1pt,
  boxsep=1pt,
  left=0pt,right=0pt,top=0pt,bottom=0pt
}

\newcommand{\qwhilesyntax}{\ensuremath{
S~::=~
    &\;\; \qpskip~
    |\; \qpinit{\vec q} ~
    |\; \qpunitary{\vec{q}}{U} ~
    |\; S_1 ; S_2~
    |\; \qprepeat{N}{S} ~\\
    &|\; \qpselect{\{\omega_1 \colon M_{\omega_1}, \omega_2 \colon M_{\omega_2}, \ldots\}}{\vec{q}}{
       \omega_1 \colon S_{\omega_1}, \ \omega_2 \colon S_{\omega_2}, \ \ldots}~\\
    &|\; \qpwhile{B}{\vec{q}}{S} ~
}}

\newcommand{\qpenv}[1]{\qpouterbox{$#1$}}
\newcommand{\qpabsprog}[1]{\qpinnerbox{$#1$}}

\newcommand{\qpskip}{\textbf{skip}}
\newcommand{\qpinit}[1]{#1 := \ket{\vec 0}}
\newcommand{\qpinitS}[1]{#1 := \ket{0}}
\newcommand{\qpunitary}[2]{#1 := #2(#1)}
\newcommand{\qprepeat}[2]{\textbf{repeat } #1 \textbf{ do } #2 \textbf{ end}}
\newcommand{\qpselect}[3]{\textbf{case meas } #2 \textbf{ with } #1 \textbf{ of } #3 \textbf{ end}}
\newcommand{\qpwhile}[4][]{\textbf{while}#1 \textbf{ meas } #3 \textbf{ with } #2 \textbf{ do } #4 \textbf{ end}}
\newcommand{\qphole}[3][]{\braces{#2}\square_{#1} \braces{#3}}

\newcommand{\qpifte}[4]{\textbf{if meas } #2 \textbf{ with } #1 \textbf{ then } #3 \textbf{ else } #4 \textbf{ end}}
\newcommand{\qpif}[3]{\textbf{if meas } #2 \textbf{ with } #1 \textbf{ then } #3 \textbf{ end}}
\newcommand{\qpifteB}[3]{\textbf{if meas } #1 \textbf{ then } #2 \textbf{ else } #3 \textbf{ end}}

\newcommand{\qpwhileB}[2]{\textbf{while meas } #1 \textbf{ do } #2 \textbf{ end}}
\newcommand{\qpselectS}[2]{\textbf{case meas } #1 \textbf{ of } #2 \textbf{ end}}

\newcommand{\refine}{\hookrightarrow}
\newcommand{\parref}{\refine_\mathrm{par}}
\newcommand{\totref}{\refine_\mathrm{tot}}
\newcommand{\ruleitem}[1]{\item[{(\crtcrossreflabel{#1}[rule:#1])}]} % usage \ruleitem{A.BC} => \item[(A.BC)]\label{rule:A.BC}
\newcommand{\ruleitemNOREF}[1]{\item[(#1)]}
\renewcommand{\ruleitemNOREF}[1]{\item[{(\crtcrossreflabel{#1}[rule:#1])}]} % usage \ruleitem{A.BC} => \item[(A.BC)]\label{rule:A.BC}
\newcommand{\ruleref}[1]{\hyperref[rule:#1]{(#1)}}
\newcommand{\ruletag}[1]{\tag{\ref{rule:#1}}} % usage \ruletag{H.init}

\newenvironment{proofcase}{}{} % \smallskip
\newcommand{\mymedskip}{} % \medskip

\usepackage[normalem]{ulem}

\AtBeginDocument{}

\setcopyright{cc}
\setcctype{by}
\acmDOI{10.1145/3720433}
\acmYear{2025}
\acmJournal{PACMPL}
\acmVolume{9}
\acmNumber{OOPSLA1}
\acmArticle{99}
\acmMonth{4}
\received{2024-10-15}
\received[accepted]{2025-02-18}

\begin{document}

%=============================================================================
\title{QbC: Quantum Correctness by Construction}
\newcommand{\theshortauthors}{Peduri, Schaefer, and Walter}

\author{Anurudh Peduri}
\email{anurudh.peduri@rub.de}
\orcid{0000-0002-6523-7098}
\affiliation{%
  \department{Chair for Quantum Information, Faculty of Computer Science}
  \institution{Ruhr~University~Bochum}
  \city{Bochum}
  \country{Germany}
}

\author{Ina Schaefer}
\orcid{0000-0002-7153-761X}
\email{ina.schaefer@kit.edu}
\affiliation{%
  \department{Chair of Testing, Validation and Analysis of Software-Intensive Systems (TVA), Institute for Information Security and Dependability (KASTEL)}
  \institution{Karlsruhe Institute of Technology}
  \city{Karlsruhe}
  \country{Germany}
}

\author{Michael Walter}
\orcid{0000-0002-3073-1408}
\email{michael.walter@rub.de}
\affiliation{%
  \department{Chair for Quantum Information, Faculty of Computer Science}
  \institution{Ruhr~University~Bochum}
  \city{Bochum}
  \country{Germany}
}
% \renewcommand{\shortauthors}{\theshortauthors}

%=============================================================================
\newcommand{\theabstract}{%
Thanks to the rapid progress and growing complexity of quantum algorithms, correctness of quantum programs has become a major concern.
Pioneering research over the past years has proposed various approaches to formally verify quantum programs using proof systems such as quantum Hoare logic.
All these prior approaches are post-hoc: one first implements a program and only then verifies its correctness.
Here we propose \emph{Quantum Correctness by Construction (QbC)}: an approach to constructing quantum programs from their specification in a way that ensures correctness.
We use pre- and postconditions to specify program properties,
and propose sound and complete refinement rules for constructing programs in a quantum while language from their specification.
We validate QbC by constructing quantum programs for idiomatic problems and patterns.
We find that the approach naturally suggests how to derive program details, highlighting key design choices along the way.
As such, we believe that QbC can play a role in supporting the design and taxonomization of quantum algorithms and software.}
\newcommand{\thekeywords}{quantum Hoare logic, correctness by construction, quantum while language}
%=============================================================================

\begin{abstract}
\theabstract
\end{abstract}
%% The code below is generated by the tool at http://dl.acm.org/ccs.cfm.
%% Please copy and paste the code instead of the example below.



%=============================================================================
\maketitle
%=============================================================================

%=============================================================================
\section{Introduction}
\label{sec:intro}
%=============================================================================
The field of quantum computing has seen tremendous progress.
There are a variety of quantum algorithms for a broad range of computational problems, including for combinatorial search and optimization~\cite{grover1996,brassard2002quantum,ambainis2004quantum,apers2022no}, factoring and other algebraic problems~\cite{shor1994}, and linear algebra~\cite{hhl2009,gilyenquantum2019,van2021quantum,Chakraborty2023quantumregularized} --
supported by algorithmic frameworks such as quantum walks~\cite{qwalks2003} and the quantum singular value transform~\cite{gilyenquantum2019,granduni2021},
as well as by novel quantum data structures~\cite{qram,tower2022}.
\Citet{montanaro_quantum_2016} provides an overview of various quantum algorithms.
To support these developments and the construction of larger quantum programs, \citet{selinger_2004} first proposed a design for a quantum programming language.
Since then, there have been numerous quantum programming languages at various levels of abstraction (see, e.g.,~\cite{qml2005,quipper2013,Steiger2018projectq,qsharp2018,silq2022,qunity2023} and references therein).

Correctness of algorithms has always been a major concern in computing,
with intensive work on program analysis, testing and verification of programs and software in past decades.
For classical computing, Hoare~\cite{hoare1969hoare} introduced a formal system,
in which for a program~$S$, one specifies its properties by using a precondition~$P$ and a postcondition~$Q$, resulting in a \emph{Hoare triple} denoted as~$\hoare{P}{S}{Q}$.
A Hoare triple is said to be correct if running the program~$S$ starting in any state satisfying $P$ results in a state that satisfies $Q$.
Hoare logic has been extended to probabilistic programs~\cite{Morgan1998pGCLFR,Morgan1998pGCLFR}, where the properties are probabilistic, and correctness is defined in terms of their expectations.
A survey of the successes of Hoare logic can be found in~\cite{apt_fifty_2019}.
Another approach to verifying programs is using Incorrectness Logic~\cite{incorrectnesslogic2019}, which attempts to find bugs in programs by finding counterexamples.
Quantum computing poses unique challenges for formal verification that are not encountered in classical computing due to the nature of its computational model.
Pioneering research over the past years has uncovered how to adapt the above-mentioned approaches to the quantum setting,
e.g., quantum Hoare logic~\cite{sanders2000quantum,chadha2006hoare,kakutani2009hoare,ying2012quantum} and quantum incorrectness logic~\cite{peng2022_quantum_incorrectness_logic}.

Still, all prior approaches to verifying quantum programs are \emph{post-hoc}:
they take the completed program as a starting point and establish whether the program meets the specification.
If post-hoc verification fails, there is often no indication of what needs to be fixed in the program.
Especially in the quantum setting, predicates specifying program properties are represented by large matrices, which makes it difficult to reason about, locate, and fix issues with quantum programs with a post-hoc approach.
In contrast, \emph{Correctness-by-Construction (CbC)}~\cite{10.5555/550359,morgan1988specification,kourie_correctness-by-construction_2012} is a programming methodology to incrementally build correct programs based on a specification.
For classical computing, CbC provides a method where one starts with a concise specification,
and then uses a small set of refinement rules to incrementally construct the program in such a way that at the end of the construction process, the program provably satisfies the specification, i.e., it is correct by construction.
This approach can help build well-structured and concise programs and draw insights from the corresponding specification, supporting algorithm designers in developing intuition and allowing them to focus on central design aspects of the constructed algorithms.
\Citet{runge2020comparecorc} provide a detailed comparison between CbC and post-hoc verification techniques.
To the best of our knowledge, the Correctness-by-Construction approach has not been applied to construct and verify quantum programs before our work.

%-----------------------------------------------------------------------------
\subsection{Our Contributions}
%-----------------------------------------------------------------------------
In this work, we propose \emph{Quantum Correctness by Construction (QbC)}, an approach to constructing quantum programs from their specification in a way that ensures correctness.
To this end, we consider a simple quantum programming language, the \emph{quantum while language}, and extend it with a new construct called a \emph{hole}, which represents a yet-to-be-constructed program.
Holes take the form $\qpenv{\qphole{P}{Q}}$, where~$P$ and~$Q$ represent the pre- and postconditions that should be satisfied by the program, as in quantum Hoare logic~\cite{ying2012quantum}.
For example, the specification for the paradigmatic problem of \emph{searching} a ``database''~$f\colon\{0,1\}^n\to\{0,1\}$ with success probability~$p$, as famously solved by Grover's algorithm~\cite{grover1996}, can be succinctly expressed as
\[ \qpenv{\qphole{pI}{\sum_{x ~\text{s.t.}~ f(x) = 1} \proj{x}}}. \]
This states that measuring the program's output results in a solution to the search problem, i.e., an~$x$ such that~$f(x) = 1$, with probability at least~$p$.
Indeed, the precondition $pI$ accepts any state with probability~$p$, and the postcondition, which then must hold with at least this probability, accepts only states that on measuring give an~$x$ such that~$f(x) = 1$.
We discuss this example
in~\cref{sec:examples}.

For the above language, we then provide \emph{refinement rules}, which allow filling in holes in quantum programs in such a way that correctness is preserved.
There are two widely used notions of correctness for Hoare triples, \emph{partial} and \emph{total} correctness, which differ in how they treat non-termination.
We provide refinement rules for both notions.
These rules take the following form:
\[ \qpenv{\qphole{P}{Q}} \refine S \quad\quad\quad \text{if certain conditions $C_1, C_2, \ldots$ hold} \]
To apply such a rule, one first checks that the conditions~$C_1, C_2, \dots$ are satisfied and then replaces the left-hand side hole with the right-hand side program~$S$ (which may itself contain other holes).
For example, we can always apply the \emph{sequence rule}, $\qpenv{\qphole{P}{Q}} \refine \qpenv{\qphole{P}{R} ; \qphole{R}{Q}}$;
intuitively, this rule states that in order to construct a program that takes precondition~$P$ to postcondition~$Q$, it suffices to construct a program that takes the precondition to some intermediate condition~$R$, and another one that takes this intermediate condition to the postcondition.
Starting from an initial specification~$\qpenv{\qphole{P}{Q}}$, one iteratively applies refinement rules until one arrives at a program $S_\text{final}$ that contains no more holes.
We prove a \emph{soundness theorem} that states that any program obtained by this process of refinement is guaranteed to satisfy the initial specification, that is, the quantum Hoare triple $\hoare{P}{S_\text{final}}{Q}$ is valid and the program is correct by construction.
Here, we use the notion of quantum Hoare triple introduced by \citet{ying2012quantum}.
We also prove \emph{completeness}: any program that satisfies a specification can be obtained by process of refinement.
We note that the completeness above is \emph{relative} to the theory of complex numbers,
as defined in Ref.~\cite{ying2012quantum}, which is also assumed in all prior works.
Because of the unintuitive nature of quantum logic, these results are more challenging to establish than in the classical case.
For example, quantum predicates cannot be interpreted as (deterministic or probabilistic) functions of some program state;
observing or ``measuring'' predicates is not a passive operation, but will in general change the program's state;
and there is in general no canonical choice of predicates in refinement rules.
All these can be traced back to the noncommutative nature of quantum information.

Finally, we validate our approach by constructing quantum programs for \emph{quantum teleportation} and \emph{quantum search}~(\cref{sec:examples}).
In each case, we start from their intuitive specification and use one key algorithmic idea or refinement step at a time.
We find that the refinement rules not only guide the construction of the desired programs, but that the QbC approach also reveals design choices that can be made along the way.
For example, when constructing a program for the quantum search problem, we show how one can naturally arrive at both a naive algorithm that proceeds by random sampling, as well as Grover's celebrated search algorithm that offers a quantum speedup.
To summarize, in this work, we:
\begin{enumerate}
\item Introduce a Quantum Correctness by Construction (QbC) approach for quantum programs.
In QbC, quantum programs are constructed starting from specifications by successively applying refinement rules.
\item In doing so, provide a formalization based on quantum while programs with \emph{holes}, which specify pre- and postconditions of subprograms that still need to be constructed.

\item Prove our refinement systems sound and complete: Any program constructed from an initial specification by using the refinement rules must satisfy the specification, and for any program satisfying a specification, there exists a sequence of refinements to obtain it from the specification.

\item Validate QbC by constructing quantum programs for idiomatic problems, starting from their specification.
Our findings suggest that QbC can play a role in supporting the design and taxonomization of quantum algorithms and software.
\end{enumerate}

%==============================================================================
\subsection{Related Work}\label{new:sec:related-work}
%==============================================================================

\paragraph{Correctness-by-Construction for Classical Programs}
The CbC approach to programming was initially introduced by \citet{10.5555/550359,morgan1988specification,back1998book} and \citet{gries_science_1981}
and later extended by \citet{kourie_correctness-by-construction_2012}, see also the relation algebra of programming~\cite{birddemoor1997algebra}.
There has since been extensive work on tools and theory for CbC~\cite{runge2019corc,bordis2022corc,runge2020latticeinfocbc,runge2023flexcbc,bordis2023corceco,knuppel_scaling_2020,traitcbc2022}.
Other similar approaches include the B method~\cite{1996bbook,2003bmethod,2010modellingEventB}, syntax-guided synthesis~\cite{solar2009sketching},
and type systems~\cite{Chlipala2013cpdt}.

\paragraph{Probabilistic Verification}
Various classical program logics have been extended to probabilistic programs:
Hoare logic~\cite{probhoare1999},
weakest precondition reasoning~\cite{10.1145/2933575.2935317,Morgan1998pGCLFR},
separation logic~\cite{10.1145/3290347},
relational logic~\cite{barthe2009PRHL,barthe2025PRHL},
and dynamic logic for forward reasoning~\cite{pardo2022pgcl}.
We refer to \citet{kaminski2016weakest} and references therein for recent work on proving termination and bounding expected runtimes. % of probabilistic programs.
Syntax-guided synthesis techniques have also been extended to probabilistic programs~\cite{andriushchenko2021paynt}.

\paragraph{Quantum Verification}
The majority of work in quantum verification has focused on post-hoc approaches.
The first approach on Hoare logic for quantum programs was introduced by \cite{kakutani2009hoare,chadha2006hoare},
and was extended  to support reasoning about unbounded quantum loops with (relative) completeness~\cite{ying2012quantum}
and to include classical variables~\cite{feng2021quantum}.
\Citet{coqq2023} provided a rigorous implementation of quantum Hoare logic in Coq with a range of applications.
\Citet{aQHL2019} focused on projections as predicates and proposed a notion of robust Hoare triples, which allow pre- and postconditions to be approximately satisfied.
\Citet{rand2019verification} surveys recent advances in Hoare-style verification logics for quantum programs.
\Citet{zuliani07formalgrover} illustrates a refinement based approach, but gives no complete system of refinement rules.
\Citet{neri2021compiling} contributes an extension of laws of classical program algebra to quantum programming.
Other related quantum verification efforts include incorrectness logic~\cite{peng2022_quantum_incorrectness_logic},
circuit verification~\cite{qwire2017,Rand_2019,lehmann2022vyzxvisionverifying},
quantum relational Hoare logic~\cite{Unruh_2019,li2021_qrhl_exp},
and equational reasoning in Dirac notation~\cite{xu2025Dirac}.

\paragraph{Recent Developments}
Shortly after our work had appeared as a preprint, \citet{feng2023refinement} reported on independent work on a refinement system that is similar to ours but differs in two key respects.
First, their work considers only projective predicates, whereas our formalism allows for arbitrary predicates.
Therefore our formalism captures a larger class of interesting properties, in particular success probabilities of algorithms, which are important in most quantum algorithms (see \cref{sec:examples} and \cref{subsec:examples:boosting}).
Second, while their refinement system only ensures partial correctness, we also provide a refinement system that ensures total correctness.
The latter gives stronger guarantees and in particular enables reasoning about termination, which is impossible otherwise.
\Citet{feng2023refinement} also provide a Python-based proof-of-concept implementation of their calculus.
A recent work~\cite{seng2024quantum} has similarly proposed a web-based proof-of-concept implementation of QbC.

%-----------------------------------------------------------------------------
\paragraph{Organization of the Paper.}
%-----------------------------------------------------------------------------
In \cref{sec:prelims}, we review the basic quantum formalism that is used in the paper.
We introduce the \emph{quantum while language}, a simple quantum programming language with control-flow and loops, as well as \emph{quantum Hoare logic}, which defines notions of correctness of quantum programs in terms of pre- and postconditions.
In \cref{sec:qbc}, we introduce \emph{Quantum Correctness by Construction (QbC)}.
We first define an extension of the quantum while language, called \emph{programs with holes}, which allows specifying subprograms that still need to be constructed.
Then we define refinement rules that can be used to construct quantum programs from given specifications, and we prove soundness and completeness of these rules.
In \cref{sec:examples}, we use QbC to naturally construct several quantum programs starting from their specification.
We conclude in \cref{sec:conclusion}.
\Cref{app:soundness,app:completeness,app:examples:boosting} contain technical proofs for results announced in the main text.
\Cref{app:coin-toss-till-zero} continues the discussion of a running example in the text.

%=============================================================================
\section{Preliminaries}
\label{sec:prelims}
%=============================================================================
In this section, after setting our notation and conventions (\cref{subsec:notation}), we give a brief introduction to the formalism of quantum computing~(\cref{subsec:quantum}).
Then we describe the syntax and semantics of a simple quantum programming language~(\cref{subsec:qwhile}) and recall quantum Hoare logic~(\cref{subsec:quantum-hoare-logic}).

%-----------------------------------------------------------------------------
\subsection{Notation and Conventions}
\label{subsec:notation}
%-----------------------------------------------------------------------------
% We fix some notation and conventions that are used throughout the paper.
% N
We take $\N$ to be the set of natural numbers including zero.
% Hilbert space
\pagebreak[3]
In this work, a Hilbert space~$\mathcal H$ is a finite-dimensional complex vector space with inner product.
% Kets and Bras
Throughout the paper we use \emph{Dirac notation}:
we write~$\ket\psi \in \mathcal H$ for vectors, $\bra\phi$ for covectors, and~$\braket{\phi|\psi}$ for the inner product.
Here, $\psi$ is an arbitrary label.
In general, $M^\dagger$ denotes the adjoint of a linear operator~$M$.
% I and Trace
The identity operator on a Hilbert space~$\cH$ is denoted by~$I_\cH$ and can be written as $I_\cH = \sum_{x\in\Sigma} \proj x$ for any (orthonormal) basis $\{ \ket x \}_{x \in \Sigma}$ of~$\cH$, where $\Sigma$ is an index set.
We write~$I$ when the Hilbert space is clear from the context.
The \emph{trace} of an operator~$M$ can be computed as~$\tr M = \sum_{x \in \Sigma} \braket{x|M|x}$ for any basis as above.
For example, a \emph{quantum bit} or \emph{qubit} corresponds to the $2$-dimensional Hilbert space~$\mathcal H =\C^2$,
with standard basis $\{\ket0, \ket1\}$ labeled by~$\Sigma=\{0,1\}$.
The identity operator is $I = \proj0 + \proj1$.
An example of an operator on~$\C^2$ is the Pauli $X$ matrix, defined as $X = \ketbra01 + \ketbra10$.
It satisfies~$X^\dagger = X$ and $\tr X = 0$.
We require two more concepts from linear algebra.
An operator~$M$ on~$\cH$ is \emph{Hermitian} if $\braket{\psi|M|\psi} \in \R$ for all vectors~$\ket\psi\in\cH$, and \emph{positive semi\-definite~(PSD)} if~$\braket{\psi|M|\psi} \ge 0$ for all vectors~$\ket\psi \in \cH$.
Equivalently,~$M$ can be diagonalized by a unitary matrix and has real resp.\ nonnegative eigenvalues.
Given two operators~$A$ and~$B$ on~$\cH$, we write $A \preceq B$ if and only if~$B - A$ is PSD; this defines the \emph{L\"owner order}.
For example, we can write $M \geq 0$ to state that $M$ is PSD.

% Unitary
% An operator~$U$ is called a \emph{unitary} if~$U U^\dagger = U^\dagger U = I$.
% PSD
% Any PSD operator~$M$ has a unique PSD \emph{square root}, denoted by~$\sqrt M$.
% Given two operators~$A$ and~$B$ on~$\cH$, we write $A \preceq B$ if and only if~$B - A$ is PSD;
% this defines the \emph{L\"owner order}.
% Projector
% A Hermitian operator~$P$ is called a \emph{projection} (or \emph{projector}) if~$P^2 = P$.
% Superops
% A linear function~$\cE$ mapping operators on one Hilbert space to operators on another is called a \emph{superoperator}.
% Trace
% It is called \emph{trace preserving} if~$\tr M = \tr \cE(M)$ for every operator~$M$,
% and \emph{trace non-increasing} if~$\tr M \ge \tr \cE(M)$ for every~$M$.
% CP
% A superoperator $\cE$ on $\cH$ is called \emph{positive} if~$\cE(M)$ is PSD for every PSD~$M$,
% and \emph{completely positive} if the superoperator~$\cE \ot \cI_{\cH'}$ is positive for every Hilbert space~$\cH'$ (where~$\cI_{\cH'}$ is the identity superoperator).
% We denote by~$\cE^\dagger$ the adjoint of a superoperator with respect to the Hilbert-Schmidt inner product, which satisfies the defining property that $\tr (A~\cE(B)) = \tr(\cE^\dagger(A)~B)$ for all operators~$A$,~$B$.
% The adjoint is completely positive iff~$\cE$ is completely positive.
% Moreover, it is easy to see that~$\cE$ is trace preserving iff its adjoint is \emph{unital}, that is,~$\cE^\dagger(I) = I$, and trace non-increasing if the adjoint is \emph{sub-unital}, that is, $\cE^\dagger(I) \preceq I$.

%-----------------------------------------------------------------------------
\subsection{Quantum Computing}
\label{subsec:quantum}
%-----------------------------------------------------------------------------
We now recall the basic formalism of quantum computing.
We refer to the excellent textbooks~\citet{nielsen2010quantum,wilde2013quantum,YM08} for more comprehensive introductions.

%-----------------------------------------------------------------------------
\paragraph{Variables}
A \emph{quantum variable}~$q$ is modeled by a Hilbert space~$\cH_q = \C^{\Sigma_q}$ for some finite index set~$\Sigma_q$.
This means that $\cH_q$ is a vector space equipped with an inner product and an orthonormal \emph{standard basis} (or \emph{computational basis})~$\{\ket x\}_{x \in \Sigma_q}$, labeled by the elements~$x\in \Sigma_q$.
When~$\Sigma_q = \{0,1\}$, then~$\cH_q = \C^2$ and~$q$ is called a \emph{quantum bit} or \emph{qubit}, with standard basis~$\{\ket0, \ket1\}$, as above.
If~$\vec{q}$ is a collection of quantum variables, then the corresponding Hilbert space is $\cH_{\vec q} = \bigotimes_{q \in \vec q} \cH_q \cong \C^{\Sigma_{\vec q}}$, where the Cartesian product~$\Sigma_{\vec q} = \prod_{q \in \vec q} \Sigma_q$ labels the standard (product) basis of the quantum variables~$\vec q$.
We assume that there is a finite set of quantum variables, denoted by $\qVars$.
Then the overall Hilbert space is
\begin{align*}
  \cH = \cH_{\qVars} = \bigotimes_{q\in\qVars} \cH_q \cong \bigotimes_{q\in\qVars} \C^{\Sigma_q} = \C^\Sigma,
\end{align*}
where $\Sigma = \Sigma_{\qVars} = \prod_{q \in \qVars} \Sigma_q$ labels the standard basis of the all quantum variables.
It is well understood how to extend the above to infinite-dimensional Hilbert spaces and an infinite number of quantum variables, but we will not need this here.

%-----------------------------------------------------------------------------
\paragraph{States}
The \emph{state} of all the quantum variables is described by a positive semidefinite (PSD) operator on~$\cH$ with trace equal to $1$, often called a \emph{density operator}.
We denote the set of all such operators by~$\dH$.
A state is called \emph{pure} if it is given by a rank-one projection, i.e., if~$\rho=\proj\Psi$ for some unit vector~$\ket\Psi\in\cH$.
For simplicity one often also refers to $\ket\Psi$ as the (pure) state.
States that are not pure are called \emph{mixed}.
More generally, we consider \emph{partial states}, also called \emph{subnormalized states}, which are PSD operators of trace at most one, denoted by~$\dHle$.
Partial states are akin to sub-probability measures in probabilistic computing.
They are useful for reasoning about programs that may terminate with probability less than one.
For example, the computational basis states of a qubit are~$\proj0$ and~$\proj1$, which are both pure states, and the partial state $\rho = 0.8 \proj0$ % \in \dHle[\C^2]$
describes the output of a program terminating in state~$\proj0$ with probability~$0.8$ (and not terminating otherwise).

%-----------------------------------------------------------------------------
\paragraph{Operations}
There are two basic kinds of operations.
The first is to apply a \emph{unitary}.
An operator~$U$ is called a \emph{unitary} if~$U U^\dagger = U^\dagger U = I$.
If we apply a unitary~$U$ on~$\cH$ to a (partial) state~$\rho$, the result is $U \rho U^\dagger$, which is again a (partial) state.
  For example, the Hadamard matrix $H = \frac1{\sqrt2} \begin{psmallmatrix} 1 & 1 \\ 1 & -1 \end{psmallmatrix}$
  is a one-qubit unitary,
  and on applying it to the input state $\proj0$, we get $\proj+$, where $\ket+ = \frac1{\sqrt2}(\ket0 + \ket1)$.

The second operation is to measure the quantum state.
A \emph{measurement} (also called a \emph{positive operator-valued measure} or \emph{POVM}) is given by a family of positive semidefinite operators~$\vec M = \{M_\omega\}_{\omega\in\Omega}$,
labeled by some finite index set~$\Omega$, such that~$\sum_{\omega\in\Omega} M_\omega = I$.
If one measures a (partial) state~$\rho$ then the probability of seeing outcome~$\omega\in\Omega$ is~$q_\omega = \tr(\rho M_\omega)$,
in which case the state changes to~$\rho_\omega = \sqrt{M_\omega} \rho \sqrt{M_\omega} / q_\omega$.
Note that any PSD operator~$M$ has a unique PSD \emph{square root}, denoted by~$\sqrt M$.
If $M$ is a projection, meaning $M^2 = M$, then $\sqrt M = M$.
We abbreviate $\cM_\omega(\rho) = \sqrt{M_\omega} \rho \sqrt{M_\omega}$.
Note that this is a partial state, with trace equal to the probability~$q_\omega$ of outcome~$\omega$.

When $\Omega = \{0,1\}$, there are only two possible outcomes, and this is called a \emph{binary measurement}.
Any binary measurement~ $\{B_0, B_1\}$ can be obtained by picking a positive semidefinite operator~$B$ with~$B \preceq I$ and setting~$B_0 = I - B$ and~$B_1 = B$.
As above, we abbreviate~$\cB_0(\rho) = \sqrt{I-B}\rho\sqrt{I-B}$ and~$\cB_1(\rho) = \sqrt{B} \rho \sqrt{B}$.
For example, for a qubit, $B = \proj1$ defines the standard basis measurement, with~$B_0=\sqrt{B_0}=\proj0$ and~$B_1=\sqrt{B_1}=\proj1$.
For succinctness, we will often refer to~$B$ rather than~$\{B_0,B_1\}$ as a binary measurement.
For example, if we apply the binary standard basis measurement to a qubit in the~$\proj{+}$ state, then $\cB_0(\proj{+}) = \frac12 \proj0$ and $\cB_1(\proj+) = \frac12 \proj1$.
Thus each outcome~$\omega \in \{0,1\}$ occurs with probability half and the state after the measurement is~$\proj\omega$.

We can also apply any of the above operations to a subset~$\vec{q} \subseteq \qVars$ of the quantum variables.
To this end, let us, for an arbitrary operator~$A$ on~$\cH_{\vec q}$, define the operator~$A_{\vec q} = A \ot I$ on~$\cH$, where the tensor product is with respect to the decomposition~$\cH = \cH_{\vec q} \ot \cH_{\vec q^c}$ and $\vec{q}^c = \qVars \setminus \vec{q}$ denotes the remaining quantum variables.
In prior work, this has been called \emph{cylindrical extension}~\cite{coqq2023}.
Then, if~$U$ is a unitary on~$\cH_{\vec q}$, we can apply it to the quantum variables~$\vec{q}$ by taking~$U_{\vec q}$ in the above prescription.
Similarly, if~$\{M_\omega\}_{\omega\in\Omega}$ is a measurement on $\cH_{\vec q}$, we take $\{ M_{\omega,\vec q} \}$ in the above prescription; we also abbreviate~$\cM_{\omega,\vec q}(\rho) = \sqrt{M_{\omega,\vec q}} \rho \sqrt{M_{\omega,\vec q}}$.

%-----------------------------------------------------------------------------
\subsection{Quantum While Language}
\label{subsec:qwhile}
%-----------------------------------------------------------------------------

In this paper, we consider a \emph{quantum while language}~\cite{ying2012quantum,coqq2023,feng2021quantum}.
We first introduce its syntax and then discuss its denotational semantics.
The language supports initializing quantum variables, applying unitary operations, and classical control flow based on quantum measurement outcomes.
%
%-----------------------------------------------------------------------------
% \paragraph{Syntax}
We first describe the language's syntax.
\begin{definition}[Syntax]
\label{def:qwhile-syntax}
Quantum while programs are given by the following grammar:
\begin{align*} \qwhilesyntax \end{align*}
where $S, S_1, S_2, S_{\omega_1}, S_{\omega_2}, \ldots$ denote programs in the language;
$\vec{q}$ denotes a set of quantum variables;
$U$ is a unitary operator on~$\cH_{\vec{q}}$;
$N$ is a natural number;
$\{M_{\omega}\}_{\omega\in\Omega}$ is a measurement on~$\cH_{\vec{q}}$ with outcomes in some set $\Omega = \{\omega_1, \omega_2, \ldots\}$;
$B$ is an operator defining a binary measurement~$\{I-B, B\}$.
\end{definition}
The first five instructions are self-explanatory:
$\qpenv{\textbf{skip}}$ is a no-op;
$\qpenv{\qpinit{\vec q}}$ initializes the set of variables $\vec q$;
$\qpenv{\qpunitary{\vec q}{U}}$ applies a unitary $U$ to the set of variables $\vec q$;
$\qpenv{\qpabsprog{S_1}~; \qpabsprog{S_2}}$ sequentially composes the two programs, running ${S_1}$ followed by ${S_2}$;
and $\qpenv{\qprepeat{N}{\qpabsprog{S}}}$ runs the program $S$ $N$ times.
The \textbf{case} statement is used for control flow: it measures the set of variables $\vec q$ with the measurement $\vec M$, and on seeing outcome $\omega$ runs the program ${S_\omega}$.
The \textbf{while} statement measures the set of variables $\vec q$ with a binary measurement $B$ and, if it succeeds (i.e. has outcome $1$), runs the loop body~${S}$ and repeats.
Thus, both \textbf{case} and \textbf{while} are classical control flow based on the outcome of a quantum measurement.

\begin{remark}\label{remark:unitary-subset}
  We comment on an aspect that is well-understood but usually left implicit.
  In the instruction $\qpenv{\qpunitary{\q}{U}}$, the unitary~$U$ will often be restricted to an arbitrary fixed subset of allowed ``gates'' (typical choices are few-qubit unitaries or the Clifford+T gate set).
  The same applies to measurement operators $M_\omega$ (typical choices are standard basis measurements or arbitrary 1-qubit measurements).
  The results of our paper, just like those of prior works, do not depend on this choice (\cref{thm:soundness-partial,thm:soundness-total,thm:completeness-partial,thm:completeness-total}), and all our examples use efficiently implementable unitaries.
\end{remark}

%-----------------------------------------------------------------------------
% \paragraph{Syntactic sugar}
We also introduce some syntactic sugar for convenience.
First, we define an \textbf{if} statement as a shorthand for a \textbf{case} statement with a binary measurement:
\begin{gather*}
\qpenv{\qpifte{B}{\vec q}{\qpabsprog{S_1}}{\qpabsprog{S_0}}}
\equiv
\\
\qpenv{\qpselect{\{1\colon B, 0\colon I-B\}}{\vec q}{1 \colon \qpabsprog{S_1},\ 0 \colon \qpabsprog{S_0}}}
\end{gather*}
We allow leaving out the \textbf{else} branch, in which case we take $S_0 = \qpenv{\qpskip}$.
Lastly, we allow leaving out the measurement in \textbf{if}, \textbf{case}, and \textbf{while}, in which case we default to standard basis measurements:
\begin{align*}
\qpenv{\qpifteB{q}{\qpabsprog{S_1}}{\qpabsprog{S_0}}}
  &\;\equiv\;
  \qpenv{\qpifte{\proj1}{q}{\qpabsprog{S_1}}{\qpabsprog{S_0}}}\\
\qpenv{\qpselectS{\vec q}{\dots}}
  &\;\equiv\;
  \qpenv{\qpselect{\{x\colon\proj{x}\}_{x \in \Sigma_{\vec q}}}{\vec q}{\dots}} \\
\qpenv{\qpwhileB{q}{\qpabsprog{S}}}
  &\;\equiv\;
  \qpenv{\qpwhile{\proj1}{q}{\qpabsprog{S}}}
\end{align*}

\begin{example}[Quantum coin toss until zero]\label{new:ex:coin-toss:program}
As a gentle introduction to the quantum while language, we consider the following program, which exercises several language features:
\begin{equation*}
  \ctprog =
  \qpenv{
    \qpinitS{q} ;
    \qpunitary{q}{H} ;
    \qpwhileB{q}{
      \qpunitary{q}{H}
    }
  }
\end{equation*}
It first initializes a qubit in the $\ket0$ state, and then repeatedly applies the Hadamard gate and measures in the standard basis, until the outcome~``0'' is seen.
We revisit this example later from the perspectives of denotational semantics and quantum Hoare logic.
\end{example}

%-----------------------------------------------------------------------------
\paragraph{Semantics}
Any program~$S$ defines a function that maps the state of the quantum variables before execution to the corresponding (partial) state after program execution.
This is called the \emph{denotational semantics} of $S$ and is defined formally in the following.

\begin{definition}[Denotational semantics]\label{def:semantics}
For any program $S$, its \emph{denotational semantics} is the map
\begin{equation*}
    \Bracks{S} \colon \dHle \rightarrow \dHle,
\end{equation*}
which is defined recursively in the following way:
\begin{enumerate}
  % skip
  \item $\Bracks{\qpenv{\qpskip}} (\rho) = \rho$
  % init
  \item $\Bracks{\qpenv{\qpinit{\vec q}}} (\rho)
  = \sum_{\vec x \in \Sigma_{\vec q}} \ketbra{\vec 0}{\vec x}_{\vec q} \rho \ketbra{\vec x}{\vec 0}_{\vec q}$
  % unit
  \item $\Bracks{\qpenv{\qpunitary{\vec{q}}{U}}}(\rho)
  = U_{\vec q} \rho U_{\vec q}^\dagger$
  % seq
  \item $\Bracks{\qpenv{\qpabsprog{S_1} ; \qpabsprog{S_2}}} (\rho)
  = \Bracks{S_2}\parens*{\Bracks{S_1}(\rho)}$
  % repeat
  \item $\Bracks{\qpenv{\qprepeat{N}{\qpabsprog{S}}}} (\rho)
  = \Bracks{S}^N(\rho)$
  % case
  \item $\Bracks{\qpenv{\qpselect{\vec M}{\vec{q}}{\ldots, \ \omega \colon \qpabsprog{S_{\omega}},\ \ldots}}} (\rho) = \sum_{\omega\in\Omega} \Bracks{S_\omega}(\cM_{\omega,\vec q}(\rho))$
  % while
  \item $\Bracks{\qpenv{\qpwhile{B}{\vec{q}}{\qpabsprog{S}}}}(\rho)$
  $= \sum_{k=0}^{\infty} \parens*{\cB_{0,\vec q} \circ \parens*{ \Bracks{S} \circ \cB_{1,\vec q} }^k}(\rho)$
\end{enumerate}
\end{definition}
We explain the semantics above using the definitions from \cref{subsec:quantum}.
Statements~(1)-(5) are self-explanatory.
The \textbf{case} statement~(6) measures a set of variables $\vec q$ with the measurement $\vec M$,
and on seeing outcome~$\omega$ executes program $S_\omega$.
The \textbf{while} statement~(7) runs a loop with body $S$, guarded by the condition that measuring $\vec q$ with the binary measurement $B$ gives outcome $1$.
The semantics for \textbf{while} is well-defined because the partial states~$\sigma_n = \sum_{k=0}^n \parens*{\cB_{0,\vec q} \circ \parens{ \Bracks{S} \circ \cB_{1,\vec q} }^k}(\rho)$ form an increasing sequence ($\sigma_n \preceq \sigma_{n+1}$ for all~$n$) that is bounded from above.
Note that~$\sigma_n$ represents the partial state corresponding to termination within at most~$n$ iterations of the loop.
Loops~$S = \qpenv{\qpwhile{B}{\vec{q}}{\qpabsprog{C}}}$ satisfy the following natural recurrence:
\begin{align}\label{eq:while loop recurrence}
  % \Bracks{S}(\rho)
% = \cB_{0,\vec q}(\rho) + \sum_{k=1}^{\infty} \parens*{\cB_{0,\vec q} \circ \parens*{ \Bracks{C} \circ \cB_{1,\vec q} }^k}(\rho)
% = \cB_{0,\vec q}(\rho) + \parens*{ \sum_{k=0}^{\infty} \parens*{\cB_{0,\vec q} \circ \parens*{ \Bracks{C} \circ \cB_{1,\vec q} }^k} \circ \Bracks{C} \circ \cB_{1,\vec q} }(\rho) \\
% = \cB_{0,\vec q}(\rho) + \parens*{ \Bracks{S} \circ \Bracks{C} \circ \cB_{1,\vec q} }(\rho)
  \Bracks{S} = \cB_{0,\vec q} + \Bracks{S} \circ \Bracks{C} \circ \cB_{1,\vec q}
% \quad\text{or, equivalently,}\quad
%   \Bracks{S}^\dagger = \cB_{0,\vec q}^\dagger + \cB_{1,\vec q}^\dagger \circ \Bracks{C}^\dagger \circ \Bracks{S}^\dagger.
\end{align}
The denotational semantics of a program, $\Bracks{S}$, is a so-called \emph{superoperator} because it is a linear function mapping operators on one Hilbert space to operators on another.
To be meaningful, this superoperator should map partial quantum states to partial quantum states, even when applied to a subset of the quantum variables.
Formally, this means that $\Bracks{S}$ should be \emph{completely positive and trace non-increasing}.
Recall that a superoperator $\cE$ is called
% \emph{positive} if~$\cE(M)$ is PSD for every PSD~$M$, and
\emph{completely positive} if for every additional Hilbert space~$\cH'$ the superoperator~$\cE \ot \cI_{\cH'}$ maps PSD operators to PSD operators, with~$\cI_{\cH'}$ denoting the identity superoperator on~$\cH'$, and it is called \emph{trace non-increasing} if~$\tr \cE(M) \leq \tr M$ for every~$M$ (it is called \emph{trace preserving} if equality holds for all~$M$).
Conversely, these two conditions guarantee that a superoperator can be physically realized.

One can note that~$\Bracks{S}$ is completely positive and trace non-increasing for any program~$S$.
Moreover, for any state~$\rho\in\dH$, the quantity~$\tr(\Bracks{S}(\rho)) \in [0,1]$ can be interpreted as the \emph{probability of termination} of the program~$S$ when started in the initial state~$\rho$.
In particular, the program~$S$ \emph{terminates almost surely} (that is, with probability one) when started in state~$\rho$ if, and only if, $\tr(\Bracks{S}(\rho))=1$.
Thus the program terminates on any input state if~$\Bracks{S}$ is trace preserving (we also say that~$S$ is trace preserving).
We note that~(1)-(3) are always trace preserving, (4)-(6) are trace preserving if all the subprograms ($S_1$, $S_2$, $S$, $S_\omega$ for $\omega\in\Omega$) are trace preserving, and~(7) is trace preserving if the subprogram~$S$ is trace preserving and the loop terminates with probability one.

\begin{example}[Semantics of quantum coin toss until zero]\label{new:ex:coin-toss:semantics}
We can use \cref{def:semantics} to compute the semantics for the quantum coin oss until zero program in \cref{new:ex:coin-toss:program}.
For any state~$\rho \in \dH$,
\begin{align*}
  \Bracks{{\ctprog}}(\rho)
  &=\Bracks{\qpenv{
      \qpinitS{q} ;
      \qpunitary{q}{H} ;
      \qpwhileB{q}{\qpunitary{q}{H}}
  }}(\rho) \\
  &= \Bracks{\qpenv{\qpwhileB{q}{\qpunitary{q}{H}}}} \parens[\big]{\Bracks{\qpenv{\qpunitary{q}{H}}} (\Bracks{\qpenv{\qpinitS{q}}}(\rho))} \\
  &= \Bracks{\qpenv{\qpwhileB{q}{\qpunitary{q}{H}}}} \parens[\big]{\Bracks{\qpenv{\qpunitary{q}{H}}} (\proj0)} \\
  &= \Bracks{\qpenv{\qpwhileB{q}{\qpunitary{q}{H}}}} \parens*{ \proj{+} }.
\end{align*}
using~(4), (2), and (3).
By~(7), the semantics of the loop is for a general state~$\sigma \in \dHle$ given by
\begin{align*}
  \Bracks{\qpenv{\qpwhileB{q}{\qpunitary{q}{H}}}}(\sigma)
  &= \sum_{k = 0}^{\infty} \parens*{\cB_0 \circ (\Bracks{\qpenv{\qpunitary{q}{H}}} \circ \cB_1)^k}(\sigma),
\end{align*}
where $\cB_j(\rho) = \proj{j}\!\rho\!\proj{j} = \braket{j|\rho|j} \proj j$.
Now, for any state~$\sigma \in \dHle$, we have
\[ (\Bracks{\qpenv{\qpunitary{q}{H}}} \circ \cB_1)(\sigma) = \braket{1|\sigma|1} \proj-, \]
and therefore for any~$k \ge 1$,
$(\Bracks{\qpenv{\qpunitary{q}{H}}} \circ \cB_1)^k(\proj+) = \frac1{2^k} \proj-$.
Thus we find that
\begin{align*}
  \Bracks{\qpenv{\qpwhileB{q}{\qpunitary{q}{H}}}} \parens*{ \proj{+} }
  &= \cB_0(\proj+) + \sum_{k = 1}^\infty \cB_0\parens*{\frac1{2^k}\proj-} \\
  &= \frac12\proj0 + \sum_{k = 1}^\infty \frac1{2^{k+1}}\proj0
  = \proj0,
\end{align*}
and hence the semantics of the coin toss until zero program is, for any initial state $\rho \in \dH$, given by
\begin{equation}
  \Bracks{\ctprog}(\rho) = \proj0 \label{new:eq:coin-toss:semantics}
\end{equation}
We see that no matter what state we start in, the program always terminates in the pure state~$\ket0$.
\end{example}
%
%-----------------------------------------------------------------------------
\subsection{Quantum Hoare Logic}
\label{subsec:quantum-hoare-logic}
%-----------------------------------------------------------------------------
Hoare logic is a formal system to state and prove correctness of programs.
For a program~$S$, one specifies a precondition~$P$ and postcondition~$Q$ to form a Hoare triple $\hoare{P}{S}{Q}$.
Such a Hoare triple is said to hold if, for any state that satisfies $P$, running the program on it results in a state that satisfies~$Q$.
Here~$P$ and~$Q$ are predicates over the state of the program variables.
In this section, we present the formalism for \emph{quantum Hoare logic} from~\cite{ying2012quantum}.

%-----------------------------------------------------------------------------
\paragraph{Predicates}
Predicates are properties of the state of the system that can hold to some degree.
Recall that any PSD operator $P \preceq I$ defines a binary measurement~$\{B_0,B_1\}$ by setting~$B_0 = I - P$ and $B_1 = P$.
We may think of~$P$ as defining a predicate:
As in probabilistic Hoare logic~\cite{Morgan1998pGCLFR}, instead of assigning a definite truth value to a given predicate and state, we rather assign an \emph{expectation} or degree to which the predicate holds in the given state -- namely the probability of getting outcome~$1$ if one were to apply the binary measurement defined by $P$.
For any state~$\rho$, this probability is given by~$\tr(P \rho)$, as explained earlier.
We thus arrive at the following definition.

\begin{definition}[Predicates and expectation]\label{def:pred exp}
A \emph{predicate} is a positive semidefinite operator~$P$ such that~$P \preceq I$, and the set of all such predicates is denoted $\cPHle$.
The \emph{expectation} of the predicate~$P$ in a (partial) state $\rho \in \dHle$ is defined as
\begin{align*}
  \satisfies{\rho}{P} = \tr(P\rho) \in [0,1].
\end{align*}
We say that \emph{$P$ implies $Q$} for two predicates $P, Q$ iff $\satisfies{\rho}{P} \le \satisfies{\rho}{Q}$ for all states~$\rho\in\dH$.
This is equivalent to $P \preceq Q$ in the L\"owner order, but is in the context of predicates often denoted $P \imp Q$.
\end{definition}

Just like in classical Hoare logic one can also transform predicates with respect to a program~$S$.
To this end, we use the adjoint~$\Bracks{S}^\dagger$ of the denotational semantics superoperator~$\Bracks{S}$ (see \cref{def:semantics} and the discussion below it).
For any superoperator~$\cE$, the adjoint~$\cE^\dagger$ %(with respect to the Hilbert-Schmidt inner product)
satisfies the defining property that $\tr (A~\cE(B)) = \tr(\cE^\dagger(A)~B)$ for all operators~$A$,~$B$.
% Recall from \cref{subsec:notation} that~$\Bracks{S}$ has an adjoint~$^\dagger$.
We note that~$\cE$ is completely positive iff this is the case for its adjoint; it is trace preserving iff its adjoint is \emph{unital}, that is,~$\cE^\dagger(I) = I$, and trace non-increasing if the adjoint is \emph{sub-unital}, that is, $\cE^\dagger(I) \preceq I$.
While~$\Bracks{S}$ transforms states, its adjoint~$\Bracks{S}^\dagger$ naturally acts on predicates, and we have the following useful duality:
\begin{align*}
  \satisfies{\Bracks{S}(\rho)}{P} = \satisfies{\rho}{\Bracks{S}^\dagger(P)}.
\end{align*}
for any program~$S$ and for any (partial) state $\rho\in\dHle$ and predicate~$P\in\cPHle$.

%-----------------------------------------------------------------------------
\paragraph{Quantum Hoare Triples}
A quantum Hoare triple is denoted by
\[ \hoare{P}{S}{Q} \]
and consists of a program~$S$, precondition~$P$, and postcondition~$Q$, where $P, Q$ are predicates as defined above.
Similar to probabilistic Hoare logic, we have notions of correctness of a Hoare triple.

We start with total correctness.
It states that the postcondition holds to a degree no less than the precondition:

\begin{definition}[Total correctness]
\label{def:total-correctness}
For a program~$S$ and predicates~$P$,~$Q$,
the Hoare triple $\hoare{P}{S}{Q}$ is said to be \emph{totally correct} if for all partial states~$\rho \in \dHle$,
\begin{equation*}
    \satisfies{\rho}{P} \le \satisfies{\Bracks{S}(\rho)}{Q}.
\end{equation*}
We denote total correctness by $\totcorr \hoare{P}{S}{Q}$.
Mathematically, this is equivalent to%~$P \preceq \Bracks{S}^\dagger(Q)$ or
~$P \imp \Bracks{S}^\dagger(Q)$, or~$P \preceq \Bracks{S}^\dagger(Q)$.
\end{definition}

Next, we define partial correctness.
Here the degree to which the postcondition holds only matters insofar as the program terminates.

\begin{definition}[Partial correctness]
\label{def:partial-correctness}
For a program~$S$ and predicates~$P$,~$Q$,
the Hoare triple $\hoare{P}{S}{Q}$ is said to be \emph{partially correct} if for all partial states~$\rho \in \dHle$,
\begin{equation}\label{eq:def par cor}
    \satisfies{\rho}{P}
    \le \satisfies{\Bracks{S}(\rho)}{Q}
    + [\tr(\rho) - \tr(\Bracks{S}(\rho))].
\end{equation}
This condition can be equivalently stated as
% \begin{align*}
%     \satisfies{\rho}{P} - \tr(\rho)
%     \le \satisfies{\Bracks{S}(\rho)}{Q} - \tr(\Bracks{S}(\rho))
% \end{align*}
% or
\begin{equation}\label{eq:def par cor alt}
    \satisfies{\rho}{I-P} \geq \satisfies{\Bracks{S}(\rho)}{I-Q}.
\end{equation}
We denote partial correctness by $\parcorr \hoare{P}{S}{Q}$.
Mathematically, this is equivalent to% $I-P \succeq \Bracks{S}^\dagger(I-Q)$, or
% ~$P \imp I - \Bracks{S}^\dagger(I-Q)$.
~$\Bracks{S}^\dagger(I-Q) \imp I - P$, or~$I-P \succeq \Bracks{S}^\dagger(I-Q)$.
\end{definition}

As explained below \cref{def:semantics}, the term $[\tr(\rho) - \tr(\Bracks{S}(\rho))]$ in \cref{eq:def par cor} is the \emph{probability of non-termination} of the program~$S$ when started in an initial state~$\rho$.
It is always non-negative (as~$\Bracks{S}$ is trace non-increasing).
Intuitively, \cref{eq:def par cor} states that the degree to which the postcondition holds is at least the degree to which the precondition holds, \emph{minus the probability of non-termination}.
The equivalent \cref{eq:def par cor alt} says that the probability that the program terminates and the postcondition does \emph{not} hold is at most the probability that the precondition does \emph{not} hold.
Total implies partial correctness.

For later use, we observe that~$\totcorr \hoare {I}{S}{Q}$ means that, for any initial state, the program terminates almost surely in a state satisfying the postcondition, while~$\parcorr \hoare {I}{S}{Q}$ means that the postcondition holds whenever the program terminates.
In particular, $\totcorr \hoare {I}{S}{I}$ states that~$S$ terminates almost surely on any initial state, while $\parcorr \hoare {I}{S}{I}$ holds trivially for any program.
\begin{example}[Hoare logic specification for quantum coin toss until zero]\label{new:ex:coin-toss:spec}
We now discuss a natural quantum Hoare triple for our running example (\cref{new:ex:coin-toss:program}).
One way to specify the behavior of the program~$\ctprog$ is by the Hoare triple
\begin{equation*}
  \hoare{I}{\ctprog}{\proj0}.
\end{equation*}
As discussed, this states that, for any input state, the program terminates in the final state~$\ket0$.
We can verify explicitly that this Hoare triple program is totally correct.
Indeed, we saw in \cref{new:eq:coin-toss:semantics} of \cref{new:ex:coin-toss:semantics} that $\Bracks{\ctprog}(\rho) = \proj0$ for every state~$\rho\in\dH$, and hence
\begin{align*}
  \satisfies{\Bracks{\ctprog}(\rho)}{\proj0}
= \satisfies{\proj0}{\proj0}
% = \tr \proj0 \proj0
= 1
= \tr\rho
= \satisfies{\rho}{I}
\end{align*}
for any state~$\rho\in\dH$.
This confirms that the triple is totally correct.
There are also much simpler programs that meet this same specification, e.g.~$\qpenv{\qpinit{q}}$.
See also the discussion in \cref{app:coin-toss-till-zero}.
\end{example}
%
%-----------------------------------------------------------------------------
\paragraph{Projections as Predicates}
When the precondition in a Hoare triple is a projection, $P^2 = P$, we only need to verify correctness for pure states~$\rho=\proj\psi$ that exactly satisfy the precondition, meaning~$\satisfies \rho P = 1$ or equivalently~$P \ket\psi = \ket\psi$~\cite[Theorem 3.2]{aQHL2019}.
In particular, we can specify the behavior of the program when run with some initial pure state~$\ket\Psi$
by using the precondition~$P = \proj\Psi$:
A Hoare triple $\hoare{\proj\Psi}{S}{Q}$ is totally correct iff~$\satisfies{\Bracks{S}(\proj\Psi)}{Q} = 1$,
and partially correct iff~$\satisfies{\Bracks{S}(\proj\Psi)}{Q} = \satisfies{\Bracks{S}(\proj\Psi)}{I}$.
Thus:
\begin{lemma}\label{cor:projection-predicates}
Let~$P$ be a projection and~$Q$ an arbitrary predicate.
Then,
$P \imp Q$ holds
if, and only if, $\braket{\psi|Q|\psi} = 1$ for every unit vector~$\ket \psi$ such that~$P\ket\psi = \ket\psi$.
\end{lemma}

More generally, the precondition~$P = \proj\psi_{\vec q} \ot I_{\vec{q}^c}$ can be used to specify the behavior on input states where the quantum variables~$\vec q$ are in some pure state~$\ket\psi \in \cH_{\vec q}$.
The situation simplifies further if the postcondition is also given by a pure state, say~$Q=\proj\Phi$.
Indeed, the Hoare triple~$\hoare {\proj\Psi} S {\proj\Phi}$ is totally correct iff~$\Bracks{S}(\proj\Psi) = \proj\Phi$,
and partially correct iff~$\Bracks{S}(\proj\Psi) = p \proj\Phi$ for some arbitrary probability of termination~$p\in[0,1]$.
In other words, total (or partial) correctness of the above Hoare triple means that running the program~$S$ on state~$\ket\Psi$ results in state~$\ket\Phi$ (if the program terminates).
We use the above observations later when specifying teleportation and search in \cref{sec:examples}.
We caution that we \emph{cannot} specify program behavior on mixed initial states~$\rho$ by taking~$P=\rho$.
% For example, for the mixed state $\rho = \frac I 2 = \frac12 \proj0 + \frac12 \proj1$ we have~$\satisfies{\rho}{\rho} = \tr(\rho^2) = \frac12 < 1$.
Since any mixed state can be purified, this does not impose a real restriction.

%-----------------------------------------------------------------------------
\subsection{Multiple Specifications}\label{subsec:multispec}
We often want to specify that a single program satisfies several Hoare triples at once.
We give three motivating examples:
\begin{enumerate}
\item
To prove a Hoare triple~$\hoare P {\qpenv{\qprepeat N S}} Q$ correct, it suffices to prove that the loop body satisfies~$\hoare {R_j} S {R_{j+1}}$ for all~$j\in\{0,\dots,N-1\}$, for predicates~$R_0,\dots,R_N$ with~$P \imp R_0$ and~$R_N \imp Q$.
\item
To assert that a program~$S$ terminates on any input, we can always add the Hoare triple~$\hoare I S I$ on top of any other Hoare triple that we also want to hold (as discussed earlier).
\item
A program~$S$ that creates a qubit that, when measured, gives $x \in \{0,1\}$ with 50\% probability each, can be specified by two Hoare triples~$\hoare {I/2} S {\proj x}$ for $x\in\{0,1\}$ (see \cref{subsec:example:fair-coin}).
\end{enumerate}
At other times, we may also want to allow the program~$S$ itself to depend on some parameter.
For example, the search algorithms of \cref{subsec:examples:qsearch} will necessarily have to depend explicitly on the database that is being queried.
We will thus consider Hoare triples
\begin{align*}
  \hoare {P_\lambda} {S_\lambda} {Q_\lambda},
\end{align*}
where the pre- and postconditions as well as the program are parameterized by~$\lambda$ in some index set~$\Lambda$.
Such a Hoare triple is (totally or partially) correct if it is correct for every definite value of~$\lambda$.
If the program~$S_\lambda$ does not depend on~$\lambda$, we have a single program that satisfies multiple specifications, while if it depends on (part of)~$\lambda$, then we have a family of programs.
To stick with the literature, we will think of~$\lambda$ as a \emph{formal (or meta) parameter} that is implicitly quantified over universally, but one could instead also adjust the definition of predicates to be functions~$\Lambda \to \cPHle$ (and extend the notions of expectation, implication, and so forth in a straightforward way).

The following notation will be useful: if~$\lambda$ ranges over a finite set of options~$\{\lambda_1, \ldots, \lambda_k\}$, then we will also write
\[
\hoare{P_{\lambda_1}, \dots, P_{\lambda_k}}{S_{\lambda_1}, \dots, S_{\lambda_k}}{Q_{\lambda_1}, \dots, Q_{\lambda_k}}
\]
instead of $\hoare {P_\lambda} {S_\lambda} {Q_\lambda}$.
If the precondition, the postcondition, or the program do not depend on~$\lambda$, we will write the corresponding term only once.
For instance, the third motivating example above could alternatively be written as $\hoare{I/2}{S}{\proj0, \proj1}$.

%=============================================================================
\section{Correctness by Construction for Quantum Programs}
\label{sec:qbc}
%=============================================================================
In the Correctness-by-Construction (CbC) approach, one starts with a specification and successively refines it to construct a program that is guaranteed to satisfy the initial specification.
Prior work on CbC defined refinement rules in terms of Hoare triples:
one could replace an ``abstract'' program~$S$ in a Hoare triple~$\hoare{P}{S}{Q}$ (where~$P, Q$ are some predicates) by some concrete program~$S'$ provided certain side conditions were satisfied, which often involved the validity of other Hoare triples for subprograms that have to be constructed beforehand.

Here, we extend the quantum while language by a new construct~$\qphole{P}{Q}$, called a \emph{hole}, which represents a yet-to-be-constructed subprogram that carries a precondition~$P$ and a postcondition~$Q$ (\cref{subsec:qwhile-with-holes}).
This notion of programs with holes is similar in spirit to abstract execution~\cite{absexec2019}, where one is interested in executing and analyzing programs containing ``abstract statements'',
and also the concept of typed holes in programming~\cite{gradualtyping,10.1145/3290327,wiki2014typedholes}.
For both partial and total correctness, we then proceed to provide \emph{refinement rules} (\cref{sec:refinement-relation}) that can be used to construct correct programs by replacing holes with concrete programs.
Next, we prove that our refinement rules are sound: any program~$S$ constructed from a specification~$\qpenv{\qphole{P}{Q}}$ must satisfy that specification, meaning that the Hoare triple~$\hoare{P}{S}{Q}$ is correct (\cref{sec:soundness}).
Finally, we show that our refinement rules are complete: any program~$S$ satisfying a Hoare triple~$\hoare{P}{S}{Q}$ can be constructed from the specification~$\qpenv{\qphole{P}{Q}}$ (\cref{sec:completeness}).

%-----------------------------------------------------------------------------
\subsection{Quantum While Language with Holes}
\label{subsec:qwhile-with-holes}
%-----------------------------------------------------------------------------
To support QbC specifications, we first extend our quantum while language in \cref{def:qwhile-syntax} with a new construct: \emph{holes}.
A hole is a yet-to-be-constructed program tagged with a precondition and a postcondition, such that the corresponding Hoare triple should be satisfied once the hole is filled by a program.
We define the syntax of the extended language below.

\begin{definition}[Syntax]
\label{def:programs-with-holes}
Programs in the \emph{quantum while language with holes} are given by the following grammar:
\begin{align*}
  \qwhilesyntax \;|\; \qphole{P}{Q}.
\end{align*}
The new construct~$\qpenv{\qphole{P}{Q}}$ is called a \emph{hole} with precondition~$P$ and a postcondition~$Q$, which are arbitrary predicates.
Apart from this, the above grammar is identical to the quantum while language (\cref{def:qwhile-syntax}).
A program that may contain holes is called an \emph{abstract program}, and one that does not contain any holes is called a \emph{concrete program}.
In other words, concrete programs are simply programs in the quantum while language.
\end{definition}

In the following we will also be interested in holes that satisfy multiple pre- and postconditions, which can be formalized just as discussed for Hoare triples (\cref{subsec:multispec}).
We will denote these as~$\qpenv{\qphole{P_\lambda}{Q_\lambda}}$, where~$\lambda$ is some formal parameter, or use short-hand notation such as~$\qpenv{\qphole{P,P'}{Q,Q'}}$.
See, e.g., \ruleref{HP.split} and \ruleref{H.repeat} below, and \cref{rem:multispec rules} for further discussion.

%-----------------------------------------------------------------------------
\subsection{Refinement Rules}
\label{sec:refinement-relation}
%-----------------------------------------------------------------------------
Refinement is the process of replacing holes in abstract programs with other (abstract or concrete) programs.
To this end, we define \emph{refinement relations} on abstract programs, and we use these iteratively to construct concrete programs from specifications given by a single hole~$\qpenv{\qphole{P}{Q}}$.
We first define refinement rules that ensure partial correctness:

\begin{definition}[Refinement for partial correctness]\label{def:parref}
We define a relation~$\parref$, called \emph{refinement for partial correctness} on programs with holes~(\cref{def:programs-with-holes}) as follows:
For any two predicates~$P$~and~$Q$,
\begin{enumerate}[align=right,leftmargin=1.5cm]
  \ruleitemNOREF{H.skip}
    $\qpenv{\qphole{P}{Q}} \parref \qpenv{\qpskip}$,
    if $P \imp Q$.
  \ruleitemNOREF{H.init}
    $\qpenv{\qphole{P}{Q}} \parref \qpenv{\qpinit{\vec q}}$,
    if $P \imp \sum_{\vec x\in\Sigma_{\vec q}} \ketbra{\vec x}{\vec 0}_{\vec q} Q \ketbra{\vec 0}{\vec x}_{\vec q}$.
  \ruleitemNOREF{H.unit}
    $\qpenv{\qphole{P}{Q}} \parref \qpenv{\qpunitary{\vec{q}}{U}}$,
    if $P \imp U_{\vec q}^\dagger Q U_{\vec q}$.
    % or, equivalently, $U_{\vec q} P U_{\vec q}^\dagger \imp Q$.
%
  \ruleitemNOREF{H.seq}
    $\qpenv{\qphole{P}{Q}} \parref \qpenv{\qphole{P}{R} ; \qphole{R}{Q}}$
    % for any pred.~$R$.
    for any predicate~$R$.
  \ruleitem{HP.split}
    $\qpenv{\qphole{P}{Q}} \parref \qpenv{\qphole{P_\gamma}{Q_\gamma}}$
    for any two families of predicates~$P_\gamma,Q_\gamma$ for~$\gamma$ in some index set $\Gamma$, such that~$P \imp \sum_\gamma p_\gamma P_\gamma$ and~$\sum_\gamma p_\gamma Q_\gamma \imp Q$ for a probability distribution~$p_\gamma$.
  \ruleitemNOREF{H.repeat}
    $\qpenv{\qphole{P}{Q}} \parref \qpenv{\qprepeat{N}{\qphole{R_j}{R_{j+1}}}}$,
    where $j \in \irange{0}{N - 1}$,
    for any predicates~$R_0, R_1, \ldots, R_N$
    such that $P \imp R_0$ and $R_N \imp Q$.
  \ruleitemNOREF{H.case}
  $\qpenv{\qphole{P}{Q}} \parref \qpenv{
    \qpselect {\vec M} {\vec q} {\{\omega \colon \qphole{P_\omega}{Q}\}_{\omega\in\Omega}}
  }$

  for any family of predicates~$P_\omega$ for $\omega\in\Omega$ such that $P \imp \sum_{\omega \in \Omega} \cM_{\omega}(P_\omega)$.
  \ruleitem{HP.while}
    $\qpenv{\qphole{P}{Q}} \parref \qpenv{
      \qpwhile{B}{\vec q}{\qphole{R}{\cB_{0,\q}(Q) + \cB_{1,\q}(R)}}
    }$,

    for any predicate $R$ such that $P \imp \cB_{0,\q}(Q) + \cB_{1,\q}(R)$.
    % where the measurement is $\{B_0 = I-B, B_1=B\}$.
\end{enumerate}
We also have rules for composite statements:
\begin{enumerate}[align=right,leftmargin=1.5cm,resume] %[align=left,resume]
\ruleitemNOREF{C.seqL}
  $\qpenv{\qpabsprog{S_1'};\qpabsprog{S_2}}
    \parref
    \qpenv{\qpabsprog{S_1};\qpabsprog{S_2}}$,
  if $S_1' \parref S_1$.
\ruleitemNOREF{C.seqR}
  $\qpenv{\qpabsprog{S_1};\qpabsprog{S_2'}}
    \parref
    \qpenv{\qpabsprog{S_1};\qpabsprog{S_2}}$,
  if ${S_2'} \parref {S_2}$.
\ruleitemNOREF{C.repeat}
  $\qpenv{\qprepeat{N}{\qpabsprog{S'}}}
    \parref
  \qpenv{\qprepeat{N}{\qpabsprog{S}}}$,
  if $S' \parref S$.
\ruleitemNOREF{C.case}
  $\qpenv{\qpselect{\vec M}{\vec q}{\omega \colon \qpabsprog{S'_\omega}~, \ \dots}} \parref$\\
  $\qpenv{\qpselect{\vec M}{\vec q}{\omega \colon \qpabsprog{S_\omega}~, \ \dots}}$, \\
  if ${S'_\omega} \parref {S_\omega}$ for one $\omega\in\Omega$ (and the rest unchanged).
\ruleitemNOREF{C.while}
  $\qpenv{\qpwhile{B}{\vec q}{\qpabsprog{S'}}} \parref \qpenv{\qpwhile{B}{\vec q}{\qpabsprog{S}}}$,\\
  if ${S'}~\parref~{S}$.
\end{enumerate}
For any two programs $S, S'$, we say \emph{$S'$ refines in one step to $S$ ensuring partial correctness} if~$S' \parref S$.
More generally, for any~$k\geq0$ we define~$S' \parref^k S$ if $S'$ refines to $S$ in $k$ such steps.
We say \emph{$S'$ refines to $S$ ensuring partial correctness} if~$S$ can be obtained from $S'$ by applying any number of refinement steps and denote this by~$S' \parref^* S$.
Clearly, $\parref^* = \bigcup_{k = 0}^{\infty} \parref^k$ is the reflexive and transitive closure of the relation~$\parref$.
\end{definition}

In \cref{def:parref}, the rules labeled (H.$*$) and (HP.$*$) are used to refine a single hole to another program (H stands for hole).
The rules labeled (C.$*$) are used to refine holes in composite programs (C stands for composite).
See \cref{subsec:example:fair-coin,app:coin-toss-till-zero} for pedagogical expositions on applying refinements to construct quantum coin-tossing programs.
The first three rules refine to concrete program statements:
\ruleref{H.skip} refines to a $\qpskip$ statement,
\ruleref{H.init} refines to an initialization $\qpinit{\q}$,
and
\ruleref{H.unit} refines to a unitary application $\qpunitary{\q}{U}$.
To refine a hole to a sequence of two holes, we can use the \ruleref{H.seq} rule with any arbitrary intermediate condition $R$.
To motivate \ruleref{HP.split}, observe that if a program satisfies two Hoare triples $\hoare P S Q$ and $\hoare{P'} S {Q'}$, then it also satisfies any combination $\hoare{p P + (1-p) Q} S {p Q + (1-p)Q'}$ for~$p\in[0,1]$.
So to ensure the latter it suffices to construct a program that satisfies the former.
Note that in stating this rule we use the syntax for multiple specifications discussed below \cref{def:programs-with-holes}.
In \ruleref{HP.split} we also allow for weakening preconditions and strengthening postconditions; we isolate this in \ruleref{H.sw} below for convenience.
To refine a hole to a \textbf{repeat} statement using \ruleref{H.repeat}, we must find a family of predicates $R_0, R_1, \ldots, R_N$,
such that $R_n$ holds after running the loop body $0 \le n \le N$ times, with $R_0$ implied by the precondition and $R_N$ implying the postcondition.
The body of the \textbf{repeat} statement is $\qpenv{\qphole{R_j}{R_{j + 1}}}$ where $j \in \irange{0}{N-1}$ is a formal parameter, meaning that the yet-to-be-constructed program must satisfy these specification for all such~$j$.
To refine a hole to a \textbf{case} statement using \ruleref{H.case}, we must find a family of predicates~$P_\omega$ for each measurement outcome~$\omega$ such that if~$P$ holds before the measurement then~$P_\omega$ holds after measurement upon seeing outcome~$\omega$.
Finally, to refine a hole to a \textbf{while} statement using \ruleref{HP.while}, we need to find an ``invariant'' $R$ for the loop body.
If the measurement succeeds after the execution of the loop body, then~$R$ must hold; otherwise~$Q$ must hold.
Finally, the composite rules allow refining any hole in a composite program using any of the rules above.

\begin{remark}[Multiple specifications and formal parameters]\label{rem:multispec rules}
As mentioned, some refinement rules can introduce holes with multiple specifications.
For example, \ruleref{H.repeat} introduces a hole with pre- and postcondition labeled by an index $j \in \irange{0}{N-1}$, and \ruleref{HP.split} introduces a new parameter~$\gamma$ in some arbitrary index set~$\Gamma$.
We can model this formally by implicitly extending the index set~$\Lambda$ of \cref{subsec:multispec} to include this new parameter.
Later refinements may depend on this parameter.
For example, if we refine a hole produced by \ruleref{H.repeat} with the \ruleref{H.seq} rule, then the intermediate predicates~$R$ may also depend on~$j$.
See \cref{sec:examples} for many examples.
\end{remark}

We can also apply these rules to the syntactic sugar introduced in \cref{subsec:qwhile} and deduce some other rules for convenience.
For example, since an \textbf{if} statement is shorthand for a \textbf{case} statement, we also have
\begin{enumerate}[align=left]
\ruleitem{H.ifElse}
  $\qpenv{\qphole{P}{Q}} \parref \qpenv{
    \qpifte {B} {\vec q} {\qphole{R_1}{Q} } { \qphole{R_0}{Q} }
  }$, \\
  for any predicates~$R_0,R_1$ s.th.\ $P \imp \cB_{0,\q}(R_0)+\cB_{1,\q}(R_1)$.
\end{enumerate}
We can also introducing an \textbf{if} statement without an \textbf{else} branch:
\begin{enumerate}[align=left]
\ruleitem{H.if}
  $\qpenv{\qphole{P}{Q}} \parref^* \qpenv{\qpif {B} {\vec q} { \qphole{R}{Q} } }$, \\
  for any predicate~$R$ such that $P \imp \cB_{0,\q}(Q) + \cB_{1,\q}(R)$.
\end{enumerate}
This follows from \ruleref{H.ifElse} by further refining the second hole using \ruleref{H.skip}.

Finally, we have the following rule that shows that we may always construct a program that has a weaker precondition and a stronger postcondition:
\begin{enumerate}[align=left]
  \ruleitem{H.sw}
    $\qpenv{\qphole{P}{Q}} \parref \qpenv{\qphole{P'}{Q'}}$
    for any predicates $P', Q'$ such that $P \imp P'$ and $Q' \imp Q$.
\end{enumerate}
This is a special case of \ruleref{HP.split} where we take the two families to consist of a single predicate~$P'$ and $Q'$, respectively.

\begin{remark}[Challenges of quantum verification]
The quantum setting poses some interesting new challenges.
For example, in a \ruleref{H.case} rule for classical programs, one can always choose the predicates as~$P_\omega = P \land (q = \omega)$.
In the quantum setting, there is generally no canonical choice of the predicates~$P_\omega$.
We leave the problem of finding suitable heuristics
to future work.
\end{remark}

We also define refinement rules that ensure total correctness:

\begin{definition}[Refinement for total correctness]
\label{def:totref}
We define a relation~$\totref$, called \emph{refinement for total correctness} on programs with holes, by using all the rules from \cref{def:parref} that are labeled (H.$*$) or (C.$*$), by replacing each $\parref$ with $\totref$, but replacing \ruleref{HP.while} and \ruleref{HP.split} by the following rules respectively:
\begin{enumerate}[align=right,leftmargin=1.5cm]
\ruleitemNOREF{HT.while}
  $\qpenv{\qphole{P}{Q}}
  \totref
  \qpenv{
      \qpwhile{B}{\q}{
      \qphole{R_{n+1}}{\cB_{0,\q}(Q) + \cB_{1,\q}(R_n)}
      }
  }
  $, \\
  for any binary measurement~$B$ and sequence of predicates~$\{R_n\}_{n\in\N}$ that is weakly increasing in the sense that~$R_n \imp R_{n+1}$ for all~$n\in\N$, such that~$R_0 = 0$ and the limit $R := \lim_{n\to\infty} R_n$ satisfies $P \imp \cB_{0,\q}(Q) + \cB_{1,\q}(R)$.
\ruleitemNOREF{HT.split}
    $\qpenv{\qphole{P}{Q}} \totref \qpenv{\qphole{P_\gamma}{Q_\gamma}}$
    for any two families of predicates~$P_\gamma,Q_\gamma$ for~$\gamma$ in some index set $\Gamma$, such that~$P \imp \sum_\gamma p_\gamma P_\gamma$ and~$\sum_\gamma p_\gamma Q_\gamma \imp Q$ for some~$p_\gamma\geq0$.
\end{enumerate}
For any two programs $S, S'$, we say that \emph{$S'$ refines in one step to $S$ ensuring total correctness} if~$S' \totref S$.
More generally, for any~$k\geq0$, we define~$S' \totref^k S$ if $S'$ refines to $S$ in $k$ such steps.
We say that~\emph{$S'$ refines to $S$ ensuring total correctness} if~$S$ can be obtained from $S'$ by applying any number of refinement steps and denote this by~$S' \totref^* S$.
Similarly as above, $\totref^* = \bigcup_{k = 0}^{\infty} \totref^k$ is the reflexive and transitive closure of the relation~$\totref$.
\end{definition}

The new \ruleref{HT.while} rule can be used to construct \textbf{while} loops that are totally correct.
% with a weakly increasing sequence of predicates~$\{R_n\}$.
To understand it intuitively, note that the subprogram obtained by unrolling the loop body $n$ times (ignoring the initial measurement) satisfies the specification~$\qpenv{\qphole{R_n}{\cB_0(Q)}}$,
This can be interpreted as follows: if we start inside the loop and~$R_n$ holds (with some probability), then the loop terminates within~$n$ iterations in a state that satisfies~$Q$ (with at least that probability).
We remark that the limit~$R$ (which is easily seen to always exist) precisely satisfies the requirements on the ``loop invariant'' of the \ruleref{HP.while} rule.
See \cref{sec:examples} for examples and \cref{new:sec:related-work} for related work on proving termination of probabilistic programs using loop invariant predicates.

The \ruleref{HT.split} rule is more general than \ruleref{HP.split} as it allows arbitrary non-negative~$p_\omega$ that need not add up to one.

% Comment: HP.split and HT.split are not needed for completeness, but are useful in practice.

%------------------------------------------------------------------------------
\subsection{Soundness of Refinement}
\label{sec:soundness}
%------------------------------------------------------------------------------
We now show that the sets of rules given above are \emph{sound}, meaning that each ensures the correctness of constructed programs, in the following sense:
if one starts with a specification, that is, a single hole~$\qpenv{\qphole{P}{Q}}$,
and repeatedly refines to construct a concrete program~$S$ (i.e., a program without holes),
then the constructed program satisfies the initial specification, meaning that Hoare triple $\hoare{P}{S}{Q}$ is correct.
This holds for both partial and total correctness:

% We first state this result for refinement for partial correctness:

\begin{restatable}[Soundness of refinement for partial correctness]{theorem}{thmsoundnesspartial}
\label{thm:soundness-partial}
For any two predicates~$P, Q$ and any concrete program~$S$,
if $\qpenv{\qphole{P}{Q}} \parref^* {S}$, then $\parcorr \hoare{P}{S}{Q}$.
\end{restatable}

% We now prove that the total correctness refinement system (\cref{def:totref}) is sound:

\begin{restatable}[Soundness of refinement for total correctness]{theorem}{thmsoundnesstotal}
\label{thm:soundness-total}
For any two predicates~$P, Q$ and any concrete program~$S$,
if $\qpenv{\qphole{P}{Q}} \totref^* S$ then $\totcorr \hoare{P}{S}{Q}$.
\end{restatable}

\Cref{thm:soundness-partial,thm:soundness-total} are proved by induction over the length of the refinement chain, with the \ruleref{HP.while} and \ruleref{HT.while} rules being most delicate.
The~detailed~proofs~can~be~found~in~\cref{app:soundness}.

%------------------------------------------------------------------------------
\subsection{Completeness of Refinement}
\label{sec:completeness}
%------------------------------------------------------------------------------
We now show that the refinement rules given above are \emph{complete}, meaning that they allow us to construct any correct program, in the following sense:
if one has a concrete program $S$ (i.e., a program without holes) such that the Hoare triple $\hoare{P}{S}{Q}$ is correct,
then the program can be constructed from $\qpenv{\qphole{P}{Q}}$ by applying a finite number of refinements.
Formally, we have the following results:
% We first prove this result for refinement for partial correctness:

\begin{restatable}[Completeness of refinement for partial correctness]{theorem}{thmcompletenesspartial}
\label{thm:completeness-partial}
For any two predicates $P, Q$ and any concrete program $S$,
if $\parcorr \hoare{P}{S}{Q}$
then $\qpenv{\qphole{P}{Q}} \parref^* S$.
\end{restatable}

% We now state and prove this result for refinement for total correctness:

\begin{restatable}[Completeness of refinement for total correctness]{theorem}{thmcompletenesstotal}
\label{thm:completeness-total}
For any two predicates~$P, Q$ and any concrete program~$S$,
if $\totcorr \hoare{P}{S}{Q}$ then $\qpenv{\qphole{P}{Q}} \totref^* S$.
\end{restatable}

\Cref{thm:completeness-partial,thm:completeness-total} can be proved by induction on the structure of the program, with the \textbf{while} construct requiring a careful analysis.
The detailed proofs can be found in \cref{app:completeness}.

%=============================================================================
\section{Examples}
\label{sec:examples}
%=============================================================================
In this section, we demonstrate how to use the QbC approach to construct quantum programs from their specification.
We first discuss a pedagogical example of a fair quantum coin.
We then construct a quantum teleportation protocol (\cref{subsec:examples:teleportation}) and two quantum search algorithms (\cref{subsec:examples:qsearch}).
We then describe new refinement rules to boost the success probability of quantum algorithms (\cref{subsec:examples:boosting}).
This illustrates how QbC can be usefully extended by higher-level algorithmic patterns and construction principles.
Particularly for the more complicated algorithms, we find that the QbC approach allows naturally discovering program detail on the fly, without explicitly using a priori knowledge of the final algorithms.
It also suggests key design decisions that give rise to different quantum programs satisfying the same specification.
Additionally, in \cref{app:coin-toss-till-zero} we discuss how the running example from \cref{sec:prelims} can be constructed from a structured specification.

%-----------------------------------------------------------------------------
\subsection{Fair Quantum Coin}\label{subsec:example:fair-coin}
%-----------------------------------------------------------------------------

%\paragraph{Problem}
A fair quantum coin is a program that prepares a quantum bit in a state that, when measured, gives rise to either outcome with 50\% probability.
There are infinitely many such states, but two natural ones are the \emph{Hadamard basis} states~$\ket\pm = ( \ket0 \pm \ket1 )/\sqrt2$.
We use a single qubit quantum variable~$q$. % in this example.

\paragraph{Specification}
To specify that each outcome occurs with 50\%~probability, we can use the following program with a single hole:
\[ S_\text{fair-coin} = \qpenv{\qphole{I/2}{\proj x}}, \]
where $x \in \{0, 1\}$ is a formal parameter (see \cref{subsec:multispec}).
Indeed, suppose we manage to refine the above into a program~$S'$ without holes which does not explicitly use the parameter~$x$.
Then our soundness result (\cref{thm:soundness-total}) guarantees that the Hoare triple
\[ \hoare{I/2}{S'}{\proj x} \]
is valid for every $x\in\{0,1\}$, meaning that if we run the program (on an arbitrary state) and measure the qubit, we obtain either outcome~$x\in\{0,1\}$ with probability at least, and hence equal to~$\frac12$.

\paragraph{Construction}
The well-known idea is that the Hadamard gate
\begin{align}\label{eq:hadamard}
  H = \frac1{\sqrt2} \begin{pmatrix}1 & 1 \\ 1 & -1\end{pmatrix}
\end{align}
maps the standard basis to the Hadamard basis, which allows us to realize the coin toss.
To confirm this, we consider the following sequence of refinements:
\begin{align*}
  S_\text{fair-coin}
  % &= \qpenv{\qphole{I/2}{\proj x}} \\
  &\totref \qpenv{\qphole{I/2}{R} ; \qphole{R}{\proj x}} \ruletag{H.seq} \\
  &\totref \qpenv{\qpinitS{q} ; \qphole{R}{\proj x}} \ruletag{H.init} \\
  &\totref \qpenv{\qpinitS{q} ; \qpunitary{q}{H}} \ruletag{H.unit}
\end{align*}
The first refinement is always valid, but how should we pick the predicate~$R$ so that the subsequent refinements can be applied?
To apply \ruleref{H.init}, we need that
$I/2 \imp \sum_{y \in \{0, 1\}} \ketbra{y}{0} R \ketbra{0}{y}$,
meaning that $\braket{0|R|0} \geq \frac12$ for $x\in\{0,1\}$.
To apply \ruleref{H.unit}, we should choose~$R$ such that~$R \imp H \proj x H$.
Since the latter is a pure state, this suggests $R = H \proj x H$ (which is also the weakest precondition for the Hadamard subprogram and postcondition), and for this choice we have that
$\braket{0|R|0} = \abs*{\braket{0|H|x}}^2 = \frac12$.
Thus the above refinements are valid and we have constructed a program that implements the fair coin toss specification correctly, by construction.
\begin{align*}
  S_\text{fair-coin} \totref^* \qpenv{\qpinitS{q} ; \qpunitary{q}{H}},
\end{align*}

%-----------------------------------------------------------------------------
\subsection{Quantum Teleportation}
\label{subsec:examples:teleportation}
%-----------------------------------------------------------------------------

%\paragraph{Problem}
Imagine two parties, Alice and Bob, who share a maximally entangled state, say qubits~$a,b$ in state
\begin{equation}\label{eq:bell-state}
  \ket{\phi^+} = \frac1{\sqrt2}(\ket{00} + \ket{11}).
\end{equation}
Alice has another qubit~$q$ in an unknown \emph{quantum} state, and wants to transfer its state to Bob's qubit~$b$, by sending only \emph{classical} information but utilizing the maximally entangled state as a resource.
Furthermore, if~$q$ was correlated or entangled with other quantum variables, then after teleportation the same should be true for~$b$.
This is called \emph{quantum teleportation} and is a basic building block for quantum communication (see, e.g., \citet{nielsen2010quantum}).

\paragraph{Specification}
To specify teleportation, consider an arbitrary quantum state~$\rho$ between Alice's and some arbitrary other quantum variable~$r$, which does not participate in the protocol (in quantum information, $r$ is called a reference system).
After teleportation, we want the state of Bob's qubit~$b$ and~$r$ to be in the same state.
Without loss of generality, we can take~$\rho$ to be a maximally entangled state, $\rho = \proj{\phi^+}$ (that is, it suffices to realize ``entanglement swapping'').
Indeed, a basic principle of quantum information theory asserts that if two quantum programs (completely positive maps) have the same behavior when applied to one half of a maximally entangled state, then they must have the same behavior on all states~\cite{nielsen2010quantum,wilde2013quantum}.
Recall from \cref{cor:projection-predicates} that we can assert that quantum variables are in a given pure state by using its projection as the predicate.
This translates to the following initial specification:
\begin{align}\label{eq:orig teleport}
  S_\text{teleport} = \qpenv{\qphole{\proj{\phi^+}_{qr} \ot \proj{\phi^+}_{ab}}{\proj{\phi^+}_{br}}}
\end{align}

Now, an arbitrary program that meets the above specification will not be a teleportation protocol, since in teleportation we want to constrain all quantum operations to Alice and Bob's variables, while only allowing them to communicate classical bits from Alice to Bob.
This disallows, e.g., simply swapping qubits~$q$ and~$b$ by applying a quantum gate.
We can implement this latter constraint by the following sequence of refinements
to obtain a program with two holes, where the first hole will be further refined by a program acting only on Alice's qubits~$q$ and~$a$, and the second hole by a program acting only on Bob's qubit~$b$:
\begin{align*}
  &S_\text{teleport}
  \totref
  \qpenv{
      \qphole{\proj{\phi^+}_{qr} \ot \proj{\phi^+}_{ab}}{P} ;
      \qphole{P}{\proj{\phi^+}_{br}}
  }
  \ruletag{H.seq} \\
  &\totref
  \qpenv{
    \begin{aligned}
      &\qphole{\proj{\phi^+}_{qr} \ot \proj{\phi^+}_{ab}}{P} ; \\
      &\qpselectS{q, a}{(x, y) \colon \qphole{Q(x, y)}{\proj{\phi^+}_{br}}}
    \end{aligned}
  },
  \ruletag{H.case}
\end{align*}
where $x, y \in \{0, 1\}$ for arbitrary predicates $Q(x, y)$.
We can ensure that \ruleref{H.case} is valid by taking
\begin{align}\label{eq:P via Qxy}
  P := \sum_{x, y} \proj{xy}_{qa} Q(x, y) \proj{xy}_{qa}.
\end{align}
We have arrived at the specification $S'_\text{teleport}$ which refines \cref{eq:orig teleport},
\begin{align*}
  S'_\text{teleport} = \qpenv{
    \begin{aligned}
      &\underbrace{ \qphole{\proj{\phi^+}_{qr} \ot \proj{\phi^+}_{ab}}{P} }_{S_\text{Alice}} ; \\
      &\qpselectS{q, a}{(x, y) \colon \underbrace{ \qphole{Q(x, y)}{\proj{\phi^+}_{br}} }_{S_{\text{Bob},x,y}}}
    \end{aligned}
  }
\end{align*}
which consists of three steps:
\begin{enumerate}
  \item Alice first applies some quantum program $S_\text{Alice}$ (which will be constructed to only act on her qubits).
  \item Alice measures her qubits~$q,a$ in the standard basis (without loss of generality) and obtains as outcomes two classical bits~$x,y$, which we imagine she communicates to Bob.
  \item Bob applies another quantum program $S_{\text{Bob},x,y}$ (which will be constructed to only act on his qubit) that is allowed to explicitly depend on the outcomes~$x$ and~$y$.
\end{enumerate}
This precisely captures the structure as well as functionality of a quantum teleportation protocol.

\paragraph{Construction}
Before starting the construction we first simplify the precondition~$Q(x,y)$ of Bob's program.
As the latter is run straight after Alice's measurements, which yielded outcomes~$x,y$, we know that Alice's qubits must be in state~$\ket{x,y}$.
This motivates and in fact implies that we may take
\begin{align}\label{eq:teleport form of Q}
  Q(x,y) = \proj{xy}_{qa} \ot R_{br}(x,y)
\end{align}
for certain predicates~$R_{br}(x,y)$ on qubits~$b,r$ that still need to be determined.

To construct Alice and Bob's programs, we make the straightforward guess that each applies some unitary, which in Bob's case may depend on the measurement outcomes~$x,y$.
Thus we refine
\begin{align*}
  S_\text{Alice} \totref \qpenv{\qpunitary{q, a}{V}} %\ruletag{H.unit}
  \qquad\text{and}\qquad
% \intertext{and}
  S_{\text{Bob},x,y} \totref \qpenv{\qpunitary{b}{U(x,y)}}, \ruletag{H.unit}
\end{align*}
where~$V$ is a two-qubit unitary and the~$U(x,y)$ are one-qubit unitaries that we still need to construct.
The second refinement is valid assuming $U_b(x,y) Q(x,y) U_b^\dagger(x,y) \imp \proj{\phi^+}_{br}$, which by \cref{eq:teleport form of Q} we can satisfy by picking
\[ R_{br}(x,y) = U^\dagger_b(x,y) \proj{\phi^+}_{br} U_b(x,y). \]
By \cref{eq:P via Qxy}, this in turn implies that
\begin{align}\label{eq:P via Uxy}
  P
% = \sum_{x, y} \proj{xy}_{qa} \ot R_{br}(x,y)
= \sum_{x, y} \proj{xy}_{qa} \ot U_b^\dagger(x,y) \proj{\phi^+}_{br} U_b(x,y).
\end{align}
The first refinement is valid if $V_{qa} \parens{ \proj{\phi^+}_{qr} \ot \proj{\phi^+}_{ab} } V_{qa}^\dagger \imp P$.
As the left-hand side is a pure state, \cref{cor:projection-predicates} shows that this condition is equivalent to
\begin{align}
\nonumber
  1
= \braket{ \phi^+_{qr} \ot \phi^+_{ab} | V_{qa}^\dagger P V_{qa} | \phi^+_{qr} \ot \phi^+_{ab}}
% \nonumber
&= \sum_{x,y} \abs{ \braket{ x_q \ot y_a \ot \phi^+_{br} |  U_b(x,y) V_{qa} | \phi^+_{qr} \ot \phi^+_{ab} } }^2 \\
\nonumber
&= \sum_{x,y} \abs{ \braket{ x_q \ot y_a \ot \phi^+_{br} |  U_b(x,y) V^T_{rb} | \phi^+_{qr} \ot \phi^+_{ab} } }^2 \\
\label{eq:tele final}
&= \frac14 \sum_{x,y} \abs{ \braket{ \phi^+_{br} | U_b(x,y) V^T_{rb} | x_r \ot y_b } }^2,
\end{align}
where we first used \cref{eq:P via Uxy} and then the identity~$V_{qa} \ket{\phi^+_{qr} \ot \phi^+_{ab}} = V^T_{rb} \ket{\phi^+_{qr} \ot \phi^+_{ab}}$,
known as the ``transpose trick'', which allows moving an arbitrary operator acting on qubits~$q,a$ to the other side of the maximally entangled states, that is, to act on qubits~$r,b$, if we replace the operator by its transpose (in the computational basis).
This well-known identity is easily verified by direct calculation.
The final step follows by observing that~$\bra{x_q} \ket{\phi^+_{qr}} = \ket{x_r}/\sqrt2$ and similarly~$\bra{y_a} \ket{\phi^+_{ab}} = \ket{x_b}/\sqrt2$.
Since all of the summands in \cref{eq:tele final} are at most one, they must all be equal to one.
In other words,
\begin{align*}
  V^T_{rb} \ket{xy}_{rb} \quad\text{and}\quad U_b^\dagger(x,y) \ket{\phi^+_{br}}
\end{align*}
must be the same states for all~$x,y\in\{0,1\}$ (up to irrelevant overall phases).
Note that the left-hand side states make up an orthonormal basis, while the right-hand states are all maximally entangled (since they obtained by applying a unitary to one of the qubits of~$\ket{\phi^+}$).
It follows that we should pick~$V^T_{rb}$ to be a unitary that maps the standard basis to a basis of maximally entangled states.
And this is also sufficient since any two maximally entangled states differ by a unitary on either of the qubits.
As is well known, the Bell basis consists of maximally entangled states and it can be prepared by the unitary~$V^T_{rb} = \text{CNOT}_{rb} H_b$.
Thus we take
\[ V = (H \ot I) \text{CNOT}. \]
Finally, we note that the Bell states can be obtained from the standard maximally entangled state as~$V^T_{rb} \ket{yx}_{rb} = Z_b^y X_b^x \ket{\phi^+}_{rb}$, where~$X$ and~$Z$ denote the Pauli~$X$ and~$Z$ matrices.
Thus Bob's unitaries should be
\[ U(x,y) = X^x Z^y. \]
Altogether, we have constructed the following program, which is nothing but the standard protocol for quantum teleportation:
\[
S'_\text{teleport} \totref^*
\qpenv{
\begin{aligned}
&\qpunitary{q, a}{(H \ot I) \text{CNOT}} ; \\
% &\qpselectS{(q, a)}{\\
% &\quad 00 \colon \qpunitary{b}{I} \\
% &\quad 01 \colon \qpunitary{b}{Z} \\
% &\quad 10 \colon \qpunitary{b}{X} \\
% &\quad 11 \colon \qpunitary{b}{XZ} \\&
% }
&\qpselectS{(q, a)}{
00 \colon \qpunitary{b}{I};
01 \colon \qpunitary{b}{Z};
10 \colon \qpunitary{b}{X};
11 \colon \qpunitary{b}{XZ}
}
\end{aligned}
}
\]
As it was obtained by refinement, it satisfies the specification by construction.
% Of course we can also replace the first line by a sequence of two statements~$\qpenv{\qpunitary{q, a}{\text{CNOT}}}$ and~$\qpenv{\qpunitary{q}{H}}$ without changing the program's semantics.

%------------------------------------------------------------------------------
\subsection{Quantum Search}
\label{subsec:examples:qsearch}
%------------------------------------------------------------------------------

We consider the following search problem~\cite{nielsen2010quantum,grover1996}.
Given query access to a Boolean function or ``database''~$f\colon\{0,1\}^n \to \{0,1\}$,
we wish to find a bitstring~$x \in \{0, 1\}^n$ such that $f(x) = 1$.
Such an~$x$ is often called a ``solution'' or ``marked element''.
In the quantum setting, we are given query access to~$f$ via the following \emph{standard quantum oracle} unitary,
$\cO_f \ket{x}_{\q}\ket{y}_{a} = \ket{x}_{\q}\ket{y \oplus f(x)}_a$,
or by the following \emph{phase oracle} unitary
\begin{align}\label{eq:phase oracle}
  P_f\ket{x}_\q = (-1)^{f(x)} \ket{x}_\q,
\end{align}
which can be obtained from the former in a straightforward fashion.
Here, $\q$ is a quantum variable consisting of~$n$ qubits.
Let us define $N = 2^n$ as the size of the search space,
and $T = \abs{\{x : f(x) = 1\}}$ as the number of solutions.
In the example constructions below,
we will assume knowledge of this number of solutions.
In the following, we will first present a specification of the search problem, and then construct \emph{two} different programs that satisfy the specification by construction: a simple algorithm based on random sampling and Grover's celebrated quantum search algorithm~\cite{grover1996}.

\paragraph{Specification}
The search problem can be specified as follows:
\[
S_\text{search}(p) = \qpenv{\qphole{pI}{\sum_{x : f(x)=1} \proj{x}_{\vec q}}},
\]
It states that if we measure the state after program execution, we obtain a solution~$x$ with probability at least~$p$.
Thus, $p$ is the probability of success of the search algorithm.
In our constructions below, we treat~$p$ as a parameter that will naturally be selected during the refinement process.
Note that whatever the value of~$p$, such an algorithm (if it terminates) can always be amplified or ``boosted'' to any desired success probability by repeating it sufficiently often until it finds a solution.
We discuss this in \cref{subsec:examples:boosting} below and propose refinement rules to automate this reasoning.

%------------------------------------------------------------------------------
\paragraph{Construction I: Random Sampling}
One can solve the search problem by sampling~$x\in\{0,1\}^n$ uniformly at random.
Clearly, this succeeds with probability $T/N$.
One way to achieve this by a quantum program is by preparing the uniform superposition~$\ket U = \frac1{\sqrt N}\sum_{x \in \{0,1\}^n} \ket x$, since measuring this state in the standard basis will yield a uniformly random~$x\in\{0,1\}^n$.
We can confirm that this construction works with the aforementioned success probability, by refining the specification into a program that prepares the state~$\ket U$.
This can be done by first initializing all qubits in the zero state and then applying Hadamard gates~$H$ (\cref{eq:hadamard}) to all of the qubits:
\begin{align*}
  S_\text{search}(p)
  %&= \qpenv{\qphole{pI}{\sum_{x : f(x)=1} \proj{x}_{\vec q}}} \\
  &\totref \qpenv{
    \qphole{pI}{R} ;
    \qphole{R}{\sum_{x : f(x)=1} \proj{x}_{\vec q}}
  } \ruletag{H.seq} \\
  &\totref \qpenv{
    \qpinit{\q} ;
    \qphole{R}{\sum_{x : f(x)=1} \proj{x}_{\vec q}}
  } \ruletag{H.init} \\
  &\totref \qpenv{
    \qpinit{\q} ;
    \qpunitary{\q}{H^{\ot n}}
  } \ruletag{H.unit}
\end{align*}
To apply \ruleref{H.init}, we need that~$pI \imp \sum_x \ketbra{x}{0}_\q R \ketbra{0}{x}_\q$, that is, $p \leq \braket{0 | R | 0}$.
To apply \ruleref{H.unit}, we can choose
  \[ R = H^{\ot n} \parens*{\sum_{x : f(x)=1} \proj{x}_{\vec q}} H^{\ot n}. \]
Thus the above refinements are valid if
\begin{align*}
  p
\leq \braket{0 | R | 0}
= \braket{U | \sum_{x : f(x)=1} \proj{x}_{\vec q} | U}
= \frac T N.
\end{align*}
Thus we should pick~$p:=T/N$ to maximize the success probability.
Altogether, we have constructed a search algorithm that succeeds with probability~$p=T/N$.
As mentioned earlier and will be discussed in detail in \cref{subsec:examples:boosting}, by repeating the above~$O(N/T)$ times until we find a solution, we can obtain a program that solves the search problem with any desired constant probability of success (say, $p=2/3$) at a cost of~$O(N/T)$ queries.

%------------------------------------------------------------------------------
\paragraph{Construction II: Grover Search}
\citet{grover1996} proposed a quantum algorithm for the search problem which gives a quadratic speedup over random sampling.
In the following, we will re-construct this algorithm by making natural choices using one key idea at a time.

\paragraph*{Step~1}
The first observation is that while the uniform superposition~$\ket U$ over all bitstrings is easy to prepare, what we are really after is the uniform superposition of all \emph{solutions}, that is, the ``good'' state $\ket G = \frac1{\sqrt T} \sum_{x ~\text{s.t.}~ f(x) = 1} \ket x$.
Indeed, measuring~$\ket G$ will yield a solution with probability one.
We can formalize this idea by strengthening the postcondition of the specification:
\[
S_\text{search}(p) \totref \qpenv{\qphole{pI}{\proj G}} =: {S_\text{good}(p)}, \ruletag{H.sw}
\]
We may apply \ruleref{H.sw} since the condition $\proj G \imp \sum_{x ~\text{s.t.}~ f(x)=1} \proj x$ is satisfied.
Indeed, the good state~$\ket G$ is clearly contained in the span of the basis states~$\ket x$ corresponding to solutions~$x$.

\paragraph*{Step~2}
The uniform state~$\ket U$ (which is easy to prepare, but not very useful) and the good state~$\ket G$ (which solves the problem, but is a priori unclear how to prepare) span a two-dimensional subspace of the exponentially large Hilbert space.
The key idea (which has no classical counterpart) then is to try to \emph{rotate} the state~$\ket U$ onto~$\ket G$ in this two-dimensional subspace.

To realize this idea, we first define an orthonormal basis of the two-dimensional subspace by picking~$\ket G$ and a vector orthogonal to it, namely the ``bad state''~$\ket{B} = \frac1{\sqrt{N-T}} \sum_{x ~\text{s.t.}~ f(x) = 0} \ket x$.
At any stage of the program, we would like the state of $\vec q$ to be of the form
\[ \ket{\theta} = \cos\theta \ket B + \sin\theta \ket G \]
for some angle~$\theta$, which will serve as a loop variant in the following.
The plan is now to prepare the uniform state, which has angle~$\upsilon = \arcsin{\sqrt{T/N}} \in [0,\pi/2]$ as it can be written as~$\ket U = \sqrt{(N-T)/N} \ket B + \sqrt{T/N} \ket G$, and then rotate it repeatedly by some angle~$\delta>0$ towards the good state~$\ket G$, which is at angle~$\frac\pi2$.
We can formalize this by the following refinements:
\begin{align*}
S_\text{good}(p)
% = \qpenv{\qphole{pI}{\proj{\frac\pi2}}} \\
&\totref \qpenv{\qphole{pI}{\proj{\theta_0}} ; \qphole{\proj{\theta_0}}{\proj{\frac\pi2}}} \ruletag{H.seq} \\
&\totref \qpenv{
    \underbrace{\qphole{pI}{\proj{\theta_0}}}_{S_\text{init}} ;
    \qprepeat{r}{
      \underbrace{\qphole{\proj{\theta_j}}{\proj{\theta_{j+1}}}}_{S_\text{rotate}(\delta)}
    }
  }
  \ruletag{H.repeat}
\end{align*}
To apply \ruleref{H.repeat}, we choose $\theta_r = \frac\pi2$ and $\theta_{j + 1} = \theta_j + \delta$ for some arbitary rotation angle~$\delta$ and number of rotations~$r$ that we will determine later.
Hence~$\theta_0 = \frac\pi2 - r \delta$.

We first construct the $S_\text{init}$ program.
Following the plan, we prepare the uniform superposition~$\ket U = \ket\upsilon$, which we already know how to do from above:
\begin{align*}
  S_\text{init}
  &\totref \qpenv{ \qphole{pI}{R} ; \qphole{R}{\proj{\theta_0}} }\ruletag{H.seq} \\
  &\totref \qpenv{ \qpinit{\q} ; \qphole{R}{\proj{\theta_0}} } \ruletag{H.init} \\
  &\totref \qpenv{ \qpinit{\q} ; \qpunitary{\q}{H^{\ot n}} }\ruletag{H.unit}
\end{align*}
Since the postcondition is different from the above, we also need to choose~$R$ differently, but we can follow the same reasoning.
In order to apply \ruleref{H.init}, we need that~$p \leq \braket{0 | R | 0}$, and to apply~\ruleref{H.unit} we can choose $R = H^{\ot n} \proj{\theta_0} H^{\ot n}$.
Together, we find that the maximum success probability for which the above refinements are valid is given by
\begin{align*}
p
:= \braket{0 | R | 0}
= \abs{\braket{0 | H^{\ot n} | \theta_0}}^2
= \abs{\braket{\upsilon | \theta_0}}^2
% = \sin^2\parens*{r\delta + \arcsin\parens*{\sqrt{\frac T N}}}
% = \cos^2\parens*{\frac\pi2 - r\delta - \arcsin\parens*{\sqrt{\frac T N}} }.
= \cos^2\parens*{ \frac\pi2 - r\delta - \upsilon }.
\end{align*}
To maximize this probability, we should further choose~$r$ such that the right-hand side is maximized.
We will pick~$r$ such that the angle in the cosine is closest to~$0$:
\begin{align}\label{eq:r}
r := \bracks*{\frac{\frac\pi2 - \upsilon}{\delta}},
\end{align}
where $\bracks{\cdot}$ rounds to the nearest integer.
Clearly, $p \ge \cos^2(\delta/2)$.

\paragraph*{Step~3}
We still need to construct the program~$S_\text{rotate}(\delta)$ for some rotation angle~$\delta$.
To this end, we first observe that it suffices to construct a program that satisfies the stronger specification:
\begin{align*}
  S'_\text{rotate}(\delta) = \qphole{\proj{\theta}}{\proj{\theta+\delta}}.
\end{align*}
Indeed, programs satisfying this specification rotate \emph{all} states~$\ket\theta$ in the two-dimensional subspace by~$\delta$, as opposed just the states~$\ket{\theta_j}$ for~$j\in\{0,1,\dots,N-1\}$.

How can we obtain such a rotation?
Observe that the quantum phase oracle~$P_f$ in \cref{eq:phase oracle} is a \emph{reflection} about the vector $\ket B$, as it maps
$P_f \ket\theta = \ket{-\theta}$.
Now, we know that two reflections make a rotation.
For our second reflection we simply pick some known state (independent of the instance of the search problem) to reflect about.
A natural choice is $\ket U$, since the corresponding reflection~$2\proj{U} - I$ can be efficiently implemented using~$O(n)$ gates.
Thus we introduce these two reflections in sequence and determine the rotation angle from the conditions of the refinements:
\begin{align*}
  S'_\text{rotate}(\delta)
  &\totref \qpenv{
    \qphole{\proj{\theta}}{\proj{-\theta}}
    ;
    \qphole{\proj{-\theta}}{\proj{\theta+\delta}}
  } \ruletag{H.seq} \\
  &\totref \qpenv{
    \qpunitary{\vec q}{P_f}
    ;
    \qphole{\proj{-\theta}}{\proj{\theta+\delta}}
  } \ruletag{H.unit} \\
  &\totref \qpenv{
    \qpunitary{\vec q}{P_f}
    ;
    \qpunitary{\vec q}{(2\proj{U} - I)}
  } \ruletag{H.unit}
\end{align*}
The first application of \ruleref{H.unit} is correct by our choice of intermediate condition, but for the second one we need that
\begin{align*}
  \parens*{ 2\proj{U} - I } \proj{-\theta} \parens*{ 2\proj{U} - I }
 \imp \proj{\theta+\delta}.
\end{align*}
Since~$\ket U$ is at angle~$\upsilon$, reflecting about it sends~$\ket{-\theta} = \ket{\upsilon-(\theta+\upsilon)}$ to~$\ket{\upsilon+(\theta+\upsilon)}=\ket{\theta+2\upsilon}$.
Thus the above refinements are valid if we choose~$\delta = 2\upsilon$ as the rotation angle, where we recall that~$\upsilon = \arcsin{\sqrt{T/N}}$.
If we plug this back into \cref{eq:r} we find that the number of rotations is
\begin{align*}
  r
= \bracks*{\frac{\frac\pi2 - \upsilon}{2\upsilon}}
= \bracks*{\frac\pi{4\arcsin{\sqrt{T/N}}} - \frac12}
= O\parens*{\sqrt{\frac N T}}.
\end{align*}
Moreover, the success probability can be lower bounded as
\begin{align*}
  p
\geq \cos^2(\delta/2)
= 1 - \sin^2(\upsilon)
% = 1 - \frac TN.
= 1 - T/N.
\end{align*}
Thus we obtain the following quantum program:
\[
  S_\text{search}\parens*{1 - \frac T N}
  % = \qpenv{\qphole{\parens*{1 - \frac{T}{N}} I}{\sum_{x ~\text{s.t.}~ f(x)=1}\proj x}}
  \totref^*
  \qpenv{
    \begin{aligned}
      &\qpinit{\vec q} ;~ 
      \qpunitary{\vec q}{H^{\ot n}} ;~ \\
      &\qprepeat{
          \bracks*{
            \tfrac
              {\pi}
              {4\arcsin{\sqrt{T/N}}}
            - \tfrac12
          }
        }{\\
        &\quad \qpunitary{\vec q}{P_f};~ 
        \qpunitary{\vec q}{(2\proj{U} - I)} \\
      &}
    \end{aligned}
  }
\]
As we have constructed it by refining the initial specification for the search problem, it satisfies the specification by construction.
It succeeds with probability~$p \geq 1 - T/N$ and uses~$O(\sqrt{N/T})$ queries to the quantum oracle.
This is in fact Grover's algorithm~\cite{grover1996}.

%------------------------------------------------------------------------------
\subsection{Boosting Success Probabilities}
\label{subsec:examples:boosting}
%------------------------------------------------------------------------------
In this section, we derive two refinement rules that formalize useful and widely used patterns (see, e.g., \cite{lim2005repeatuntilsuccess,paetznick2014repeatuntilsuccess}).
To motivate it, recall that in the preceding example, we constructed two quantum programs that succeed with some probability.
We modeled this by a specification of the form~$\qpenv{\qphole{\eps I}{Q}}$ for some $\eps>0$ (we now write~$\eps$ rather than~$p$ because the discussion that follows is most relevant when~$\eps$ is a small probability).
Indeed, a Hoare triple $\hoare{\eps I}{S}{Q}$ is totally correct if the program~$S$ terminates and the postcondition~$Q$ holds with probability at least~$\eps$.
% This can be seen from the definition of total correctness:
% $p \le \satisfies{\Bracks{S}(\rho)}{P}$ for every state~$\rho$.
% Similarly, the triple is partially correct if with probability at least~$p$ either the postcondition holds or the program does not terminate.
We can amplify or ``boost'' the success probability of such a program~$S$ arbitrarily by simply repeating it until the postcondition holds, provided (i)~the program~$S$ terminates almost surely (so that we keep repeating) and (ii)~the postcondition is given by a projection (so that measuring it does not impact its expectation).
To incorporate the termination requirement we can consider the multiple specification~$\qpenv{\qphole{\eps I,I}{Q,I}}$.

We first give a rule that reduces the construction of a program that succeeds with some probability~$p\in(0,1)$ to the construction of a program that succeeds with some smaller probability~$\eps\in(0,p)$:
\begin{enumerate}[align=left]
\ruleitem{H.boostRep}
  $\qpenv{\qphole{pI, I}{Q_\q, I}}
  \totref^*
  \qpenv{
  \begin{aligned}
      &\qprepeat{\ceil{\log_{1-\eps} (1 - p)}}{ \\
      &\quad \qpif{Q^\perp}{\q}{
        \qphole{\eps Q^\perp_\q, Q^\perp_\q}{Q_\q, I}
      }
      \\&}
  \end{aligned}
  }
  $, \\
  for any projection~$Q$ and any~$\eps\in(0,p)$,
  where $Q^\perp = I - Q$
\end{enumerate}
Second, we give a rule to reduce the construction of programs that succeed with probability one to ones that succeed with some finite probability~$\eps\in(0,1)$, by repeating the program until it succeeds: % an \emph{unbounded} number of times until it succeeds:
\begin{enumerate}[align=left]
\ruleitem{H.boostWhile}
  $\qpenv{\qphole{I}{Q_\q}}
  \totref^*
  \qpenv{
      \qpwhile{Q^\perp}{\q}{
      \qpenv{\qphole{\eps Q^\perp_\q, Q^\perp_\q}{Q_\q, I}}
      }
  }
  $, \\
  for any projection~$Q$ and any~$\eps \in (0,1)$,
  where $Q^\perp = I - Q$
\end{enumerate}

\begin{restatable}[Boosting success probability]{theorem}{thmboostingsuccessprobability}\label{thm:boosting}
The relations \ruleref{H.boostRep}, \ruleref{H.boostWhile} hold.
\end{restatable}

The proof can be found in \cref{app:examples:boosting}.
We emphasize that both rules are totally correct.

%------------------------------------------------------------------------------
\subsection{Quantum Fourier Transform}
\label{subsec:examples:qft}
%------------------------------------------------------------------------------
The Quantum Fourier Transform (QFT) is widely used in many algorithms, such as Shor's factoring algorithm~\cite{shor1994} and quantum phase estimation~\cite{kitaev1995qpe}.
For $n$ qubits, it computes the following unitary:
\[
  \QFT_n = \frac{1}{\sqrt{2^n}} \sum_{x, y \in \{0, \ldots, 2^n - 1\}} \omega_n^{xy}  \ketbra{x}{y}
\]
where $\omega_n = \exp\parens*{\frac{2 \pi i}{2^n}}$.
Here, the $n$-bit numbers $x, y$ are identified with big-endian bitstrings,
that is, $\ket{x} = \ket{x_1} \ot \ldots \ot \ket{x_n}$ where $x = \sum_{j = 1}^n x_j 2^{n - j}$.
The key observation is that $\omega_k^2 = \omega_{k - 1}$, which naturally motivates a recursive approach.
Therefore, we will attempt to construct a program implementing $\QFT_n$ using a program for $\QFT_{n-1}$, and so.
In this section only, we will use the shorthand notation $\ket{\Psi}$ to represent the predicate~$\proj{\Psi}$, which is a projection for any unit vector~$\ket{\Psi}$.

\paragraph{Specification}
As in the teleportation example~(\cref{subsec:examples:teleportation}), we can fully specify the $\QFT$ by considering its action on half of a maximally entangled input state.
To allow us to recurse on the number of qubits, we will define a specification for $\QFT_k$ for each $k$ from $1$ to $n$.
To this end let $S_k$ denote the specification for applying $\QFT_k$ on the first $k$ qubits of an $n$-qubit quantum variable~$\q$:
\[
  S_k =
   \qpenv{\qphole
    {\ket{\Phi^+_n}_{\vecr, \q}}
    {(\QFT_k)_{q_1, \ldots, q_k} \ket{\Phi^+_n}_{\vecr, \q}}},
\]
where $\ket{\Phi^+_n}_{\vecr, \q} = \bigotimes_{j = 1}^n \ket{\phi^+}_{r_j, q_j} = \frac1{\sqrt{2^n}} \sum_{x = 0}^{2^n-1} \ket{x}_{\vecr}\ket{x}_{\q}$ (with $\ket{\phi^+}$ as in \cref{eq:bell-state}) is a maximally entangled state between~$\q$ and an additional $n$-qubit variable~$\vecr$ that will not be used in the program.
Any program obtained by refining~$S_k$ implements the QFT on the first~$k$ qubits of~$\q$ while acting trivially on the rest.
In particular, $S_n$ specifies the $n$-qubit QFT.

\paragraph{Construction}
We will show that a quantum program for~$S_n$ can be constructed recursively, starting with~$S_1$ and subsequently constructing a program for~$S_k$ from one for~$S_{k-1}$ for any $k=2,\dots,n$.

\paragraph{Base Case ($k = 1$)}
As $\omega_1 = -1$, we can see that $\QFT_1 = \frac1{\sqrt2} \sum_{x, y \in \{0, 1\}} (-1)^{xy}\ketbra{x}{y} = H$ straight from the definition,
i.e., the $1$-qubit QFT is the Hadamard gate.
Therefore, we can refine:
\[
  S_1 \totref \qpenv{\qpunitary{q_1}{H}} \ruletag{H.unit}
\]
We can apply \ruleref{H.unit} because $H_{q_1}^\dagger \parens*{(\QFT_1)_{q_1} \proj{\Phi^+_n}_{\vecr, \q} (\QFT_1)_{q_1}^\dagger} H_{q_1} = \proj{\Phi^+_n}_{\vecr, \q}$ holds.

\paragraph{General Case ($k > 1$)}
We refine $S_k$ by using \ruleref{H.seq} into a sequence of two holes, with intermediate condition being the postcondition of $S_{k-1}$.
Then the first hole in the sequence matches $S_{k-1}$:
\begin{align*}
  S_k
  &\totref
    \qpenv{
      \qpabsprog{S_{k - 1}} ;
      \qphole{(\QFT_{k-1})_{q_1, \ldots, q_{k-1}} \ket{\Phi^+_n}_{\vecr, \q}}{(\QFT_k)_{q_1,\ldots,q_k} \ket{\Phi^+_n}_{\vecr, \q}}}
    \ruletag{H.seq}
\end{align*}
The specification $S_{k-1}$ can be refined recursively,
\newcommand{\ketpsirest}{\ket{\Psi^\text{rest}}}%
but we still need to refine the right-hand side hole.
Let us denote the states defining the pre- and postcondition by $\ket{\Psi_{k-1}} := (\QFT_{k-1})_{q_1, \ldots, q_{k-1}} \ket{\Phi^+_n}_{\vecr, \q}$ and $\ket{\Psi_k} := (\QFT_k)_{q_1,\ldots,q_k} \ket{\Phi^+_n}_{\vecr, \q}$, respectively.
Because the program should only act on the first $k$ qubits, we expand these accordingly:
\begin{align*}
  \ket{\Psi_{k-1}} &= \parens*{
        \frac1{2^{k-1} \sqrt2}
        \sum_{\tilx,\tily=0}^{2^{k-1}-1}
        \sum_{x_k \in\{0, 1\}}
        \omega_{k-1}^{\tilx \tily}
        \ket{\tilx}_{r_1,\ldots,r_{k-1}}
        \ket{x_k}_{r_k}
        \ket{\tily}_{q_1,\ldots,q_{k-1}}
        \ket{x_k}_{q_k}
        } \ot \ketpsirest, \\
  \ket{\Psi_{k}} &= \parens*{
        \frac1{2^k}
        \sum_{x,y=0}^{2^k-1}
        \omega_{k}^{xy}
        \ket{x}_{r_1,\ldots,r_k}
        \ket{y}_{q_1,\ldots,q_k}
        } \ot \ketpsirest,
\end{align*}
where $\ketpsirest$ % = \ket{\Phi^+_{n-k}}_{r_{k+1},\ldots,r_n,q_{k+1}\ldots,q_n}$
is the (unchanged) state on the last $n - k$ qubits of $\vecr, \q$.
To compare the state~$\ket{\Psi_{k}}$ with~$\ket{\Psi_{k-1}}$, let us write $x = 2\tilx + x_k$.
Then we can also simplify
\begin{align*}
\omega_k^{xy}
  = \omega_k^{2\tilx y} \omega_k^{x_k y}
  = \omega_{k-1}^{\tilx y} \omega_k^{y x_k}.
\end{align*}
Now as $\omega_{k-1}^{2^{k-1}} = 1$, the most significant bit of $y$ does not affect the first part of the above term.
Therefore, it is natural to write the index~$y$ as $y = y_1 2^{k-1} + \tily$, so that we can further simplify
\[
  \omega_{k-1}^{\tilx y} \omega_k^{y x_k}
  = \omega_{k-1}^{\tilx \tily} (-1)^{y_1 x_k} \omega_k^{\tily x_k}.
\]
Altogether we find that the postcondition is described by the state
\begin{align*}
  \ket{\Psi_{k}} = \parens*{
        \frac1{2^k}
        \sum_{\tilx,\tily=0}^{2^{k-1}-1}
        \sum_{x_k, y_1 \in \{0, 1\}}
        \omega_{k-1}^{\tilx \tily} (-1)^{y_1 x_k} \omega_k^{\tily x_k}
        \ket{\tilx}_{r_1,\ldots,r_{k-1}}
        \ket{x_k}_{r_k}
        \ket{y_1}_{q_1}
        \ket{\tily}_{q_2,\ldots,q_k}
        } \ot \ketpsirest.
\end{align*}
The state $\ket{\Psi_{k}}$ looks quite similar to $\ket{\Psi_{k-1}}$ but there are some key differences.
For one, the index~$\tily$ refers qubits~$q_2,\dots,q_k$ rather than~$q_1,\dots,q_{k-1}$ and hence $q_1$ plays a distinguished role rather than~$q_k$.
Thus a natural first step is to move the $k$-th qubit to the front by applying a suitable sequence of swaps.
That is, we use \ruleref{H.seq} and \ruleref{H.unit} to refine
\[
  \qpenv{\qphole{\ket{\Psi_{k-1}}}{\Psi_{k}}}
  \totref^*
  \qpenv{
    \qpunitary{q_k, q_{k-1}}{\SWAP};
    \ldots;
    \qpunitary{q_2, q_1}{\SWAP};
    \qphole{\ket{\Psi'}}{\ket{\Psi_{k}}}
  },
  % \tag{\ref{rule:H.seq}, \ref{rule:H.unit}}
\]
which is allowed if we pick the intermediate condition given by the state
\begin{align*}
  \ket{\Psi'} &= \parens*{
        \frac1{2^{k-1} \sqrt2}
        \sum_{\tilx,\tily=0}^{2^{k-1}-1}
        \sum_{x_k \in\{0, 1\}}
        \omega_{k-1}^{\tilx \tily}
        \ket{\tilx}_{r_1,\ldots,r_{k-1}}
        \ket{x_k}_{r_k}
        \ket{x_k}_{q_1}
        \ket{\tily}_{q_2,\ldots,q_k}
        } \ot \ketpsirest.
\end{align*}
Now the qubits are in the right place but we see that in $\ket{\Psi_k}$ there is an additional relative phase~$\omega_k^{\tily x_k}$.
Because
$\omega_k^{\tily x_k} = \parens{\omega_k^{y_k} \omega_{k-1}^{y_{k-1}} \ldots \omega_2^{y_2}}^{x_k}$
we are led to refining the second hole by a series of phase gates, one on each qubit $q_2, \ldots, q_k$, and each controlled by the value~$x_k$ that is stored in qubit~$q_1$:
\[
  \qpenv{\qphole{\ket{\Psi'}}{\ket{\Psi_{k}}}}
  \totref^*
  \qpenv{
    \qpunitary{q_1, q_{2}}{\CRz_2};
    \ldots;
    \qpunitary{q_1, q_k}{\CRz_k};
    \qphole{\ket{\Psi''}}{\ket{\Psi_{k}}}
  }
  % \ruletag{H.unit}
\]
where $\Rz_k = \begin{psmallmatrix} 1 & 0 \\ 0 & \omega_k \end{psmallmatrix}$ and $\CRz_k$ denotes the corresponding controlled gate.
The above refinement is valid if pick the intermediate condition given by the state
\begin{align*}
  \ket{\Psi''} &= \parens*{
        \frac1{2^{k-1} \sqrt2}
        \sum_{\tilx,\tily=0}^{2^{k-1}-1}
        \sum_{x_k \in\{0, 1\}}
        \omega_{k-1}^{\tilx \tily}
        \omega_{k}^{\tily x_k}
        \ket{\tilx}_{r_1,\ldots,r_{k-1}}
        \ket{x_k}_{r_k}
        \ket{x_k}_{q_1}
        \ket{\tily}_{q_2,\ldots,q_k}
        } \ot \ketpsirest
\end{align*}
This is almost identical to~$\ket{\Psi_k}$, except that~$q_1$ is in a basis state rather than a suitable superposition.
% Compared to~$\ket{\Psi_k}$, the qubit~$q_1$ is in state~$\ket{x_k}$ rather than in a suitable superposition of basis states~$\ket{y_1}$ with a relative phase of $(-1)^{y_1 x_k}$.
As we have $H \ket{x_k} = \frac1{\sqrt2} \sum_{y_1} (-1)^{x_k y_1} \ket{y_1}$ by definition, this is easily fixed by a Hadamard gate:
\[
  \qpenv{\qphole{\ket{\Psi''}}{\ket{\Psi_{k}}}} \totref \qpenv{\qpunitary{q_1}{H}}
  \ruletag{H.unit}.
\]
This concludes the construction of the quantum Fourier transform.

%==============================================================================
\section{Conclusion and Outlook}\label{sec:conclusion}
%==============================================================================

In this work, we proposed Quantum Correctness by Construction (QbC), an approach for constructing quantum programs that are guaranteed to be correct by construction.
To this end, we extended a quantum while language with a construct called \emph{holes}, which represent yet-to-be-constructed subprograms that carry a precondition and a postcondition.
We presented refinement rules that iteratively refine such quantum programs and proved that these rules are sound and complete:
every program is guaranteed to satisfy the specification it was constructed from, and every correct program can always be constructed from the specification.
Finally, we demonstrated the QbC approach by constructing quantum programs for some idiomatic problems, starting from their natural specification.
We found that in these examples, QbC naturally suggested how to derive program details and highlighted key design choices that had to be made along the way.
We take these findings to suggest that QbC could play a meaningful role in supporting the design of quantum algorithms, their taxonomization, and the construction and verification of larger quantum software.
We now describe some promising directions for future research to further pave the way in this direction and conclude with a perspective on the role of automation in algorithm development.

\paragraph{Future Directions: Theory}
A natural and interesting direction would be to extend the QbC methodology to other settings and non-functional properties, such as by building on the expected-runtime calculus introduced by~\citet{Liu2022qweakest} to construct programs that are \emph{efficient by construction}.
Another direction would be to extend the QbC approach to other quantum programming languages.
While the \emph{quantum while language} used in our paper is well-understood to provide a clean theoretical model, it is often cumbersome to express complex quantum programs in it.
It would therefore be desirable to extend QbC to a more expressive language, which might include both classical and quantum variables~\cite{feng2021quantum}, oracles and subroutines, quantum data structures, and so forth.
It would also be highly interesting to identify further refinement rules that encode high-level reasoning and design patterns that are commonly used in quantum algorithms, and extend the language to natively support these operations.
For example, quantum amplitude amplification~\cite{brassard2002quantum}, which generalizes Grover's algorithm to offer a quantum speedup for boosting the success probability of a subroutine that improves over naive repetition (cf.\ \cref{subsec:examples:qsearch,subsec:examples:boosting}) and is widely used.
It would also be interesting to devise hybrid approaches that combine both post-hoc (Hoare or weakest-precondition logic) and by-construction (QbC) reasoning, which can be useful particularly when constructing larger and more complex programs~\cite{CbCandPhV2016}.

\paragraph{Future Directions: Implementation and Mechanization}
Our framework and results are agnostic to the choice of assertion language, in the interest of generality.
But committing to a concrete assertion language is an important choice for implementations.
From our examples and refinement rules, we find that it is convenient to consider projections scaled by scalar values and finite linear combinations thereof, expressed in Dirac notation.
For instance, a web-based prototype proposed recently uses a simple assertion language based on Dirac notation for a fixed number of qubits~\cite{seng2024quantum}.
The conditions of refinement are checked using a decision procedure on finite-sized complex matrices to verify the L\"owner order.
To further mechanize QbC and in particular to handle general programs when the number of qubits is not fixed, we believe a natural and ambitious future work would be to integrate QbC with a proof assistant such as Coq or Lean to handle proof obligations of side conditions.
Notable prior work includes CoqQ~\cite{coqq2023} and QWIRE~\cite{Rand2018qwirepractice}, which formalize reasoning about quantum programs in CoQ, and the recent decision procedure for Dirac notation~\cite{xu2025Dirac}.

\paragraph{Outlook: Algorithm Development}
Our view is that algorithm design is a creative process that often requires insights that are difficult to obtain purely by automation.
We believe that a by-construction approach can help algorithm designers focus on this creative aspect, one key insight at a time, to develop algorithms that are ensured to be formally correct, without being bogged down by small details, as the framework guides the refinement and gives a principled way to generate side-conditions that need to be checked per refinement step.
Several well-known quantum algorithms have been post-hoc verified using the quantum Hoare logic proof system.
It is also possible to construct these same programs using QbC, at roughly the same complexity.
Compared with post-hoc verification, this can provide additional benefits: one does not have to decide on all program details a priori, but can do so during the refinement process.
This can lead to different choices, resulting in different programs and trade-offs.
We showcased this in the search example where we derived two programs from the same specification, by making a different choice at a key step.
Similar examples are known in the classical CbC literature.
This may also open up the possibility of discovering alternate and cleaner implementations to existing algorithms.
However, we emphasize that we do not see by-construction and post-hoc approaches as mutually exclusive, but rather as complementing each other with each having its role.
The most natural approach may well be a hybrid approach, as we propose to explore in future work above.

%==============================================================================
\section{Data-Availability Statement} % https://2025.splashcon.org/track/OOPSLA#data-availability-statement
%==============================================================================
This paper proposes the %neccessary
theoretical foundations for a correctness-by-construction approach for quantum programs.
We do not provide an artifact, but note that \cite{seng2024quantum} gives a web-based prototype.

%==============================================================================
\begin{acks}
We thank Gilles Barthe and Bruce Watson for fruitful discussions on the subject of this work, and the anonymous referees for valuable feedback on earlier versions of this manuscript.
All authors acknowledge support by the BMBF (QuBRA, 13N16135 \& 13N16303; QuSol, 13N17173 \& 13N17170).
IS also acknowledges support by the BMWK (ProvideQ, 01MQ22006F).
MW also acknowledges support by the European Union (ERC, SYMOPTIC, 101040907), by the Deutsche Forschungsgemeinschaft (DFG, German Research Foundation) under Germany's Excellence Strategy - EXC\ 2092\ CASA - 390781972, and by the Dutch Research Council (NWO grant OCENW.KLEIN.267).
\end{acks}
%==============================================================================

%==============================================================================
\bibliographystyle{ACM-Reference-Format}
\bibliography{qbc}
%==============================================================================

%==============================================================================
\makeatletter
\par\bigskip\noindent{\small\normalfont\@received\par}
\def\@received{}
\makeatother
%==============================================================================

%==============================================================================
%% Appendix
\clearpage
% \onecolumn
\appendix
\section*{Appendix}
%==============================================================================

The supplementary material below contains the proofs for our results for both partial and total correctness,
and a pedagogical example on constructing correct quantum programs using QbC.

%==============================================================================
\section{Soundness of Refinement}
\label{app:soundness}
%==============================================================================
In this appendix we present the proofs of \cref{thm:soundness-partial} and \cref{thm:soundness-total}.
Before giving the proof, we first analyze the notion of how a program with multiple holes can be filled.
Consider the following intermediate program obtained by refinining to a sequence:
\[
  \qpenv{\qphole{P}{Q}} \totref \qpenv{\qphole{P}{R} ; \qphole{R}{Q}} \ruletag{H.seq}
\]
The only way to refine this program is by picking a hole and refining it individually.
The same is true more generally for any composite program (sequence, \textbf{repeat}, \textbf{case}, \textbf{while}): if one refines such a program, then the resulting program has the same outer structure; only its subprograms get refined.
We prove this in the following lemma.

\begin{lemma}[Refining programs containing holes]
\label{lem:chain_split}
For the relation~$\refine$ being either $\parref$ or $\totref$, we have the following for every $k \in \N$:
\begin{enumerate}
  \item
    If $\qpenv{\qpabsprog{S'_1};\qpabsprog{S'_2}} \refine^k S$,
    then there exist programs $S_1$ and $S_2$, and integers $k_1, k_2 \in \N$
    such that $S = \qpenv{\qpabsprog{S_1};\qpabsprog{S_2}}$
    and $S'_1 \refine^{k_1} S_1$ and $S'_2 \refine^{k_2} S_2$ and $k_1 + k_2 = k$.
  \item
    If $\qpenv{\qprepeat{N}{\qpabsprog{S_\text{body}'}}} \refine^k S$,
    then there exists a program $S_\text{body}$ such that
    $S_\text{body}' \refine^k S_\text{body}$
    and
    $S~=~\qpenv{\qprepeat{N}{\qpabsprog{S_\text{body}}}}$.
  \item
    If $\qpenv{\qpselect{\vec M}{\vec q}{\{\omega \colon \qpabsprog{S'_\omega}\}_{\omega\in\Omega}}}
    \parref^k S$ \\
    then there exist programs $S_\omega$ and integers $k_\omega \in \N$ for all $\omega\in\Omega$
    such that \\
    $S = \qpenv{\qpselect{\vec M}{\vec q}{{\{\omega \colon \qpabsprog{S_\omega}\}_{\omega\in\Omega}}}}$,
    and $S'_\omega \refine^{k_\omega} S_\omega$ for each $\omega\in\Omega$, and $\sum_\omega k_\omega = k$.
  \item
    If $\qpenv{\qpwhile{B}{\vec q}{\qpabsprog{S_\text{body}'}}} \refine^k S$
    then there exists a program $S_\text{body}$
    such that
    $S = \qpenv{\qpwhile{B}{\vec q}{\qpabsprog{S_\text{body}}}}$
    and
    $S_\text{body}' \refine^k S_\text{body}$.
\end{enumerate}
\end{lemma}
\begin{proof}
(1) We first prove the claim for $k = 1$.
If $\qpenv{\qpabsprog{S'_1};\qpabsprog{S'_2}} \refine S$, then we could only have applied one of \ruleref{C.seqL} or \ruleref{C.seqR} and therefore $S$ must be some sequence statement of the form $\qpenv{\qpabsprog{S_1};\qpabsprog{S_2}}$.
If we use \ruleref{C.seqL} then we have $S'_1 \refine S_1$ and $S'_2 = S_2$, and similarly for \ruleref{C.seqR}.
Since, in both cases, $S$ is again a sequence statement, we can apply the preceding observation inductively to prove the claim for arbitrary~$k$.

Similarly, for~(2) we can only refine by using \ruleref{C.repeat},
for~(3) only using \ruleref{C.case},
and for~(4) only using \ruleref{C.while}.
All these rules preserve the root program structure and only refine the body of the program, hence proving the lemma.
\end{proof}

We can now proceed to prove our soundness theorem for refinement for partial correctness:

%------------------------------------------------------------------------------
\thmsoundnesspartial*
%------------------------------------------------------------------------------

\begin{proof}[Proof of \cref{thm:soundness-partial}]
We prove this by induction over the length of the refinement chain.
That is, we prove by induction for every $k \ge 1$, if $\qpenv{\qphole{P}{Q}} \parref^k S$ and $S$ has no holes, then $\parcorr \hoare{P}{S}{Q}$.

\mymedskip\emph{Base case~$(k = 1)$:}
We have $\qpenv{\qphole{P}{Q}}~\parref~S$.
Since~$S$ has no holes, the refinement can only be due to one of the rules \ruleref{H.skip}, \ruleref{H.init} or \ruleref{H.unit}.
We analyze these three cases separately:

\begin{proofcase}
\ruleref{H.skip}
$S = \qpenv{\qpskip}$ and $P \imp Q$:
The semantics of the former is~$\Bracks{S}(\rho) = \Bracks{\qpenv{\qpskip}}(\rho) = \rho$ and the latter means that~$\tr(P \rho) \le \tr(Q \rho)$ for all states~$\rho$.
Therefore, we have for all states~$\rho$ that
\[
\satisfies{\rho}{P} = \tr(P\rho) \le \tr(Q\rho) = \tr(Q~\Bracks{S}(\rho)) = \satisfies{\Bracks{S}(\rho)}{Q}.
\]
This proves $\totcorr \hoare{P}{\qpenv{\qpskip}}{Q}$ and therefore also $\parcorr \hoare{P}{\qpenv{\qpskip}}{Q}$.
\end{proofcase}

\begin{proofcase}
\ruleref{H.init}
$S = \qpenv{\qpinit{\vec q}}$
and
$P \imp \sum_{\vec x\in\Sigma_{\vec q}} \ketbra{\vec x}{\vec 0}_{\vec q} Q \ketbra{\vec 0}{\vec x}_{\vec q}$:
From the latter, we have for all states $\rho$:
\begin{align*}
    \satisfies{\rho}{P}
    = \tr(P\rho)
    &\le \tr\left(\parens*{\sum_{\vec x\in\Sigma_{\vec q}} \ketbra{\vec x}{\vec 0}_{\vec q} Q \ketbra{\vec 0}{\vec x}_{\vec q}} \rho\right) \\
    &= \tr\left({\sum_{\vec x\in\Sigma_{\vec q}} Q \ketbra{\vec 0}{\vec x}_{\vec q}} \rho \ketbra{\vec x}{\vec 0}_{\vec q} \right) \\
    &= \tr\left(Q ~ \Bracks{\qpenv{\qpinit{\vec q}}}(\rho) \right)
    = \satisfies{\Bracks{S}(\rho)}{Q}.
\end{align*}
This proves $\totcorr\hoare{P}{\qpenv{\qpinit{q}}}{Q}$ and therefore also $\parcorr\hoare{P}{\qpenv{\qpinit{q}}}{Q}$.
\end{proofcase}

\begin{proofcase}
\ruleref{H.unit}
$S = \qpenv{\qpunitary{\vec q}{U}}$
and
$P \imp U^\dagger Q U$:
The latter condition means that $\tr(P\rho) \le \tr(U^\dagger Q U \rho)$ for all states $\rho$, and therefore
\begin{align*}
    \satisfies{\rho}{P}
    = \tr(P\rho)
    &\le \tr(\parens{U^\dagger Q U} \rho) \\
    &= \tr(Q(U\rho U^\dagger)) \\
    &= \tr(Q~\Bracks{\qpenv{\qpunitary{\vec q}{U}}}(\rho))
    = \satisfies{\Bracks{S}(\rho)}{Q}.
\end{align*}
This proves $\totcorr \hoare{P}{\qpenv{\qpunitary{\vec q}{U}}}{Q}$ and therefore also $\parcorr \hoare{P}{\qpenv{\qpunitary{\vec q}{U}}}{Q}$.
\end{proofcase}

\noindent We have thus established the base case, namely that if~$\qpenv{\qphole{P}{Q}} \parref^1 S$ then~$\parcorr \hoare{P}{S}{Q}$.

\mymedskip\emph{Induction Step $(k > 1)$:}
Using the induction hypothesis, we assume that all refinements of length at most $k - 1$ are partially correct.
We now prove that any program that refines in $k$ steps must also be partially correct.
Now, $\qpenv{\qphole{P}{Q}} \parref^k S$ means that there exists a program~$S'$ such that $\qpenv{\qphole{P}{Q}} \parref S' \parref^{k-1} S$.
Since~$S'$ must have holes, the refinement~$\qpenv{\qphole{P}{Q}} \parref S'$ can only be due to one of the following rules:

\begin{proofcase}
\ruleref{H.seq}
$S' = \qpenv{\qphole{P}{R};\qphole{R}{Q}}$:
As $S' \parref^{k-1} S$, \cref{lem:chain_split} states that there are programs~$S_1, S_2$ and integers~$k_1, k_2$
such that $S = \qpenv{\qpabsprog{S_1} ; \qpabsprog{S_2}}$
and $\qpenv{\qphole{P}{R}}~\parref^{k_1}~S_1$
and $\qpenv{\qphole{R}{Q}}~\parref^{k_2}~S_2$
(where~$k_1 + k_2 = k - 1$).
From the induction hypothesis we have that $\parcorr \hoare{P}{S_1}{R}$ and $\parcorr \hoare{R}{S_2}{Q}$.
Therefore, for any partial state $\rho$,
\begin{align*}
  \satisfies{\rho}{P} - \tr(\rho)
  &=\tr(P\rho) - \tr(\rho) \\
  &\le \tr(R~\Bracks{S_1}(\rho)) - \tr(\Bracks{S_1}(\rho)) \\
  &\le \tr(Q~\Bracks{S_2}(\Bracks{S_1}(\rho))) - \tr(\Bracks{S_2}(\Bracks{S_1}(\rho))) \\
  &= \tr(Q~\Bracks{\qpenv{\qpabsprog{S_1};\qpabsprog{S_2}}}(\rho))
      - \tr(\Bracks{\qpenv{\qpabsprog{S_1};\qpabsprog{S_2}}}(\rho)) \\
  &= \satisfies{\Bracks{S}(\rho)}{Q} - \tr(\Bracks{S}(\rho)).
\end{align*}
Therefore, $\parcorr \hoare{P}{S}{Q}$.
\end{proofcase}

\begin{proofcase}
\ruleref{HP.split}
$S' = \qpenv{\qphole{P_\gamma}{Q_\gamma}}$ and $P \imp \sum_\gamma p_\gamma P_\gamma$ and~$\sum_\gamma p_\gamma Q_\gamma \imp Q$ for a probability distribution~$p_\gamma$:
Recall that the two implications mean that~$P \preceq \sum_\gamma p_\gamma P_\gamma$ and~$\sum_\gamma p_\gamma Q_\gamma \preceq Q$.
From the induction hypothesis we have~$\parcorr \hoare{P_\gamma}{S}{Q_\gamma}$ for all~$\gamma$.
Therefore, for any partial state $\rho$,
\begin{align*}
\satisfies{\rho}{P} - \tr(\rho)
&=\tr(P\rho) - \tr(\rho) \\
&\le % \sum_\gamma p_\gamma \tr(P_\gamma \rho) - \tr(\rho) \\ &=
\sum_\gamma p_\gamma \parens*{ \tr(P_\gamma \rho) - \tr(\rho) } \\
&\leq \sum_\gamma p_\gamma \parens*{ \tr(Q_\gamma \Bracks{S}(\rho)) - \tr(\Bracks{S}(\rho)) } \\
%&= \sum_\gamma p_\gamma \tr(Q_\gamma \Bracks{S}(\rho) - \tr(\Bracks{S}(\rho)) \\
&\leq \tr(Q \Bracks{S}(\rho)) - \parens*{\sum_\gamma p_\gamma}\tr(\Bracks{S}(\rho)) \\
&= \satisfies{\Bracks{S}(\rho)}{Q} - \tr(\Bracks{S}(\rho)).
\end{align*}
Therefore, $\parcorr \hoare{P}{S}{Q}$.
\end{proofcase}

\begin{proofcase}
\ruleref{H.repeat}
$S' = \qpenv{\qprepeat{N}{\qphole{R_j}{R_{j+1}}}}$, with $j~\in~\irange{0}{N-1}$ a formal parameter, $P~\imp~R_0$, and~$R_N~\imp~Q$:
As $S' \parref^{k-1} S$, \cref{lem:chain_split} states that $S = \qpenv{\qprepeat{N}{\qpabsprog{C}}}$ for some program~$C$ satisfying~$\qpenv{\qphole{R_j}{R_{j+1}}}~\parref^{k-1}~C$.
By the induction hypothesis, $\parcorr \hoare{R_j}{C}{R_{j+1}}$.
It follows that for all~$j \in \irange{0}{N-1}$ and any partial state $\sigma_j$, we have
\begin{align*}
  \tr(R_j \sigma_j) - \tr\sigma_j \le \tr(R_{j+1} \Bracks{C}(\sigma_j)) - \tr(\Bracks{C}(\sigma_j)).
\end{align*}
Choosing $\sigma_j = \Bracks{C}^j(\rho)$, for some arbitrary partial state $\rho$, the above inequality becomes
\begin{align*}
  \tr(R_j~\Bracks{C}^j(\rho)) - \tr(\Bracks{C}^j(\rho)) \le \tr(R_{j+1} \Bracks{C}^{j+1}(\rho)) - \tr(\Bracks{C}^{j+1}(\rho))
\end{align*}
and chaining the above inequalities for~$j \in \irange{0}{N-1}$ gives us
\begin{align*}
  \satisfies{\rho}{R_0} - \tr(\rho)
&= \tr(R_0~\rho) - \tr(\rho) \\
&\leq \tr(R_N~\Bracks{C}^N(\rho)) - \tr(\Bracks{C}^N(\rho)) \\
&= \satisfies{\Bracks{S}(\rho)}{R_N} - \tr(\Bracks{S}(\rho)),
\end{align*}
where we used that~$\Bracks{S} = \Bracks{C}^N$.
This shows that $\parcorr \hoare{R_0}{S}{R_N}$.
Since~$P~\imp~R_0$ and~$R_N~\imp~Q$, $\parcorr \hoare{P}{S}{Q}$ follows just like in the preceding case.
\end{proofcase}

\begin{proofcase}
\ruleref{H.case}
$S' = \qpenv{\qpselect{\vec M}{\vec q}{\{\omega \colon \qphole{P_\omega}{Q}\}_{\omega\in\Omega}}}$
and $P \imp \sum_{\omega\in\Omega} \cM_\omega(P_\omega)$:
According to \cref{lem:chain_split}, there exist programs~$S_\omega$ and numbers~$k_\omega$ for every~$\omega \in \Omega$,
such that~$S = \qpenv{\qpselect{\vec M}{\vec q}{\{\omega \colon \qpabsprog{S_\omega}\}_{\omega\in\Omega}}}$
and~$\qpenv{\qphole{P_\omega}{Q}} \parref^{k_\omega} S_\omega$
and~$\sum_\omega k_\omega = k - 1$.
From the induction hypothesis we have that $\parcorr \hoare{P_\omega}{S_\omega}{Q}$ for every $\omega\in\Omega$.
Therefore, for any partial state $\sigma$ and any $\omega\in\Omega$,
\[
\tr(P_\omega \sigma) - \tr(\sigma)
\le \tr(Q~\Bracks{S_\omega}(\sigma)) - \tr(\Bracks{S_\omega}(\sigma))
\]
So for every partial state $\rho$, we have
\begin{align*}
  \satisfies{\rho}{P} - \tr(\rho)
  &=\tr(P\rho) - \tr(\rho) \\
  &\le \tr{\parens*{\parens*{\sum_{\omega\in\Omega}\cM_\omega(P_\omega)} \rho}} - \tr(\rho) \\
  % &= \sum_{\omega\in\Omega} \tr{\parens*{\cM_\omega(P_\omega) \rho}} - \sum_{\omega\in\Omega} \tr(\cM_\omega(\rho)) \\
  &= \sum_{\omega\in\Omega} \tr{\parens*{P_\omega \cM_\omega(\rho)}} - \tr(\cM_\omega(\rho)) \\
  &\le \sum_{\omega\in\Omega} \tr{\parens*{Q~\Bracks{S_\omega}(\cM_\omega(\rho))) - \tr(\Bracks{S_\omega}(\cM_\omega(\rho))}} \\
  &= \tr(Q~\Bracks{S}(\rho)) - \tr(\Bracks{S}(\rho)) \\
  &= \satisfies{\Bracks{S}(\rho)}{Q} - \tr(\Bracks{S}(\rho)).
\end{align*}
Therefore, $\parcorr \hoare{P}S{Q}$.
\end{proofcase}

\begin{proofcase}
\ruleref{HP.while}
$S' = \qpenv{\qpwhile{B}{\vec q}{ \qphole{R}{\cB_{0,\q}(Q) + \cB_{1,\q}(R)} }}$ and \\ $P~\imp~\cB_{0,\q}(Q)~+~\cB_{1,\q}(R)$:
As $S' \parref^{k-1} S$, by \cref{lem:chain_split} we must have 
\[ S = \qpenv{\qpwhile{B}{\vec q}{\qpabsprog{C}}} \] for some program~$C$ such that $\qpenv{\qphole{R}{\cB_{0,\q}(Q) + \cB_{1,\q}(R)}} \parref^{k-1} C$.
From the induction hypothesis we have that~$\parcorr~\hoare{R}{C}{\cB_{0,\q}(Q) + \cB_{1,\q}(R)}$, meaning that for every partial state~$\sigma$ we have
\begin{equation*}
  \tr(R\sigma) - \tr(\sigma)
  \le \tr{\parens*{\parens[\big]{\cB_{0,\q}(Q) + \cB_{1,\q}(R)} ~ \Bracks{C}(\sigma)}} - \tr(\Bracks{C}(\sigma)),
\end{equation*}
or equivalently, after some algebraic manipulations,
\begin{align}\label{eq:partial-loop-body-correct}
%   \tr(R\sigma) - \tr(\sigma)
% \le \tr{\parens*{Q \cB_{0,\q}(\Bracks{C}(\sigma))}}
% + \tr{\parens*{R \cB_{1,\q}(\Bracks{C}(\sigma))}}
% - \tr(\cB_{0,\q}(\Bracks{C}(\sigma)))
% - \tr(\cB_{1,\q}(\Bracks{C}(\sigma)))
% =
% \tr(\cB_{0,\q}(\Bracks{C}(\sigma)))
% - \tr{\parens*{Q \cB_{0,\q}(\Bracks{C}(\sigma))}}
% \leq
% - \tr(R\sigma)
% + \tr(\sigma)
% + \tr{\parens*{R \cB_{1,\q}(\Bracks{C}(\sigma))}}
% - \tr(\cB_{1,\q}(\Bracks{C}(\sigma)))
% =
  \tr\braces*{(I - Q)~\cB_{0,\q}(\Bracks{C}(\sigma))}
  \leq
  \tr\braces{(I - R) (\sigma - \cB_{1,\q}(\Bracks{C}(\sigma)))}.
\end{align}
Recalling the semantics of \textbf{while} applied to an arbitrary (partial) state $\rho$, we have
\[
  \Bracks{S}(\rho)
  % = \Bracks{\qpenv{\qpwhile{B}{\vec q}{\qpabsprog{C}}}}(\rho)
  = \sum_{k = 0}^\infty \rho_k,
  \quad \text{where}~ \rho_k = \parens*{ \cB_{0,\q} \circ \parens*{ \Bracks{C} \circ \cB_{1,\q} }^k }(\rho).
\]
For each $k \ge 0$, define $\sigma_k = (\cB_{1,\q} \circ (\Bracks{C} \circ \cB_{1,\q})^k)(\rho)$,
and therefore $\rho_{k+1} = \cB_{0,\q}(\Bracks{C}(\sigma_k))$.
Applying \cref{eq:partial-loop-body-correct} to~$\sigma = \sigma_k$, we get
\begin{align*}
  \tr\braces*{(I - Q)~\rho_{k+1}} \leq \tr\braces{(I - R) (\sigma_k - \sigma_{k+1})}.
\end{align*}
Summing the above over~$k\in\{0,\dots,n-1\}$, for some arbitrary~$n$, and adding~$\tr\braces*{(I - Q)~\rho_0}$, we get
\begin{align*}
\tr\braces*{(I - Q)~\sum_{k = 0}^n \rho_k}
&\leq \tr\braces*{(I - Q)~\rho_0} + \tr\braces{(I - R) (\sigma_0 - \sigma_n)} \\
&= \tr\braces*{(I - Q)~\cB_{0,\q}(\rho)} + \tr\braces{(I - R) (\cB_{1,\q}(\rho) - \sigma_n)} \\
&= \tr(\rho) - \tr\braces*{\parens*{ \cB_{0,\q}(Q) + \cB_{1,\q}(R) } \rho} - \tr\braces{(I - R) \sigma_n} \\
&\leq \tr(\rho) - \tr\braces*{\parens*{ \cB_{0,\q}(Q) + \cB_{1,\q}(R) } \rho} \\
&\leq \tr(\rho) - \tr\braces*{P \rho}
= \satisfies{\rho}{I-P},
\end{align*}
using $P~\imp~\cB_{0,\q}(Q)~+~\cB_{1,\q}(R)$ in the last inequality.
Taking the limit~$n\to\infty$, we obtain that
\begin{align*}
  \satisfies{\Bracks{S}(\rho)}{I-Q} \leq \satisfies{\rho}{I-P}
\end{align*}
which means that $\parcorr \hoare{P}S{Q}$.
\end{proofcase}
% From the above cases we can conclude that if $\qpenv{\qphole{P}{Q}} \parref^k S$, then $\parcorr \hoare{P}{S}{Q}$.

\mymedskip

This concludes the proof of the induction step and hence the proof of the theorem.
\end{proof}

We now prove soundness of refinement for total correctness.
Its proof is very similar to the one of \cref{thm:soundness-partial}.

%------------------------------------------------------------------------------
\thmsoundnesstotal*
%------------------------------------------------------------------------------

\begin{proof}
We prove by induction for $k \ge 1$ that if $\qpenv{\qphole{P}{Q}} \totref^k S$ and~$S$ is a concrete program (i.e., has no holes) then~$\totcorr \hoare{P}{S}{Q}$.

\mymedskip\emph{Base case~$(k = 1)$:}
We have $\qpenv{\qphole{P}{Q}}~\totref~S$.
Since~$S$ has no holes, the refinement can only be due to one of the rules \ruleref{H.skip}, \ruleref{H.init} or \ruleref{H.unit}.
In the proof of \cref{thm:soundness-partial} we already saw that these rules are totally correct.
This concludes the base case.

\mymedskip\emph{Induction Step $(k > 1)$:}
Using the induction hypothesis, we assume that all refinements of length at most $k - 1$ are totally correct.
We now prove that any program that refines in $k$ steps must also be totally correct.
Now, $\qpenv{\qphole{P}{Q}} \totref^k S$ means that there exists a program~$S'$ such that $\qpenv{\qphole{P}{Q}} \totref S' \totref^{k-1} S$.
Since~$S'$ must have holes, the refinement~$\qpenv{\qphole{P}{Q}} \totref S'$ can only be due to one of the following rules:

\begin{proofcase}
\ruleref{H.seq}
$S' = \qpenv{\qphole{P}{R};\qphole{R}{Q}}$:
As $S' \totref^{k-1} S$, \cref{lem:chain_split} states that there are programs~$S_1, S_2$ and integers~$k_1, k_2$
such that $S = \qpenv{\qpabsprog{S_1} ; \qpabsprog{S_2}}$
and $\qpenv{\qphole{P}{R}}~\totref^{k_1}~S_1$
and $\qpenv{\qphole{R}{Q}}~\totref^{k_2}~S_2$
(where~$k_1 + k_2 = k - 1$).
From the induction hypothesis we have that $\totcorr \hoare{P}{S_1}{R}$ and $\totcorr \hoare{R}{S_2}{Q}$.
Therefore, for any partial state $\rho$,
\begin{align*}
  \satisfies{\rho}{P}
  =\tr(P\rho)
  &\le \tr(R~\Bracks{S_1}(\rho)) \\
  &\le \tr(Q~\Bracks{S_2}(\Bracks{S_1}(\rho))) \\
  &= \tr(Q~\Bracks{\qpenv{\qpabsprog{S_1};\qpabsprog{S_2}}}(\rho)) \\
  &= \satisfies{\Bracks{S}(\rho)}{Q}.
\end{align*}
Therefore, $\totcorr \hoare{P}{S}{Q}$.
\end{proofcase}

\begin{proofcase}
\ruleref{HT.split}
$S' = \qpenv{\qphole{P_\gamma}{Q_\gamma}}$ and $P \imp \sum_\gamma p_\gamma P_\gamma$ and~$\sum_\gamma p_\gamma Q_\gamma \imp Q$:
Recall the latter mean that~$P \preceq \sum_\gamma p_\gamma P_\gamma$ and~$\sum_\gamma p_\gamma Q_\gamma \preceq Q$.
From the induction hypothesis we have~$\totcorr \hoare{P_\gamma}{S}{Q_\gamma}$ for all~$\gamma$.
Therefore, for any partial state $\rho$,
\begin{align*}
\satisfies{\rho}{P}
&=\tr(P\rho) \\
&\leq \sum_\gamma p_\gamma \tr(P_\gamma \rho) \\
&\leq \sum_\gamma p_\gamma \tr(Q_\gamma~\Bracks{S}(\rho)) \\
&\leq \tr(Q \Bracks{S}(\rho)) = \satisfies{\Bracks{S}(\rho)}{Q}.
\end{align*}
Therefore, $\totcorr \hoare{P}{S}{Q}$.
\end{proofcase}

\begin{proofcase}
\ruleref{H.repeat}
$S' = \qpenv{\qprepeat{N}{\qphole{R_j}{R_{j+1}}}}$, with $j~\in~\irange{0}{N-1}$ a formal parameter, $P~\imp~R_0$, and~$R_N~\imp~Q$:
As $S' \totref^{k-1} S$, \cref{lem:chain_split} states that $S = \qpenv{\qprepeat{N}{\qpabsprog{C}}}$ for some program~$C$ satisfying~$\qpenv{\qphole{R_j}{R_{j+1}}}~\totref^{k-1}~C$.
By the induction hypothesis, $\totcorr \hoare{R_j}{C}{R_{j+1}}$.
It follows that for all~$j \in \irange{0}{N-1}$ and any partial state $\sigma_j$, we have
\begin{align*}
  \tr(R_j \sigma_j) \le \tr(R_{j+1}~\Bracks{C}(\sigma_j)).
\end{align*}
Choosing $\sigma_j = \Bracks{C}^j(\rho)$, for some arbitrary partial state $\rho$, the above inequality becomes
\begin{align*}
  \tr(R_j~\Bracks{C}^j(\rho)) \le \tr(R_{j+1}~\Bracks{C}^{j+1}(\rho)),
\end{align*}
and chaining the above inequalities for~$j \in \irange{0}{N-1}$ gives us
\begin{align*}
  \satisfies{\rho}{R_0}
  = \tr(R_0~\rho)
  \leq \tr(R_N~\Bracks{C}^N(\rho))
  = \satisfies{\Bracks{S}(\rho)}{R_N},
\end{align*}
where we used that~$\Bracks{S} = \Bracks{C}^N$.
This shows that $\totcorr \hoare{R_0}{S}{R_N}$.
Since~$P~\imp~R_0$ and~$R_N~\imp~Q$, it follows that~$\totcorr \hoare{P}{S}{Q}$ just like in the preceding case.
\end{proofcase}

\begin{proofcase}
\ruleref{H.case}
$S' = \qpenv{\qpselect{\vec M}{\vec q}{\{\omega \colon \qphole{P_\omega}{Q}\}_{\omega\in\Omega}}}$
and $P \imp \sum_{\omega\in\Omega} \cM_\omega(P_\omega)$:
According to \cref{lem:chain_split}, there exist programs~$S_\omega$ and numbers~$k_\omega$ for every~$\omega \in \Omega$,
such that~$S = \qpenv{\qpselect{\vec M}{\vec q}{\{\omega \colon \qpabsprog{S_\omega}\}_{\omega\in\Omega}}}$
and~$\qpenv{\qphole{P_\omega}{Q}} \totref^{k_\omega} S_\omega$
and~$\sum_\omega k_\omega = k - 1$.
From the induction hypothesis we have that $\totcorr \hoare{P_\omega}{S_\omega}{Q}$ for every $\omega\in\Omega$.
Therefore, for any partial state $\sigma$ and any $\omega\in\Omega$,
\[
\tr(P_\omega \sigma)
\le \tr(Q~\Bracks{S_\omega}(\sigma))
\]
So for every partial state $\rho$ we have
\begin{align*}
  \satisfies{\rho}{P}
  =\tr(P\rho)
  &\le \tr{\parens*{\parens*{\sum_{\omega\in\Omega}\cM_\omega(P_\omega)} \rho}}  \\
  &= \sum_{\omega\in\Omega} \tr{\parens*{P_\omega \cM_\omega(\rho)}} \\
  &\le \sum_{\omega\in\Omega} \tr{\parens*{Q~\Bracks{S_\omega}(\cM_\omega(\rho)))}} \\
  &= \tr(Q~\Bracks{S}(\rho))
  = \satisfies{\Bracks{S}(\rho)}{Q}.
\end{align*}
Therefore, $\totcorr \hoare{P}S{Q}$.
\end{proofcase}

\begin{proofcase}
\ruleref{HT.while}
  $S' = \qpenv{\qpwhile{B}{\q}{\qphole{R_{n+1}}{\cB_{0,\q}(Q) + \cB_{1,\q}(R_n)}}}$, with~$n\in\N$ a formal parameter, $B$~a binary measurement, and~$\{R_n\}_{n\in\N}$ sequence of predicates such that~~$R_0 = 0$, $R_n \imp R_{n+1}$ for all~$n\in\N$, and the limit $R := \lim_{n\to\infty} R_n$ satisfies $P \imp \cB_{0,\q}(Q) + \cB_{1,\q}(R)$:
  By \cref{lem:chain_split}, because~$S' \totref^{k-1} S$, we must have~$S = \qpenv{\qpwhile{B}{\q}{\qpabsprog{C}}}$ for some program~$C$ satisfying~$\qphole{R_{n+1}}{\cB_{0,\q}(Q) + \cB_{1,\q}(R_n)} \totref^{k-1} C$.
  By the induction hypothesis we have~$\totcorr \hoare{R_{n+1}}{C}{\cB_{0,\q}(Q) + \cB_{1,\q}(R_n)}$ for all~$n\in\N$.
  This means that for every partial state~$\sigma$ and every~$n\in\N$, it holds that
  \begin{align*}
    \tr\parens*{ R_{n+1} \sigma }
  &\leq \tr \parens*{ \parens*{ \cB_{0,\q}(Q) + \cB_{1,\q}(R_n) }~\Bracks{C}(\sigma) } \\
  &= \tr \parens*{ Q~\parens*{ \cB_{0,\q} \circ \Bracks{C} }(\sigma) } + \tr \parens*{ R_n~\parens*{ \cB_{1,\q} \circ \Bracks{C} }(\sigma) }
  \end{align*}
  By repeatedly applying this inequality for~$R_n,R_{n-1},\ldots,R_1$ and using that $R_0 = 0$, we get
  \begin{align*}
    \tr\parens*{ R_{n+1} \sigma }
  % &\leq \tr \parens*{ Q~\parens*{ \cB_{0,\q} \circ \Bracks{C} }(\sigma) } + \tr \parens*{ R_n~\parens*{ \cB_{1,\q} \circ \Bracks{C} }(\sigma) } \\
  % &\leq \tr \parens*{ Q~\parens*{ \cB_{0,\q} \circ \Bracks{C} }(\sigma) }
  %   + \tr \parens*{ Q~\parens*{ \cB_{0,\q} \circ \Bracks{C} \circ \cB_{1,\q} \circ \Bracks{C} }(\sigma) }
  %   + \tr \parens*{ R_{n-1}~\parens*{ (\cB_{1,\q} \circ \Bracks{C})^2(\sigma) } } \\
  % &\leq \dots \\
  &\leq \sum_{k=0}^n \tr \parens*{ Q~\parens*{ \cB_{0,\q} \circ \Bracks{C} \circ \parens*{ \cB_{1,\q} \circ \Bracks{C} }^k }(\sigma) }
    + \tr \parens*{ R_0~\parens*{ (\cB_{1,\q} \circ \Bracks{C})^n(\sigma) } } \\
  &= \tr \parens*{ Q~\sum_{k=0}^n \parens*{ \cB_{0,\q} \circ \parens*{ \Bracks{C} \circ \cB_{1,\q} }^k }\parens*{ \Bracks{C}(\sigma) } }.
  \end{align*}
  Therefore at the limit~$n\to\infty$ we obtain, using the semantics of the \textbf{while} loop,
  \begin{align*}
    \tr\parens*{ R \sigma }
  &= \lim_{n\to\infty} \tr\parens*{ R_{n+1} \sigma }
  \\ &\leq \tr \parens*{ Q~\sum_{k=0}^\infty \parens*{ \cB_{0,\q} \circ \parens*{ \Bracks{C} \circ \cB_{1,\q} }^k }\parens*{ \Bracks{C}(\sigma) } }
  \\ &= \tr \parens*{ Q~\Bracks{S}\parens*{ \Bracks{C}(\sigma) } }.
  \end{align*}
  Using the assumption that $P \imp \cB_{0,\q}(Q) + \cB_{1,\q}(R)$ and the recurrence relation in \cref{eq:while loop recurrence}, it follows that for every partial state~$\rho$ we have
  \begin{align*}
    \satisfies{\rho}{P}
  = \tr\parens*{ P \rho }
  &\leq \tr \parens*{ Q~\cB_{0,\q}(\rho) } + \tr \parens*{ R~\cB_{1,\q}(\rho) } \\
  &\leq \tr \parens*{ Q~\cB_{0,\q}(\rho) } + \tr \parens*{ Q~\Bracks{S}\parens*{ \Bracks{C} \parens*{ \cB_{1,\q}(\rho) } } } \\
  &= \tr \parens*{ Q~\parens*{ \cB_{0,\vec q} + \Bracks{S} \circ \Bracks{C} \circ \cB_{1,\vec q} }(\rho) } \\
  &= \tr \parens*{ Q~\Bracks{S}(\rho) }
  = \satisfies{\Bracks{S}(\rho)}{Q}.
  \end{align*}
  Therefore, $\totcorr \hoare{P}{S}{Q}$.
\end{proofcase}

\mymedskip

This concludes the proof of the induction step and hence the proof of the theorem.
\end{proof}

%==============================================================================
\section{Completeness of Refinement}\label{app:completeness}
%==============================================================================
In this appendix we present the proofs of \cref{thm:completeness-partial} and \cref{thm:completeness-total}.

%------------------------------------------------------------------------------
\thmcompletenesspartial*
%------------------------------------------------------------------------------

\begin{proof}
  We prove this by induction on the structure of the concrete program~$S$.
  Throughout we use that~$\parcorr \hoare{P}{S}{Q}$ is equivalent to~$I - P \succeq \Bracks{S}^\dagger(I - Q)$.

  \emph{Base cases:}
  There are three types of programs without subprograms to consider:

  \begin{proofcase}
    $S = \qpenv{\qpskip}$:
    By assumption, we have~$\parcorr \hoare{P}{\qpenv{\qpskip}}{Q}$, which is equivalent to~$P \imp Q$.
    Thus we can use the refinement rule \ruleref{H.skip} to obtain $\qpenv{\qphole{P}{Q}} \parref S$, proving this case.
  \end{proofcase}

  \begin{proofcase}
    $S = \qpenv{\qpinit{\q}}$:
    By assumption, we have~$\parcorr \hoare{P}{\qpenv{\qpinit{\q}}}{Q}$ which is equivalent to
    \begin{align*}
      I - P
    \succeq \sum_{\vec x \in \Sigma_{\vec q}} \ketbra{\vec x}{\vec 0}_{\vec q} (I - Q) \ketbra{\vec 0}{\vec x}_{\vec q}
    = I - \sum_{\vec x \in \Sigma_{\vec q}} \ketbra{\vec x}{\vec 0}_{\vec q} Q \ketbra{\vec 0}{\vec x}_{\vec q}
    \end{align*}
    and hence $P \imp \sum_{\vec x \in \Sigma_{\vec q}} \ketbra{\vec x}{\vec 0}_{\vec q} Q \ketbra{\vec 0}{\vec x}_{\vec q}$.
    Thus we can use the refinement rule \ruleref{H.init} to obtain $\qpenv{\qphole{P}{Q}} \parref S$, proving this case.
  \end{proofcase}

  \begin{proofcase}
    $S = \qpenv{\qpunitary{\q}{U}}$:
    By assumption, we have~$\parcorr \hoare{P}{\qpenv{\qpunitary{\q}{U}}}{Q}$, which is equivalent to
    \begin{align*}
      I - P \succeq U_\q^\dagger (I - Q) U_\q = I - U_\q^\dagger Q U_\q
    \end{align*}
    and hence~$P \imp U_\q^\dagger Q U_\q$.
    Thus can use the refinement rule \ruleref{H.unit} to obtain $\qpenv{\qphole{P}{Q}} \parref S$, proving this case.
  \end{proofcase}

  \emph{Induction step:}
  We now consider a concrete program~$S$ with subprograms.
  We assume that~$\parcorr \hoare{P}{S}{Q}$ and we must prove that $\qpenv{\qphole{P}{Q}} \parref^* S$.
  By the induction hypothesis, we know that for every syntactic subprogram~$S'$ of~$S$ (which is necessarily concrete as well), and for arbitrary predicates~$P'$ and~$Q'$, it holds that~$\parcorr \hoare{P'}{S'}{Q'}$ implies~$\qpenv{\qphole{P'}{Q'}} \parref^* S'$.
  There are four cases to consider for~$S$:

  \begin{proofcase}
    $S = \qpenv{\qpabsprog{S_1} ; \qpabsprog{S_2}}$:
    By assumption, we have $\parcorr \hoare{P}{S_1 ; S_2}{Q}$, which means $I - P \succeq \Bracks{S_1}^\dagger (\Bracks{S_2}^\dagger (I - Q))$.
    If we choose $R = I - \Bracks{S_2}^\dagger(I - Q)$, then it follows that~$I - P \succeq \Bracks{S_1}^\dagger (I - R)$, as well as, of course, $I - R = \Bracks{S_2}^\dagger (I - Q)$, which mean that~$\parcorr \hoare{P}{S_1}{R}$ and~$\parcorr \hoare{R}{S_2}{Q}$, respectively.
    Since~$S$ is a concrete program, so are~$S_1$ and~$S_2$.
    Hence we can see from the induction hypothesis that
    $\qpenv{\qphole{P}{R}} \parref^* S_1$
    and
    $\qpenv{\qphole{R}{Q}} \parref^* S_2$.
    We can thus construct~$S$ by first using \ruleref{H.seq} with the intermediate condition~$R$ to obtain $\qpenv{\qphole{P}{Q}} \parref \qpenv{\qphole{P}{R} ; \qphole{R}{Q}}$ and then using \ruleref{C.seqL} and \ruleref{C.seqR}, proving this case.
  \end{proofcase}

  \begin{proofcase}
    $S = \qpenv{\qprepeat{N}{\qpabsprog{S'}}}$:
    By assumption, we have $\parcorr \hoare{P}{\qpenv{\qprepeat{N}{\qpabsprog{S'}}}}{Q}$, which means $I - P \succeq (\Bracks{S'}^\dagger)^N (I - Q)$.
    Let us choose~$R_j = I - (\Bracks{S'}^\dagger)^{N - j}(I - Q)$ for $j \in \{0, \ldots, N\}$.
    On the one hand, this ensures that~$I - R_j = \Bracks{S'}^\dagger(I - R_{j+1})$ for every~$j \in \{0,\dots,N-1\}$, that is, $\parcorr \hoare{R_j}{S'}{R_{j+1}}$.
    Since~$S$ is a concrete program, so is~$S'$, hence we can see from the induction hypothesis that~$\qpenv{\qphole{R_j}{R_{j+1}}} \parref^* S'$.
    On the other hand, it also holds that~$P \imp R_0$ and $R_N = Q$.
    We can thus construct~$S$ by first using \ruleref{H.repeat} with the family~$\{R_j\}$ to obtain
    $\qpenv{\qphole{P}{Q}} \parref \qpenv{\qprepeat{N}{\qphole{R_j}{R_{j+1}}}}$
    and then using \ruleref{C.repeat}, proving this case.
  \end{proofcase}

  \begin{proofcase}
    $S = \qpenv{
      \qpselect
        {\{\omega_1 \colon M_{\omega_1}, \omega_2 \colon M_{\omega_2}, \ldots\}}
        {\q}
        {\omega_1 \colon \qpabsprog{S_{\omega_1}}, \ \omega_2 \colon \qpabsprog{S_{\omega_2}}, \ \ldots}
    }$, where \\ $\{M_\omega\}_{\omega\in\Omega}$ is a measurement on~$\mathcal H_q$ with outcomes in~$\Omega = \{\omega_1,\omega_2,\dots\}$:
    By assumption, we have $\parcorr~\hoare{P}{S}{Q}$, which means
    \begin{align}\label{eq:case compl interm}
      I - P
      \succeq \Bracks{S}^\dagger (I - Q)
      = \sum_{\omega\in\Omega} \cM_\omega \parens*{\Bracks{S_\omega}^\dagger(I - Q)}
      = I - \sum_{\omega\in\Omega} \cM_\omega \parens*{I - \Bracks{S_\omega}^\dagger(I - Q)}.
    \end{align}
    Let us choose~$P_\omega = I - \Bracks{S_\omega}^\dagger(I - Q)$ for~$\omega\in\Omega$.
    On the one hand, this ensures that~$\parcorr \hoare{P_\omega}{S_\omega}{Q}$, hence we can see from the induction hypothesis that~$\qpenv{\qphole{P_\omega}{Q}} \parref^* S_\omega$ for every~$\omega\in\Omega$.
    On the other hand, by \cref{eq:case compl interm} it also holds that~$P \imp \sum_{\omega\in\Omega} \cM_\omega(P_\omega)$.
    We can thus construct~$S$ by first using \ruleref{H.case} to obtain
    $\qpenv{\qphole{P}{Q}} \parref \qpenv{\qpselect {\vec M} {\q} { \{\omega \colon \qphole{P_\omega}{Q}\}_{\omega\in\Omega}}}$
    and then using \ruleref{C.case} for each~$\omega\in\Omega$, proving this case.
  \end{proofcase}

  \begin{proofcase}
    $S = \qpenv{\qpwhile{B}{\q}{\qpabsprog{C}}}$:
    By assumption, we have~$\parcorr \hoare{P}{S}{Q}$,
    which means that~$I - P \succeq \Bracks{S}^\dagger(I-Q)$, or
    \begin{align}\label{eq:while compl 1}
      P \imp I - \Bracks{S}^\dagger (I - Q).
    \end{align}
    Let us choose~$R = I - \Bracks{C}^\dagger(\Bracks{S}^\dagger(I - Q))$.
    Then we have from the recurrence in \cref{eq:while loop recurrence} that
    \begin{align}\label{eq:while compl 2}
      % \Bracks{S}^\dagger(Q) = \cB_{0,\q}(Q) + \cB_{1,\q}(\Bracks{C}^\dagger(\Bracks{S}^\dagger(Q)))
      \Bracks{S}^\dagger(I - Q)
    = \cB_{0,\q}(I - Q) + \cB_{1,\q}(\Bracks{C}^\dagger(\Bracks{S}^\dagger(I - Q)))
    % = \cB_{0,\q}(I - Q) + \cB_{1,\q}(I - R)
    = I - (\cB_{0,\q}(Q) + \cB_{1,\q}(R)).
    \end{align}
    On the one hand, it follows that
    \begin{align*}
      I - R
    = \Bracks{C}^\dagger(\Bracks{S}^\dagger(I - Q))
    = \Bracks{C}^\dagger( I - (\cB_{0,\q}(Q) + \cB_{1,\q}(R)) )
    \end{align*}
    Thus, $\parcorr \hoare{R}{C}{\cB_{0,\q}(Q) + \cB_{1,\q}(R)}$, so we find~$\qpenv{\qphole{R}{\cB_{0,\q}(Q) + \cB_{1,\q}(R)}} \parref^* C$ from the induction hypothesis.
    On the other hand, we have from \cref{eq:while compl 1,eq:while compl 2} that~$P \imp \cB_{0,\q}(Q) + \cB_{1,\q}(R)$.
    We can thus construct~$S$ by first using the refinement rule \ruleref{HP.while} with the above $R$, to obtain $\qpenv{\qphole{P}{Q}} \parref \qpenv{\qpwhile{\q}{B}{\qphole{R}{\cB_{0,\q}(Q) + \cB_{1,\q}(R)}}}$, and then \ruleref{C.while}, proving this case.
  \end{proofcase}

\mymedskip

This concludes the structural induction and hence the proof of the theorem.
\end{proof}

We now present the proof of completeness of refinement for total correctness.
The proof is very similar to the one of \cref{thm:completeness-partial}.

%------------------------------------------------------------------------------
\thmcompletenesstotal*
%------------------------------------------------------------------------------

\begin{proof}
  We prove this by induction on the structure of the concrete program~$S$, similar to the proof of \cref{thm:completeness-partial}.
  Throughout we use that~$\totcorr \hoare{P}{S}{Q}$ is equivalent to~$P \preceq \Bracks{S}^\dagger(Q)$.

  \emph{Base cases:}
  There are three types of programs without subprograms to consider:

  \begin{proofcase}
    $S = \qpenv{\qpskip}$:
    By assumption, we have~$\totcorr \hoare{P}{\qpenv{\qpskip}}{Q}$, which is equivalent to~$P \imp Q$.
    Thus we can use the refinement rule \ruleref{H.skip} to obtain $\qpenv{\qphole{P}{Q}} \totref S$, proving this case.
  \end{proofcase}

  \begin{proofcase}
    $S = \qpenv{\qpinit{\q}}$:
    By assumption, we have~$\totcorr \hoare{P}{\qpenv{\qpinit{\q}}}{Q}$ which is equivalent to
    \begin{align*}
      P
      \preceq \sum_{\vec x \in \Sigma_{\vec q}} \ketbra{\vec x}{\vec 0}_{\vec q} Q \ketbra{\vec 0}{\vec x}_{\vec q}.
    \end{align*}
    Thus we can use the refinement rule \ruleref{H.init} to obtain $\qpenv{\qphole{P}{Q}} \totref S$, proving this case.
  \end{proofcase}

  \begin{proofcase}
    $S = \qpenv{\qpunitary{\q}{U}}$:
    By assumption, we have~$\totcorr \hoare{P}{\qpenv{\qpunitary{\q}{U}}}{Q}$, which is equivalent to
    \begin{align*}
      P \preceq U_\q^\dagger Q U_\q.
    \end{align*}
    Thus can use the refinement rule \ruleref{H.unit} to obtain $\qpenv{\qphole{P}{Q}} \totref S$, proving this case.
  \end{proofcase}

  \emph{Induction step:}
  We now consider a concrete program~$S$ with subprograms.
  We assume that~$\totcorr \hoare{P}{S}{Q}$ and we must prove that $\qpenv{\qphole{P}{Q}} \totref^* S$.
  By the induction hypothesis, we know that for every syntactic subprogram~$S'$ of~$S$ (which is necessarily concrete as well), and for arbitrary predicates~$P'$ and~$Q'$, it holds that~$\totcorr \hoare{P'}{S'}{Q'}$ implies~$\qpenv{\qphole{P'}{Q'}} \totref^* S'$.
  There are four cases to consider for~$S$:

  \begin{proofcase}
    $S = \qpenv{\qpabsprog{S_1} ; \qpabsprog{S_2}}$:
    By assumption, we have $\totcorr \hoare{P}{S_1 ; S_2}{Q}$, which means $P \preceq \Bracks{S_1}^\dagger (\Bracks{S_2}^\dagger (Q))$.
    If we choose $R = \Bracks{S_2}^\dagger(Q)$, then it follows that~$P \preceq \Bracks{S_1}^\dagger (R)$, as well as, of course, $R \preceq \Bracks{S_2}^\dagger (Q)$, which mean that~$\totcorr \hoare{P}{S_1}{R}$ and~$\totcorr \hoare{R}{S_2}{Q}$, respectively.
    Since~$S$ is a concrete program, so are~$S_1$ and~$S_2$.
    Hence we can see from the induction hypothesis that
    $\qpenv{\qphole{P}{R}} \totref^* S_1$
    and
    $\qpenv{\qphole{R}{Q}} \totref^* S_2$.
    We can thus construct~$S$ by first using \ruleref{H.seq} with the intermediate condition~$R$ to obtain $\qpenv{\qphole{P}{Q}} \totref \qpenv{\qphole{P}{R} ; \qphole{R}{Q}}$ and then using \ruleref{C.seqL} and \ruleref{C.seqR}, proving this case.
  \end{proofcase}

  \begin{proofcase}
    $S = \qpenv{\qprepeat{N}{\qpabsprog{S'}}}$:
    By assumption, we have $\totcorr \hoare{P}{\qpenv{\qprepeat{N}{\qpabsprog{S'}}}}{Q}$, which means $P \preceq (\Bracks{S'}^\dagger)^N (Q)$.
    Let us choose~$R_j = (\Bracks{S'}^\dagger)^{N - j}(Q)$ for $j \in \{0, \ldots, N\}$.
    On the one hand, this ensures that~$R_j = \Bracks{S'}^\dagger(R_{j+1})$ for every~$j \in \{0,\dots,N-1\}$, that is, $\totcorr \hoare{R_j}{S'}{R_{j+1}}$.
    Since~$S$ is a concrete program, so is~$S'$, hence we can see from the induction hypothesis that~$\qpenv{\qphole{R_j}{R_{j+1}}} \totref^* S'$.
    On the other hand, it also holds that~$P \imp R_0$ and $R_N = Q$.
    We can thus construct~$S$ by first using \ruleref{H.repeat} with the family~$\{R_j\}$ to obtain
    $\qpenv{\qphole{P}{Q}} \totref \qpenv{\qprepeat{N}{\qphole{R_j}{R_{j+1}}}}$
    and then using \ruleref{C.repeat}, proving this case.
  \end{proofcase}

  \begin{proofcase}
    $S = \qpenv{
      \qpselect
        {\{\omega_1 \colon M_{\omega_1}, \omega_2 \colon M_{\omega_2}, \ldots\}}
        {\q}
        {\omega_1 \colon \qpabsprog{S_{\omega_1}}, \ \omega_2 \colon \qpabsprog{S_{\omega_2}}, \ \ldots}
    }$, where \\ $\{M_\omega\}_{\omega\in\Omega}$ is a measurement on~$\mathcal H_q$ with outcomes in~$\Omega = \{\omega_1,\omega_2,\dots\}$:
    By assumption, we have $\totcorr~\hoare{P}{S}{Q}$, which means
    \begin{align}\label{eq:case compl interm tot}
      P
      \preceq \Bracks{S}^\dagger (Q)
      = \sum_{\omega\in\Omega} \cM_\omega \parens*{\Bracks{S_\omega}^\dagger(Q)}
    \end{align}
    Let us choose~$P_\omega = \Bracks{S_\omega}^\dagger(Q)$ for~$\omega\in\Omega$.
    On the one hand, this ensures that~$\totcorr \hoare{P_\omega}{S_\omega}{Q}$, hence we can see from the induction hypothesis that~$\qpenv{\qphole{P_\omega}{Q}} \totref^* S_\omega$ for every~$\omega\in\Omega$.
    On the other hand, by \cref{eq:case compl interm tot} it also holds that~$P \imp \sum_{\omega\in\Omega} \cM_\omega(P_\omega)$.
    We can thus construct~$S$ by first using \ruleref{H.case} to obtain
    $\qpenv{\qphole{P}{Q}} \totref \qpenv{\qpselect {\vec M} {\q} { \{\omega \colon \qphole{P_\omega}{Q}\}_{\omega\in\Omega}}}$
    and then using \ruleref{C.case} for each~$\omega\in\Omega$, proving this case.
  \end{proofcase}

  \begin{proofcase}
    $S = \qpenv{\qpwhile{B}{\q}{\qpabsprog{C}}}$:
    By assumption, we have~$\totcorr \hoare{P}{S}{Q}$,
    which means that~$P \imp \Bracks{S}^\dagger(Q)$.
    Recall that the semantics of the loop is given by
    \begin{align}\label{eq:def S_n}
      \Bracks{S} = \lim_{n\to\infty} \cS_n,
    \quad\text{where}\quad
      \cS_n = \sum_{k=0}^n \cB_{0,\q} \circ (\Bracks{C} \circ \cB_{1,\q})^k.
    \end{align}
    Let us choose~$R_0 := 0$ and~$R_{n+1} := \Bracks{C}^\dagger(\cS_n^\dagger(Q))$ for all~$n\in\N$.
    On the one hand, using the relation~$\cS_n = \cB_{0,\q} + \cS_{n-1} \circ \Bracks{C} \circ \cB_{1,\q}$, we get
    \begin{align*}
      R_{n+1}
    &= \Bracks{C}^\dagger(\cS_n^\dagger(Q))
    \\ &= \Bracks{C}^\dagger\parens*{ \cB_{0,\q}(Q) + \cB_{1,\q}(\Bracks{C}^\dagger(\cS_{n-1}^\dagger(Q))) }
    \\ &= \Bracks{C}^\dagger\parens*{ \cB_{0,\q}(Q) + \cB_{1,\q}(R_n) },
    \end{align*}
    and this also holds for~$n=0$ if we set~$\cS_{-1}=0$.
    Thus, $\totcorr \hoare{R_{n+1}}{C}{\cB_{0,\q}(Q) + \cB_{1,\q}(R_n)}$, so we find~$\qpenv{\qphole{R_{n+1}}{\cB_{0,\q}(Q) + \cB_{1,\q}(R_n)}} \totref^* C$ from the induction hypothesis.
    On the other hand, it is clear from its definition and \cref{eq:def S_n} that the sequence~$\{R_n\}_{n\in\N}$ is weakly increasing (that is, $R_n \imp R_{n+1}$ for all~$n\in\N$).
    Moreover, using the recurrence in \cref{eq:while loop recurrence} (which is of course simply the limit of the relation used above), the limit~$R = \lim_{n\to\infty} R_n = \Bracks{C}^\dagger(\Bracks{S}^\dagger(Q))$ of the sequence satisfies
    \begin{align*}
      P
    \preceq \Bracks{S}^\dagger(Q)
    = \parens*{ \cB_{0,\vec q} + \cB_{1,\vec q} \circ \Bracks{C}^\dagger \circ \Bracks{S}^\dagger }(Q)
    = \cB_{0,\q}(Q) + \cB_{1,\q}(R).
    \end{align*}
    Thus, $P \imp \cB_{0,\q}(Q) + \cB_{1,\q}(R)$.
    We can thus construct~$S$ by first using the refinement rule \ruleref{HT.while} with the above family~$\{R_n\}$,
    to obtain
    \[
    \qpenv{\qphole{P}{Q}} \totref \qpenv{\qpwhile{B}{\q}{\qphole{R_{n + 1}}{\cB_{0,\q}(Q) + \cB_{1,\q}(R_n)}}},
    \]
    and then \ruleref{C.while}, proving this case.
  \end{proofcase}

\mymedskip

This concludes the structural induction and hence the proof of the theorem.
\end{proof}

%==============================================================================
\section{Example: Boosting Success Probability}
\label{app:examples:boosting}
%==============================================================================

In this appendix we prove \cref{thm:boosting}.

% ------------------------------------------------------------------------------
% \subsection{Boosting Success Probability}
% ------------------------------------------------------------------------------

\thmboostingsuccessprobability*

\begin{proof}

\begin{proofcase}\ruleref{H.boostRep}
We show that it holds by successively refining the left-hand side into the right-hand side.
To this end we first refine as follows,
\begin{align*}
  \qpenv{\qphole{pI, I}{Q_\q, I}}
\totref \qpenv{ \qprepeat{N}{ \underbrace{\qphole{R_j, I}{R_{j+1}, I} }_{S_\text{body}} } }, \ruletag{H.repeat}
\end{align*}
where we choose $N = \ceil{\log_{1-\eps} (1 - p)}$ and
\[
  R_j
= (1 - \eps)^{N-j} Q_\q + \parens[\big]{ 1 - (1 - \eps)^{N-j} } I
% = \parens[\big]{ 1 - (1 - \eps)^{N-j} } (I - Q_\q) + Q_\q
\]
for $j\in\{0,\dots,N\}$.
The application of \ruleref{H.repeat} is valid since $R_N = Q_\q$, while $pI \imp R_0$ holds because~$p \leq 1 - (1 - \eps)^{N}$ by our choice of~$N$.
% Indeed:
% \begin{align*}
%   p \leq 1 - (1 - \eps)^N \\
%   1 - p \geq (1 - \eps)^N \\
%   \log(1 - p) \geq N \log(1 - \eps) \\
%   N \geq \frac {\log(1 - p)} {\log(1 - \eps)} = \log_{1-\eps}(1-p).
% \end{align*}
Next, since we only want to apply the subroutine if the postcondition does not already hold, we refine the loop body by measuring the postcondition:
\begin{align*}
  S_\text{body}
\totref \qpenv{ \qpif{I - Q}{\vec q}{ \underbrace{ \qphole{P_j, I - Q_\q}{R_{j+1}, I} }_{S_\text{then}} } } \ruletag{H.if}
\end{align*}
which is valid if we take
\begin{align*}
  P_j
= (I - Q_\q) R_j (I - Q_\q)
= \parens[\big]{ 1 - (1 - \eps)^{N-j} } (I - Q_\q).
\end{align*}
% Indeed, the side condition for the first spec holds with equality:
% \begin{align*}
%   R_j
% = \parens[\big]{ 1 - (1 - \eps)^{N-j} } (I - Q_\q) + Q_\q
% = Q_\q R_{j+1} Q_\q + (I - Q_\q) P_j (I - Q_\q)
% \end{align*}
% and for the second specification this is also true:
% \begin{align*}
%   I
% = Q_\q I Q_\q + (I - Q_\q)(I - Q_\q)(I - Q_\q)
% \end{align*}
We can obtain the desired program by splitting up this specification in the following way:
\begin{align*}
  S_\text{then}
\totref\ &\qpenv{\qphole{\eps(I - Q_\q), I - Q_\q, I - Q_\q}{Q_\q, I, I}} \ruletag{HP.split} \\
=\ &\qpenv{\qphole{\eps(I - Q_\q), I - Q_\q}{Q_\q, I}}
\end{align*}
Indeed, this refinement can be applied since
\begin{align*}
  P_j
&= \parens[\big]{ 1 - (1 - \eps)^{N-j} } (I - Q_\q) \\
&= \parens[\big]{ 1 - (1 - \eps)^{N-j-1} (1 - \eps) } (I - Q_\q) \\
&= (1 - \eps)^{N-j-1} \eps (I - Q_\q) + \parens[\big]{ 1 - (1 - \eps)^{N-j-1} } (I - Q_\q),
\intertext{while}
  R_{j+1}
&= (1 - \eps)^{N-j-1} Q_\q + \parens[\big]{ 1 - (1 - \eps)^{N-j-1} } I.
\end{align*}
so \ruleref{HP.split} can be applied with probabilities~$(1 - \eps)^{N-j-1}$ and~$1-(1 - \eps)^{N-j-1}$.
\end{proofcase}

\begin{proofcase}\ruleref{H.boostWhile}
We can apply the rule \ruleref{HT.while} with the binary measurement~$B=I-Q$ and the sequence~$\{R_n\}_{n\in\N}$ defined by
\[
  R_n = \parens*{1 - \parens*{ 1 - \eps }^n} \parens*{I - Q_\q}
\]
for all~$n\in\N$, %with the convention that~$(1-\eps)^0 = 1$ even if~$\eps=1$,
to obtain
\begin{align*}
  \qpenv{\qphole{I}{Q_\q}}
  \totref
  \qpenv{
  \begin{aligned}
      &\qpwhile{I - Q}{\q}{\\
      &\quad \underbrace{\qphole{R_{n+1}}{Q_\q + R_n}}_{S_\text{body}} \\
      &}
  \end{aligned}
  }.
  \ruletag{HT.while}
\end{align*}
To see that the application of \ruleref{HT.while} is valid, we observe that~$R_0 = 0$, $R_n \imp R_{n+1}$ for all~$n\in\N$, and the limit~$R = \lim_{n\to\infty} R_n = I - Q_\q$ satisfies~$I \imp \cB_{0,\q}(Q_\q) + \cB_{1,\q}(R) = Q_\q + (I - Q_\q) = I$, where we used that~$Q_\q$ and~$I - Q_\q$ are projections.
To see that it yields the desired specification on the loop body, note that~$\cB_{0,\q}(Q_\q) + \cB_{1,\q}(R_n) = Q_\q + R_n$.

Now observe that we can rewrite the pre- and postcondition of the loop body in the following way:
\begin{align*}
  R_{n+1}
= \parens*{1 - \parens*{ 1 - \eps }^{n+1}} \parens*{I - Q_\q}
&= \parens*{ 1 - \eps }^n \eps \parens*{I - Q_\q}
 + \parens*{1 - \parens*{ 1 - \eps }^n } \parens*{I - Q_\q},
\\
  Q_\q + R_n
= Q_\q + \parens*{1 - \parens*{ 1 - \eps }^n} \parens*{I - Q_\q}
&= \parens*{ 1 - \eps }^n Q_\q + \parens*{1 - \parens*{ 1 - \eps }^n} I.
\end{align*}
Thus we see that we can refine the loop body using
\ruleref{HT.split}, with weights $\{ \parens*{ 1 - \eps }^n, \parens*{1 - \parens*{ 1 - \eps }^n} \}$, to obtain the desired result:
\begin{align*}
  S_\text{body} \totref \qpenv{\qphole{\eps(I - Q_\q), I - Q_\q}{Q_\q, I}}
  \ruletag{HT.split}
\end{align*}
This concludes the proof of the theorem.

\end{proofcase}

\end{proof}

%==============================================================================
\section{Example: Quantum Coin Toss until Zero}
\label{app:coin-toss-till-zero}
%==============================================================================
In this section we continue the discussion of our running example from \cref{sec:prelims}.
To keep the discussion self-contained, we first restate the program as well as its semantics and specification, as given in \cref{new:ex:coin-toss:program,new:ex:coin-toss:semantics,new:ex:coin-toss:spec}.
Then we show how QbC can be used to construct programs that meet this specification.

\paragraph{Program}
Consider the following algorithm:
Initialize a qubit in the $\ket0$ state.
Repeatedly apply the Hadamard gate and measure in the standard basis until the outcome~``0'' is seen.
This can be realized by the following program in the quantum while language~(\cref{def:qwhile-syntax}):
\begin{equation}\label{ex:coin-toss:program}
  \ctprog =
  \qpenv{
  \begin{aligned}
    &\qpinitS{q} ;  \\
    &\qpunitary{q}{H} ; \\
    &\qpwhileB{q}{\\
      &\quad\qpunitary{q}{H}\\&
    }
  \end{aligned}
  }
\end{equation}

%------------------------------------------------------------------------------
\paragraph{Semantics}
We can use \cref{def:semantics} to compute the semantics for the above program~(\cref{ex:coin-toss:program}).
For any state $\rho \in \dH$,
\begin{align*}
  \Bracks{{\ctprog}}(\rho)
  &=\Bracks{\qpenv{
      \qpinitS{q} ;
      \qpunitary{q}{H} ;
      \qpwhileB{q}{\qpunitary{q}{H}}
  }}(\rho) \\
  &= \Bracks{\qpenv{\qpwhileB{q}{\qpunitary{q}{H}}}} \parens[\big]{\Bracks{\qpenv{\qpunitary{q}{H}}} (\Bracks{\qpenv{\qpinitS{q}}}(\rho))} \\
  &= \Bracks{\qpenv{\qpwhileB{q}{\qpunitary{q}{H}}}} \parens[\big]{\Bracks{\qpenv{\qpunitary{q}{H}}} (\proj0)} \\
  &= \Bracks{\qpenv{\qpwhileB{q}{\qpunitary{q}{H}}}} \parens*{ \proj{+} }.
\end{align*}
The semantics of the loop is for a general state~$\sigma \in \dHle$ given by
\begin{align*}
  \Bracks{\qpenv{\qpwhileB{q}{\qpunitary{q}{H}}}}(\sigma)
  &= \sum_{k = 0}^{\infty} \parens*{\cB_0 \circ (\Bracks{\qpenv{\qpunitary{q}{H}}} \circ \cB_1)^k}(\sigma),
\end{align*}
where $\cB_j(\rho) = \proj{j}\!\rho\!\proj{j}$.
Now, for any state~$\sigma \in \dHle$, we have
\[ (\Bracks{\qpenv{\qpunitary{q}{H}}} \circ \cB_1)(\sigma) = \braket{1|\sigma|1} \proj-, \]
and therefore for any~$k \ge 1$,
$(\Bracks{\qpenv{\qpunitary{q}{H}}} \circ \cB_1)^k(\proj+) = \frac1{2^k} \proj-$.
Altogether, we find that
\begin{align*}
  \Bracks{\qpenv{\qpwhileB{q}{\qpunitary{q}{H}}}} \parens*{ \proj{+} }
  &= \cB_0(\proj+) + \sum_{k = 1}^\infty \cB_0\parens*{\frac1{2^k}\proj-} \\
  &= \frac12\proj0 + \sum_{k = 1}^\infty \frac1{2^{k+1}}\proj0 \\
  &= \proj0.
\end{align*}
Therefore the semantics of the coin toss until zero program~\eqref{ex:coin-toss:program} is given by
\begin{equation} \label{ex:coin-toss:semantics}
  \Bracks{\ctprog(\rho)} = \proj0
\end{equation}
for any initial state $\rho \in \dH$.
We see that no matter what state we start in, the program always terminates in the pure state~$\ket0$.

%------------------------------------------------------------------------------
\paragraph{Hoare Logic Specification}
We now discuss a natural quantum Hoare triple~(\cref{subsec:quantum-hoare-logic}) and its correctness for our example.
One way to specify the behavior of the program~$\ctprog$ is by the Hoare triple
\begin{equation}\label{ex:coin-toss:spec}
  \hoare{I}{\ctprog}{\proj0}.
\end{equation}
As discussed, this states that the program terminates in the final state~$\ket0$.
We can verify explicitly that this Hoare triple program is totally correct.
Indeed, we saw in \cref{ex:coin-toss:semantics} that $\Bracks{\ctprog}(\rho) = \proj0$ for every state~$\rho\in\dH$, and hence
\begin{align*}
  \satisfies{\Bracks{\ctprog}(\rho)}{\proj0}
= \satisfies{\proj0}{\proj0}
% = \tr \proj0 \proj0
= 1
= \tr\rho
= \satisfies{\rho}{I}
\end{align*}
for every state~$\rho\in\dH$.
This confirms the triple is totally correct.

%------------------------------------------------------------------------------
\paragraph{Applying Refinement Rules}
We now illustrate how to apply refinement rules~(\cref{def:totref}).
Consider the following abstract program, which consists of a single hole, with pre- and postcondition as in the Hoare triple in \cref{ex:coin-toss:spec} for the quantum coin toss until zero program:
\begin{equation}\label{eq:orig toss until zero}
  \Spec_\text{toss-until-zero} = \qpenv{\qphole{I}{\proj0}}
\end{equation}
One possible refinement is:
\begin{align*}
  \qpenv{\qphole{I}{\proj0}} \totref \qpenv{\qpinitS{q}} \ruletag{H.init}
\end{align*}
The above refinement rule can be applied as the condition~$I \imp \sum_{x \in \{0,1\}} \ketbra{x}{0} \proj0 \ketbra{0}{x}$ is satisfied, because the RHS equals $I$.

However, this is not the only program one can construct satisfying it.
Instead, we can also use the sequence rule \ruleref{H.seq}:
\begin{align*}
  \qpenv{\qphole{I}{\proj0}}
  \totref
  \qpenv{\qphole{I}{I} ; \qphole{I}{\proj0}}
  \ruletag{H.seq}
\end{align*}
Now the program at hand is a sequence of two holes.
We can now refine the second hole as follows:
\[
  \qpenv{\qphole{I}{I} ; \qphole{I}{\proj0}}
  \totref
  \qpenv{\qphole{I}{I} ; \qpinitS{q}}
  \ruletag{H.init}
\]
We can achieve this by using \ruleref{C.seqR}, as we already know that $\qpenv{\qphole{I}{\proj0}} \totref \qpenv{\qpinitS{q}}$.
One can now continue and fill in the remaining hole arbitrarily.

%------------------------------------------------------------------------------
\paragraph{Structured Specification}
On its own, \cref{eq:orig toss until zero} is not a very interesting specification, as much simpler programs than our original program~\eqref{ex:coin-toss:program} also satisfy it, e.g.,\ $\qpenv{\qpinit{q}}$.
However, we can refine the initial specification to arrive at a more structured one that naturally leads to our original program:
\begin{align*}
  \Spec_\text{toss-until-zero}
  &= \qpenv{\qphole{I}{\proj0}} \\
  &\totref
  \qpenv{
    {\qphole{I}{\proj+}} ;
    {\qphole{\proj+}{\proj0}}
  }
  \ruletag{H.seq} \\
  &\totref
  \qpenv{
    \begin{aligned}
    &\qphole{I}{\proj+} ; \\
    &\qpwhileB{q}{ \\
      &\quad\qphole
        {I - \frac1{2^n} \proj-}
        {I - \frac1{2^n} \proj1}
    \\&}
    \end{aligned}
  }
  \ruletag{HT.while} \\
  &\totref
  \qpenv{
    \begin{aligned}
    &\qphole{I}{\proj+} ; \\
    &\qpwhileB{q}{ \\
      &\quad\qphole
        {\proj+, \proj-}
        {\proj0, \proj1}
    \\&}
    \end{aligned}
  }
  \ruletag{HT.split}
\end{align*}
In the second step, we applied the rule \ruleref{HT.while} with the weakly increasing sequence defined by~$R_0 = 0$ and~$R_{n+1} = I - \frac1{2^n} \proj-$% for~$n\in\N$,
, with limit~$R = \lim_{n\to\infty} R_n = I$.
This is allowed, since~$\proj+ \imp \cB_0(\proj0) + \cB_1(I) = I$, and it gives rise to desired loop body since
$\cB_0(\proj0) + \cB_1(R_n) = I - \frac1{2^n} \proj1$ for all~$n\in\N$;
$\cB$ denotes the standard basis measurement of qubit~$q$. % (i.e., $\cB_j(\proj j) = \proj j$).
In the third step, we applied \ruleref{HT.split} with weights $\{1, 1 - \frac1{2^n}\}$.
Thus we have arrived at the following more structured specification, which refines \cref{eq:orig toss until zero}:
\begin{align*}
  \Spec'_\text{toss-until-zero} = \qpenv{
    \begin{aligned}
    &\underbrace{\qphole{I}{\proj+}}_{S_\text{init}} ; \\
    &\qpwhileB{q}{ \\
      &\quad \underbrace{\qphole{\proj+, \proj-}{\proj0, \proj1}}_{S_\text{body}}
    \\&}
    \end{aligned}
  }
\end{align*}

\paragraph{Construction}
We now construct a program from this specification by refining the two remaining holes in $\Spec'_\text{toss-until-zero}$:
\begin{align*}
S_\text{init} =
% \qpenv{\qphole{I}{\proj+}}
&\totref \qpenv{\qphole{I}{\proj0} ; \qphole{\proj0}{\proj+}} \ruletag{H.seq} \\
&\totref \qpenv{\qpinitS{q}        ; \qphole{\proj0}{\proj+}} \ruletag{H.init} \\
&\totref \qpenv{\qpinitS{q}        ; \qpunitary{q}{H}       } \ruletag{H.unit}
\end{align*}
and
\begin{align*}
S_\text{body} =
% \qpenv{\qphole{\proj+, \proj-}{\proj0, \proj1}}
\totref \qpenv{\qpunitary{q}{H}} \ruletag{H.unit}
\end{align*}
It is easy to see that the conditions required for these refinements are satisfied.
Putting it all together, we have re-constructed the program~\eqref{ex:coin-toss:program} from a specification:
\begin{gather*}
  \Spec_\text{toss-until-zero} \totref^*
  \Spec'_\text{toss-until-zero}
  \totref^*
  \qpenv{
  \begin{aligned}
    &\qpinitS{q} ;  \\
    &\qpunitary{q}{H} ; \\
    &\qpwhileB{q}{\\
      &\quad\qpunitary{q}{H}\\&
    }
  \end{aligned}
  }
  = \ctprog
\end{gather*}

\end{document}